\documentclass[aps,pra,10pt,twocolumn,floatfix]{revtex4-2}
\usepackage[utf8]{inputenc}
\usepackage[T1]{fontenc}
\usepackage{amsmath,amsfonts,amssymb,mathrsfs,graphicx,xcolor}
\usepackage{enumerate,array,booktabs,theorem,physics}
\newtheorem{lemma}{Lemma}
\usepackage{yquant}
\usetikzlibrary{fit,quotes,snakes,arrows,shapes}
\useyquantlanguage{groups}
\usepackage{xurl}
\usepackage[hidelinks]{hyperref}
\hypersetup{pdftitle={Trade-offs and experimental feasibility of nonlocal polygamy with two-outcome Bell inequalities},
 pdfauthor={Josep Batle, Tomasz Rybotycki, Tomasz Bialecki, Piotr Gawron, Adam Bednorz}}

\newcommand{\ba}{\begin{eqnarray}}
\newcommand{\ea}{\end{eqnarray}}

\def\be#1\ee{\begin{equation}#1\end{equation}}

\def\mb{\begin{pmatrix}}
\def\me{\end{pmatrix}}

\begin{document}

\title{Trade-offs and experimental feasibility of nonlocal polygamy with two-outcome Bell inequalities}	

\author{Josep Batle$^{1}$}
\author{Tomasz Rybotycki$^{2,3,4}$}
	\author{Tomasz Bia{\l}ecki$^{5,6}$}
	\author{Piotr Gawron$^3$}
	\author{Adam Bednorz$^5$}
	\email{abednorz@fuw.edu.pl}
	\affiliation{$^1$CRISP -- Centre de Recerca Independent de sa Pobla, 07420 sa Pobla,
		Balearic Islands, Spain}
	\affiliation{$^2$Systems Research Institute, Polish Academy of Sciences, ul. Newelska 6,
		01-447 Warsaw, Poland}
	\affiliation{$^3$Nicolaus Copernicus Astronomical Center, Polish Academy of Sciences,
		ul. Bartycka 18, 00-716 Warsaw, Poland
	}
	\affiliation{$^4$Center of Excellence in Artificial Intelligence, AGH University,
		al. Mickiewicza 30, 30-059 Cracow, Poland
	}
	\affiliation{$^5$Faculty of Physics, University of Warsaw, ul. Pasteura 5, PL02-093
		Warsaw, Poland}
	\affiliation{$^6$Faculty of Physics and Applied Informatics, University of Lodz,
		ul. Pomorska 149/153, PL90-236 Lodz, Poland}
	
\begin{abstract}
Entanglement and Bell nonlocality have different sharing constraints across
subsystems of a multipartite quantum system. We examine simultaneous violations
of Bell inequalities with two-outcome observables and two or three measurement
settings per party. We identify configurations whose simultaneous violations can
be demonstrated on IBM Quantum hardware, and derive bounds and trade-off
relations for several polygamous configurations. The results connect experimental
capabilities of quantum computers with mathematical methods for studying the
limits of quantum correlations.
\end{abstract}
\maketitle

\section{Introduction}

Nonlocal properties  of quantum mechanics \cite{epr} are verified by violation of Bell and Clauser--Horne--Shimony--Holt (CHSH) inequalities
\cite{bell,chsh}, with many generalizations \cite{mermin,ardehali,beli,froiss,collins,sliwa}.
The nonclassical core of this feature is entanglement, which cannot be simply reduced to a classical correlated system.
The entangled states especially many party, are the fundamental building blocks of quantum computation and cryptography \cite{nielsen}.
Verification of entanglement is the basic diagnostics of properly functioning quantum computers, revealing
the limits of their capabilities.

However, the entanglement is fragile, and a many-party system may not necessarily exhibit entanglement on a subset 
of parties. The most prominent example is a two-party entanglement monogamy, i.e. a CHSH inequality cannot be simultaneously violated on
two-party subsets sharing one party \cite{toner}. The picture gets complicated when taking more parties and asking about various
subsets. In fact, many configurations are bounded by monogamy relations \cite{kurz,mironowicz,tran}, i.e. entanglement of subsets is not always possible.
Interestingly, when the number of parties grows, especially for more than 4, subsets start to break the monogamy \cite{pnas,munne}.
This is the consequence of exponentially growing the gap between the classical and quantum bound, which can be distributed linearly among
subsets.

The commonly considered case is the two-level party set tested by two-setting inequalities, with the same sets of settings of two-outcome measurements
 in all subsets.
By Jordan's lemma \cite{mironowicz,jorle}, the dimension restriction suffices in the case of linear tests. Taking 4 settings or 4 outcomes could be easily cheated by the quantum system
by using separate entanglement states for each pair of parties. Therefore we shall restrict here to two outcomes and maximally 3 settings.
The case of 2 settings is already interesting in the case of 3-element subsets of a 4-party set as the monogamy rules out standard Mermin inequalities 
\cite{mermin,kurz}.
Abandoning symmetry, we could simply entangle two separate pairs, i.e. a biseparable state and claim that each 3-subset contains at least one of the pairs and demonstrate violation
of a standard Bell inequality.  If we impose the symmetry condition, i.e. require the topologically equivalent inequalities, then simple 
cheating fails. It turns out that monogamy is violated not by Mermin but other 3-party inequalities based on \'Sliwa 46 classes \cite{sliwa} with one such example shown in \cite{pnas}.

In the present paper, we find another inequality of \'Sliwa, with a larger violation, and demonstrate its actual simultaneous violation on IBM Quantum.
However, we also find that randomly biseparable states can also violate it, albeit with smaller limits. This shows that
polygamic violation of Bell-type inequalities alone is insufficient to prove genuine multipartite entanglement.
The requirement of genuine multipartite entanglement is arguable, as the polygamy deliberately focuses on subsets and not the full set.
The traditional Mermin-Ardehali-Belinskii--Klyshko (MABK) inequality \cite{mermin,ardehali,beli} 
is jointly violated for $N\geq 5$ parties. We have also demonstrated the violation at $N=5$ on IBM Quantum \cite{ibm},
and discussed the bi- and tri- and even fourseparable limits, which turn out to be still beyond the classical bounds.

As a theoretical intermezzo, we develop explicit semidefinite programming
(SDP) certificates \cite{sdp,sdp2} for sums of squared Bell correlations.
The construction is the dual form of the Lov\'asz-$\vartheta$
uncertainty bound \cite{lovasz1979,degois2023}, extending the
anticommuting-clique method used in Bell monogamy
\cite{kurz,mironowicz,tran}. We exploit the structure of the network
examples to obtain exact certificates and uniform trade-off bounds.
For cyclic configurations, the sharper graph parameter $\beta$ and the
odd-cycle uncertainty relation \cite{xu2024beta} supply an
algebraic proof. 

In the final part we discuss violations with 3 settings, starting with the three-party case, when one of the settings is shared in the CHSH correlators.
Here the quantum optimum equals an optimum attainable by biseparable strategies: one party is classical and the other two are maximally entangled, giving two classical maxima and one quantum maximum,
on average above classical limit. We also confirmed it on IBM Quantum, creating the randomization within the qubit space itself.
We do not forget  the two-party $I_{3322}$ inequality \cite{froiss,collins,sliwa}, whose three cyclic copies, with the common sign pattern considered below, cannot all be violated simultaneously.
However, two of the inequalities can be violated, with the violation amount and separability depending highly on the relative signs.
One can even get a hybrid violation of $I_{3322}$ and CHSH inequalities, although reaching only the biseparable limit.

A large part of the work is supported by numerical/symbolic codes, which we present on the open repository \cite{zen}.

\section{Bell framework of monogamy and polygamy testing by inequalities}

Let us recall standard CHSH inequality for two parties $A$, and $B$, with two measurement settings $1$, $2$ \cite{chsh},
\begin{equation}
\mathcal S= A_1B_1+A_1B_2+A_2B_1-A_2B_2\leq 2
\end{equation}
for $A_i,B_j=\pm 1$. The above classical bound is a result of a simple case-by-case verification.
In quantum mechanics, it is violated  in the average sense, $\bar{\mathcal{S}}=\langle\mathcal S\rangle$,
with the averages defined $\langle O\rangle=\mathrm{Tr}O\rho$ for the state $\rho\succeq 0$, $\rho=\rho^\dag$, $\mathrm{Tr}\rho=1$,
a semidefinite Hermitian normalized matrix, and $O=O^\dag$ being the Hermitian observable. In a pure state $\rho=|\psi\rangle\langle\psi|$ with $\langle \psi|=(|\psi\rangle)^\dag$, normalization $\langle \psi|\psi\rangle=1$, and
$\langle O\rangle=\langle\psi|O|\psi\rangle$.
In the two-outcome scheme,  $O=O(+)-O(-)$ where $O(\pm)\succeq 0 $ are semidefinite operators, generating probabilities $p(\pm)=\mathrm{Tr}O(\pm)\rho$ with $O(+)+O(-)=I$ (identity) \cite{nielsen}.
The crucial point is the tensor structure $AB\equiv A\otimes B$, i.e. the spaces of $A$ and $B$ are multiplied. 
For linear inequalities the maxima lie on the boundaries, so we can assume $O^2=I$, which is the case we consider throughout the text.
When clear from the context, we shall use $A$, $B$, or $a$, $b$, etc., interchangeably to denote parties, observables, or basis states.

Now in quantum mechanics in the Bell state
$
\sqrt{2}|\psi\rangle=|01\rangle-|10\rangle
$,
with $A_1=X$, $A_2=Z$ and $\sqrt{2}B_1=-X-Z$, $\sqrt{2}B_2=Z-X$, and Pauli operators
$X=|0\rangle\langle 1|+|1\rangle\langle 0|$ and $Z=|0\rangle\langle 0|-|1\rangle\langle 1|$.
We denote $|ab\rangle\equiv |a\rangle\otimes |b\rangle$, the tensor basis in $AB$ space.
then 
\begin{equation}
\langle A_1B_1\rangle=\langle A_1B_2\rangle=\langle A_2B_1\rangle=-\langle A_2B_2\rangle=1/\sqrt{2}
\end{equation}
giving the combination $\bar{\mathcal S}=2\sqrt{2}>2$.
The quantum Tsirelson bound $2\sqrt{2}$ \cite{tsirel} is strictly below the algebraic maximum $4$; generic Bell-type inequalities have bounds
to be determined by Navascu\'es--Pironio--Ac\'in hierarchy (NPA) \cite{npa} with help of SDP \cite{sdp,sdp2}.
 It turns out that the inequality cannot be jointly violated by pairs $(A,B)$ and $(A,C)$
 sharing the same observables $(A_1,A_2)$, with possible exchange and/or sign changes

In the 3-party $ABC$ scenario there is a trade-off of the violations, say
\begin{equation}
\mathcal S_{AB}=\sum_{ij}\epsilon_{ij}\langle A_iB_j\rangle,\;
\mathcal S_{AC}=\sum_{ij}\epsilon'_{ij}\langle A_iC_j\rangle,
\end{equation}
The sign arrays $\epsilon_{ij}=\pm1$ and $\epsilon'_{ij}=\pm1$ may differ, but each must satisfy $\prod_{i,j}\epsilon_{ij}=\prod_{i,j}\epsilon'_{ij}=-1$ to define a CHSH inequality.

We have a general monogamy relation \cite{toner}
\begin{equation}
\bar{\mathcal S}_{AB}^2+\bar{\mathcal S}_{AC}^2\leq 8 \label{mon2}
\end{equation}
for arbitrary dichotomic observables and states.

For 3-party setup we have Mermin inequality \cite{mermin}
\begin{equation}
\mathcal M=A_1B_1C_2+A_1B_2C_1+A_2B_1C_1- A_2B_2C_2\leq 2
\end{equation}
with the same assumptions on observables as in Bell case.
It is violated, with $\bar{\mathcal M}=4$ e.g. taking $A_1=B_1=C_1=Y$ and $A_2=B_2=C_2=X$, for $Y=i|1\rangle\langle 0|-i|0\rangle\langle 1|$
and the $GHZ$ state \cite{ghz}
\begin{equation}
\sqrt{2}|\psi\rangle=|000\rangle-|111\rangle
\end{equation}
Again, it cannot be jointly violated in a 4-party setup $ABCD$ for any $3$-subset due to monogamy relation \cite{kurz}
\begin{equation}
\bar{\mathcal M}_{ABC}^2+\bar{\mathcal M}_{BCD}^2+\bar{\mathcal M}_{CDA}^2+\bar{\mathcal M}_{DAB}^2\leq 16\label{mon3}
\end{equation}
even if we assign random signs, e.g.
replace $A_i\to \pm A'_{\sigma(i)}$ for any permutation $\sigma$, with signs and permutations independent for every party, setting, and subset.
In fact every inequality of the form
\begin{equation}
\sum_i\langle F_i\rangle^2\leq \Lambda
\end{equation}
can, when each party has only two dichotomic settings, be reduced to real qubits and pure states.
The proof, based on Jordan's lemma and convexity \cite{mironowicz, jorle}, is presented in Appendix \ref{appjor}.
The traditional proofs of (\ref{mon2}) and (\ref{mon3}) by grouping anticommuting sets \cite{kurz} are recalled in Appendix \ref{appmon}.

\section{Symmetrized polygamy in the 4-party set}

We cannot violate Mermin inequality equally in 3-subsets of 4 parties, because of the monogamy relation.
However, there are different symmetrized inequalities in the 3-party, 2-settings, 2-outcomes scenario, that can be simultaneously violated.
There are 46 classes defined by \'Sliwa \cite{sliwa}, $\mathcal W^n$ for $n=1,\dots 46$.

In \cite{pnas}
class 39 was used:
\begin{align}
&\mathcal W^{39}=2(A_1+B_1+C_1)-A_1B_1-B_1C_1-C_1A_1\nonumber\\
&+A_1B_2+A_2B_1+B_1C_2+B_2C_1+C_1A_2+C_2A_1\nonumber\\
&+A_2B_2+B_2C_2+C_2A_2+2A_1B_1C_1+A_2B_2C_2\nonumber\\
&-A_2B_1C_1-A_1B_2C_1-A_1B_1C_2\nonumber\\
&-2(A_1B_2C_2+A_2B_1C_2+A_2B_2C_1)
\end{align}
with classical bound $\leq 6$.
It is convenient to work with Dicke states
\begin{equation}
{\binom{n}{k}}^{1/2}|D^k_n\rangle=\sum_{h(x)=k}|x\rangle\label{dicc}
\end{equation}
where $h(x)$ is the number of bits $1$ in the bitstring $x$.
In other words, the $|D^k_n\rangle$ distributes uniformly $k$ bits $1$.
Optimizing over identical settings
\begin{align}
&A_1=B_1=C_1=D_1=Z,\nonumber\\
&A_2=B_2=C_2=D_2=Z\cos\phi+X\sin\phi
\end{align}
in the space of symmetric Dicke states $|D^k_4\rangle$,
we obtain the numerical value
\begin{equation}
\bar{\mathcal W}^{39}_{ABC}\simeq 6.19989>6
\end{equation}
which is a bit higher than in \cite{pnas}. Due to the symmetry of $|D_k\rangle$ the violation is identical for all subsets of parties
$ABC$,$BCD$, $CDA$, $DAB$.

It turns out that the optimal violation is obtained by class 5
\begin{align}
&\mathcal W^5=A_1+B_1+C_1 -A_2B_2-B_2C_2-C_2A_2\nonumber\\
& +A_1B_2+A_2B_1+B_1C_2+B_2C_1+A_1C_2+A_2C_1\nonumber\\
&-A_1B_1C_1+A_2B_2C_2-A_2B_1C_1-A_1B_2C_1-A_1B_1C_2
\end{align}
with the classical bound $\leq 3$, quantum-classical (biseparable) bound $4\sqrt{2}-1\simeq 4.65685424949$ with $C_i=-1$, and the individual full quantum bound $\leq 8\sqrt{5}-13\simeq 4.88854382$ \cite{max322}
The quantum numerical bound over $|D^k_4\rangle$ gives
\begin{align}
\bar{\mathcal W}^{5}_{ABC}\simeq 3.2204751396970845
\end{align}
The value is the $4$th real root of the equation $f(x)\equiv\sum_i f_i x^i=0$, see Table \ref{fs5} in Appendix \ref{appmon},
confirmed also by NPA.

Expressing
\begin{equation}
\cos\phi=(u-1)/(u+1),\;\sin\phi=2\sqrt{u}/(u+1)
\end{equation}
we obtain the angle 
$
u\simeq 0.4352296958477135
$
which is the $4$th real root of $h(x)=\sum_i h_ix^i=0$, see Table \ref{us5}
so that
$
 \phi\simeq 1.9752375970999096
$
 and the maximizing state in Dicke basis
\begin{equation}
 |\psi\rangle=\sum_k\psi_k|D_k\rangle
\end{equation}
 reads 
\begin{equation}
\begin{pmatrix}
\psi_0\\
\psi_1\\
\psi_2\\
\psi_3\\
\psi_4
\end{pmatrix}
\simeq
\begin{pmatrix}
0.455452\\
 0.720303\\
0.315565\\
-0.362381\\
0.206945
\end{pmatrix}
\end{equation}
However, we can eliminate the middle term $\psi_2$ by rotating the basis.
Recall a single qubit rotation
\begin{equation}
Y_{2\theta}=\begin{pmatrix}
\cos\theta &-\sin\theta\\
\sin\theta &\cos\theta\end{pmatrix}=\exp(-i\theta Y)
\end{equation}
The joint rotation $Y_{2\theta}^{\otimes4}$
reads in the Dicke basis
$
Y_{2\theta}^{\otimes4}=\exp(\theta R)
$
where
\begin{equation}
R=\begin{pmatrix}
0&-2&0&0&0\\
2&0&-\sqrt{6}&0&0\\
0&\sqrt{6}&0&-\sqrt{6}&0\\
0&0&\sqrt{6}&0&-2\\
0&0&0&2&0\end{pmatrix}
\end{equation}
To eliminate $\psi_2$ the middle row of the rotation, corresponding to $D_2$, should nullify $|\psi\rangle$ which happens for
\begin{equation}
\theta\simeq -0.24462876880560303
\end{equation}
and the state reads
\begin{equation}
\begin{pmatrix}
\psi_0\\
\psi_1\\
\psi_2\\
\psi_3\\
\psi_4
\end{pmatrix}
\simeq
\begin{pmatrix}
0.623925216287\\
0.63436367573\\
 0\\
 -0.347285917038\\
 0.296129267758
\end{pmatrix}
\end{equation}
Such a combination can be easily implemented on IBM Quantum \cite{ibm,qis}.

One starts by creating an entangled 4-qubit state
\begin{equation}
\sqrt{\psi_0^2+\psi_1^2}|0000\rangle+\sqrt{\psi_3^2+\psi_4^2}|1111\rangle
\end{equation}
using a $CZ$ gates and then rotating qubit 1 conditionally with respect to qubit 2 by $CZ/RZZ$ gates to
 \begin{equation}
 \psi_0|0000\rangle+\psi_1|0100\rangle+\psi_3|1011\rangle+\psi_4|1111\rangle
 \end{equation}
 and finally distributing the middle states to the neighbors by $\exp(i\pi(XX+YY)/8)$ rotations which keep
 $|00\rangle$ and $|11\rangle$ unaffected, but rotate $|01\rangle\to (|01\rangle+i|10\rangle)/\sqrt{2}$
 and $|10\rangle\to (|10\rangle+i|01\rangle)/\sqrt{2}$, see details in Appendix \ref{appdi}. First between qubits $1$ and $2$, then individually between $1$ and $0$
 and $2$ and $3$, using zero-based qubit labels $0,1,2,3$, with phase corrections when necessary.
 
  The actual test has been run on \emph{ibm\_kingston} with 10 jobs, 12 repetitions, and 10000 shots, with qubits $8,9,10,11$, as $ABCD$, see Table \ref{ww5}.
  
\begin{table}
\begin{tabular}{ccccc}
\toprule
&$A$&$B$&$C$&$D$\\
\midrule
$\bar{\mathcal W}^5$&3.0556&3.0384&3.0617&3.1000\\
$\Delta{\mathcal W}^5$&0.0014&0.0014&0.0014&0.0014\\
\bottomrule
\end{tabular}
\caption{Experimental values of $\bar{W}^5$ and shot-noise standard errors for parties $ABCD$ obtained on IBM Quantum with the specified party excluded}
\label{ww5}
\end{table}

 There are several negative results. 
 A biseparable state, with two parties, e.g. $AB$ separated from $CD$, we have the total bound
 \begin{equation}
 \bar{\mathcal W}^5_{ABC}+ \bar{\mathcal W}^5_{BCD}+ \bar{\mathcal W}^5_{CDA}+ \bar{\mathcal W}^5_{DAB}\leq 12\label{wbis}
 \end{equation}
 which means that at least one of inequalities cannot be violated.
 As earlier all observables can be reduced to combinations of $X$ and $Z$ and so real.
 For this real-valued objective, it suffices to optimize over convex mixtures of products of real local states on $AB$ and $CD$. Such states satisfy the partial-transpose identity
$\langle PQ\rangle=\langle P^TQ\rangle$
 if $P$ is the operator polynomial in $AB$ space and $P^T$ is its transpose while $Q$ is a polynomial operator in $CD$ space. 
It is a simple consequence of products $\langle P\rangle=\langle P^T\rangle$ and $\langle PQ\rangle=\langle P\rangle\langle Q\rangle$ for product. Using the Partial Transpose equality $P^TQ\equiv PQ$ we can write the difference of the right-hand side and the left-hand side of (\ref{wbis}) as a sum of squares, see (\ref{sq5}) in Appendix \ref{appw5}.

Interestingly, separating a single party, e.g. $D$ from $ABC$ turns $D$ into a classical set of variables $D_{1,2}=\pm 1$
but the total sum
\begin{align}
& \bar{\mathcal W}^5_{ABC}+ \bar{\mathcal W}^5_{BCD}+ \bar{\mathcal W}^5_{CDA}+ \bar{\mathcal W}^5_{DAB}\nonumber\\
&\leq W_1\simeq 12.41985481086368
\end{align}
 where the bound $W_1$ is the second largest real root of $\bar{f}(x)\equiv\sum_i \bar{f}_ix^i=0$, see Table \ref{bs5a},
 and requires $D_1=1=-D_2$ and the other observables equal $A_1=B_1=C_1=Z$, $A_2=B_2=C_2=Z\cos\phi'+X\sin\phi'$ with $\phi'\simeq 2.08553$.
 It shows that the violations at the level $3.10496370272$ do not need genuine multipartite entanglement, only among random 3 parties.
 
 In fact dropping symmetry with 4 parties, we can always maximize $\bar{\mathcal S}_{AB}$ and $\bar{\mathcal S}_{CD}$ with separate entangled states.
 Then every subset $ABC$, $BCD$, $CDA$, $DAB$ contains one of them and so detects nonlocality. Therefore the polygamy challenge requires fair goal
 to capture shared multipartite entanglement.

 \section{Polygamy with MABK inequality with 5 parties}
 It is known that in the case of 5 parties, $ABCDE$, one can violate simultaneously MABK inequality with 4 parties \cite{mermin,ardehali,beli}
i.e.
\begin{equation}
\mathcal M_{ABCD}=\mathrm{Re}(1-i)ABCD\leq 4
\end{equation}
for $A=A_1+iA_2$, etc. or explicitly
\begin{align}
&A_1B_1C_1D_1+A_2B_2C_2D_2+A_2B_1C_1D_1+A_1B_2C_1D_1\nonumber\\
&+A_1B_1C_2D_1+A_1B_1C_1D_2-A_2B_2C_1D_1-A_1B_2C_2D_1\nonumber\\
&-A_1B_1C_2D_2-A_2B_1C_1D_2-A_2B_1C_2D_1-A_1B_2C_1D_2\nonumber\\
&-A_1B_2C_2D_2-A_2B_1C_2D_2-A_2B_2C_1D_2-A_2B_2C_2D_1
\end{align}
The inequality is flexible with respect to sign flips and exchanging settings.
The quantum bound $\bar{\mathcal M}_{ABCD}=8\sqrt{2}$ is obtained e.g. for $A_1=B_1=C_1=D_1=Y$, $A_2=B_2=C_2=D_2=X$
and the state
\begin{equation}
\sqrt{2}|\psi\rangle=|0000\rangle+e^{-i\pi/4}|1111\rangle
\end{equation}

Making one party classical, e.g. setting $D_1=D_2=1$ the inequality collapses to $2\mathcal M_{ABC}$ \cite{cere,werner}.
Interestingly for the $AB|CD$ biseparability, the quantum bound reduces to $4\sqrt{2}$ \cite{collinsm4} which is saturated either by 
a single entangled state $AB$ and classical $CD$ or a pair of entangled states at $AB$ and $CD$, see Appendix \ref{appms} and the
auxiliary Uffink inequality \cite{uff} in Appendix \ref{appuff}.

In total $ABCDE$ case we can consider
\begin{align}
&\mathcal M=\mathcal M_{ABCD}+\mathcal M_{BCDE}+\nonumber\\
&\mathcal M_{CDEA}+\mathcal M_{DEAB}+\mathcal M_{EABC}
\end{align}
For a genuinely multipartite entangled state the sum can reach $\bar{\mathcal M}=\pm 8\sqrt{10}\simeq \pm 25.298221$ or 
$8\sqrt{2/5}\simeq  5.05964425627$ per the subset of parties \cite{pnas}
in the state
\begin{align}
&\sqrt{10}|\psi\rangle=\sqrt{5}|11111\rangle\pm e^{-i\pi/4}\times\nonumber\\
&(|10000\rangle+|01000\rangle+|00100\rangle+|00010\rangle+|00001\rangle)
\end{align}
with $A_1=B_1=C_1=D_1=E_1=X$ and $A_2=B_2=C_2=D_2=E_2=Y$.

We tested this on \emph{ibm\_kingston}
$ABCDE$ represented by qubits $122,123,136,143,144$
with 10 jobs 6 repetitions and 10000 shots, see Table \ref{mm4}.

\begin{table}
\begin{tabular}{cccccc}
\toprule
&$A$&$B$&$C$&$D$&$E$\\
\midrule
$\bar{\mathcal M}_4$&4.11592&4.15381&4.17333&4.17890&4.19899\\
$\Delta{\mathcal M}_4$&0.00353&0.00353&0.00352&0.00352&0.00352\\
\bottomrule
\end{tabular}
\caption{Experimental values of $\bar{M}_4$ and shot-noise standard errors for parties $ABCDE$ with the specified party excluded} 
\label{mm4}
\end{table}

For the biseparability $ABCD|E$ with $E_1=E_2=1$, $A_1=B_1=C_1=D_1=Z$
and $A_2=B_2=C_2=D_2=((u-1)Z+2\sqrt{u}X)/(u+1)$
reaching $\bar{\mathcal M}$ value
$22.726730684148507$
 being the second largest real root of $w(x)\equiv\sum_i w_ix^i=0$ with coefficients in the Table \ref{m5a} in Appendix \ref{appm5},
while $u\simeq 1.233235$ is the fourth largest root of $v(x)=\sum_i v_i x^i=0$ in Table \ref{mu5a}.
The lowest bound in $ABCD|E$ is not symmetric, but minimally larger absolute values
$-22.76075931575539$ being the fourth smallest root of $\bar{w}(x)\equiv\sum_i \bar{w}_ix^i=0$ with coefficients in the Table \ref{m5b}
while $u\simeq 1.333520346238141$ is the third largest real root of $\bar{v}(x)=\sum_i \bar{v}_i x^i=0$ in Table \ref{mu5b}.

In the biseparability $ABC|DE$, the sum of 5 values can reach $20.649110640673517328$, also for $ABC|D|E$,see Appendix \ref{appm5}
However, one can get larger violation by the lower bound
$-21.0787936983$, see Appendix \ref{appm5}.
The triseparable $AB|CD|E$ and the fourseparable  $AB|C|D|E$ case gives the upper bound only classical $20$
but surprisingly the lower bound
is $-20.2624403701$  for the fourseparable state with only $AB$ entangled and $C_i=D_i=E_i=1$ for $i=1,2$
see Appendix \ref{appm5} 

In the complete classical case $A_i=B_i=C_i=D_i=E_i=1$ we get $-20$ ($-4$ per party).
The positive result $+20$ is obtained e.g. if $D_2=E_2=-1$.

  \section{Monogamies and trade-off with MABK inequalities}




The normalized MABK inequality reads
\begin{equation}
\mathcal A_n=\mathrm{Re}\;e^{i\pi (n-1)/4}A^{(1)}\cdots A^{(n)}/2^{(n-1)/2}
\end{equation}
Up to outcome and setting relabellings, this gives $\mathcal A_2=\mathcal S_{AB}/2$, $\mathcal A_3=\mathcal M_{ABC}/2$, $\mathcal A_4=\mathcal M_{ABCD}/4$, identifying $A^{(1)}=A$,
$A^{(2)}=B$, $A^{(3)}=C$, $A^{(4)}=D$.
Classically $A=\pm 1\pm i=\sqrt{2}e^{ik\pi/4}$ where $k=1,-1,3,-3$, i.e. $k$ is odd so $A^{(1)}\cdots A^{(n)}=2^{n/2}e^{i\alpha\pi/4}$ where
$\alpha$ has the same parity as $n$ so the classical bound is $|\bar{\mathcal A}_n|\leq 1$ (the extra $\sqrt{2}$ stems from real part of $\pm e^{\pm i\pi/4}$). 
On the other hand $A^{(k)}_1=X$, $A^{(k)}_2=Y$  on the  $GHZ$ state \cite{ghz}
$(|0^n\rangle+e^{i(1-n)\pi/4}|1^n\rangle)/\sqrt{2}$ gives $\bar{\mathcal A}_n=2^{(n-1)/2}$.
From Jordan's lemma, we can assume $A^{(j)}_{1,2}=Z^{(j)}\cos\phi_j\pm X^{(j)}\sin\phi_j$ and  then \cite{zuk}
\begin{align}
&|\bar{\mathcal A}_n|^2\nonumber\\
&=2\left|\mathrm{Re}\;e^{i(2n-1)\pi/4}\left\langle \prod_j
 (Z^{(j)}\cos\phi_j-iX^{(j)}\sin\phi_j)\right\rangle\right|^2\nonumber\\
&\leq \sum_i \langle P^{(1)}_{i_1}\cdots P^{(n)}_{i_n}\rangle^2\equiv (12\cdots n)\label{sumsq}
\end{align}
for  $i_k=0,1$ and $P_0=X$, $P_1=Z$ by Cauchy-Schwarz inequality applied to the two vectors of $2^n$ elements: averages $\langle P^{(1)}_{i_1}\cdots P^{(n)}_{i_n}\rangle$, and products $\pm\cos\phi_j$ and $\pm\sin\phi_j$ such that
exactly one of the two is chosen for each $j$. The norm of the latter vector is $1$ by $\cos^2+\sin^2=1$ applied to each angle. Note that $\mathrm{Re}\; e^{i(2n-1)\pi/4}(-i)^k=\pm 1/\sqrt{2}$ cancels the front factor.
The right-hand side of (\ref{sumsq}) allows to derive a variety of monogamy and trade-off relations.

We shall denote \emph{Pauli string} any tensor operator of the form
\begin{equation}
F_i=P^{(1)}_{i_1}\cdots P^{(n)}_{i_n}
\end{equation}
where $i_k=\bullet,0,1$ and $P_{\bullet}=I$ (inactive), $P_0=X$, $P_1=Z$ (active). For example a Pauli string $i=01\bullet\bullet 001010\bullet 110\bullet$ reads
\begin{equation}
XZIIXXZXZXIZZXI
\end{equation}

The traditional method of finding upper bounds on $\sum_i \langle F_i\rangle^2$ relied on grouping Pauli strings into anticommuting cliques \cite{kurz,mironowicz, tran}.
Note that
\begin{equation}
F_iF_j=(-1)^{h(i,j)} F_jF_i \label{annt}
\end{equation}
where $h(i,j)$ is the number of different bits (Hamming distance \cite{hamm}) on non-$\bullet$ (active) positions.
The scope can  be generalized to the whole set of Pauli matrices $P\in\{I,X,Y,Z\}$.
The sign is $+$ (commuting case) if it is even and $-$ if odd.
Since 
\begin{equation}
\sum_{i\in C}\langle F_i\rangle ^2\leq 1
\end{equation}
for the anticommuting clique $C$, identifying clique cover of the given set of Pauli strings, one find the bound.

We now use a semidefinite strengthening of clique grouping. 
The Pauli-correlation method connects the earlier monogamy constructions
\cite{kurz,mironowicz,tran} with graph uncertainty bounds and their
state-polynomial refinements \cite{degois2023,bermejo2024}.
Suppose we want to find an upper bound on the sum
\begin{equation}
\sum_j \langle F_j\rangle^2\leq \Lambda\label{labe}
\end{equation}
where $\langle F\rangle\equiv \langle\psi|F|\psi\rangle$, for real normalized state $|\psi\rangle$ ($\langle\psi|\psi\rangle=1$).
Remember that from convexity and Jordan's lemma  (Appendix \ref{appjor}), we can restrict to real qubit space.
More generally, let $\mathcal G$ be the anticommutation graph whose vertices are the Pauli
strings $F_i$ as vertices and edges/no-edges for anticommuting/commuting pairs, i.e. $h(i,j)\equiv 1/0$ mod 2 in (\ref{annt}). 
The bound (\ref{labe}) is then known as $\beta(\mathcal G)$ \cite{hast,degois2023,xu2024beta}.
The see-saw method \cite{sdp2} searches for feasible values, and therefore lower bounds on the global maximum, by maximizing eigenvalues of
\begin{equation}
H(x)=\sum_j x_j F_j
\end{equation}
at the constraint $\sum_j x_j^2=1$. Note that $\langle H(x)\rangle^2\leq \sum_j\langle F_j\rangle^2=S$
and, when $S>0$, equality holds at $x_j=\langle F_j\rangle/\sqrt{S}$.
The exact optimum satisfies
\[
 \Lambda_* =\max_\psi\sum_j\langle F_j\rangle_\psi^2
 =\left(\max_{\|x\|=1}\lambda_{\max}(H(x))\right)^2.
\]
The zero-correlation case is immediate. The see-saw algorithm is as follows. Start with some $x$.
The main loop reads:
\begin{enumerate}
\item Find maximal eigenvalue of $H$
\item Find the corresponding eigenvector and normalize it to $|\psi\rangle$
\item Evaluate $\bar{x}_j=\langle\psi|F_j|\psi\rangle$
\item Set new $x_j=\bar{x}_j/\sqrt{S}$, with $S=\sum_j \bar{x}_j^2$
\end{enumerate}
One stops when the desired accuracy is reached.
The problem is not convex, so the iteration does not certify the global maximum. Multiple starts can improve the feasible value but cannot certify an upper bound. In principle one can test standard zeros of the derivative vector but, having a lot of variables and
the exponentially growing space, this way is completely impractical especially for large systems.

We present a guaranteed approach. Due to the fact $\langle H^2\rangle-\langle H\rangle^2=\langle (H-\langle H\rangle)^2\rangle\geq 0$,
take 
\begin{align}
&\langle\psi|H(x)|\psi\rangle^2\leq\langle\psi| H^2(x)|\psi\rangle =\nonumber\\
&\sum_{ij} x_i\langle \psi|(F_iF_j+F_jF_i)|\psi\rangle x_j/2=\langle x|\langle\psi|G|\psi\rangle|x\rangle/2
\end{align}
where $G$ is tensor matrix indexed by both strings and qubit states 
with blocks $G_{ij}=F_iF_j+F_jF_i$, i.e. symmetrized in the qubit space. We remind that all strings have the property $F_iF_j=\pm F_jF_i$, i.e. they commute ($+$) or anticommute ($-$)
so $G_{ij}=0$ whenever $F_i$ and $F_j$ anticommute.

Now we can apply SDP method. Our goal is to find the upper bound on
\begin{equation}
\tfrac12\mathrm{Tr}\bigl[G(\rho_\psi\otimes\rho_x)\bigr]
\end{equation}
with $\rho_\psi=|\psi\rangle\langle\psi|$
and $\rho_x=|x\rangle\langle x|$.
The product condition $\rho=\rho_\psi\otimes\rho_x$ is nonlinear. For an SDP upper bound we enlarge the set of real product states to real density matrices $\rho$, with entries $\rho_{ai,bj}$ and
$a,b$ indexing the qubit space, vectors  $|\psi\rangle$ while $i,j$ indexing the string space, vectors $|x\rangle$.
The matrix is
\begin{itemize}
\item Hermitian, semidefinite $\rho\succeq 0$\\
\item normalized $\mathrm{Tr}\rho=1$\\
\item invariant under partial transpose ($T_x$): $\rho^{T_x}_{ai,bj}=\rho_{aj,bi}$ and $\rho=\rho^{T_x}$.
\end{itemize}
Every mixture of such product states is partial-transpose invariant. The converse is not assumed: this is a relaxation, not a characterization of separability. Nor does the PPT criterion \cite{ppt,ppth,upb} becomes an equality. Before this relaxation, the linear objective has the same maximum over product states and their convex mixtures
\begin{equation}
\rho=\sum_k p_k\rho_{k,\psi}\otimes\rho_{k,x}
\end{equation}
with $p_k\geq 0$, $\sum_k p_k=1$. Separately, convexity of the original objective in a physical state $\sigma=\sum_k p_k\sigma_k$ gives
\begin{align}
&\sum_j\langle F_j\rangle_\sigma^2=\sum_j\left(\sum_k p_k\langle F_j\rangle_k\right)^2=\nonumber\\
&\sum_j\left(\sum_k p_k\bar{F}_{jk}\right)^2\leq \sum_{jk}p_k\bar{F}_{jk}^2=\sum_k p_k\sum_j\bar{F}_{jk}^2
\end{align}
with $\bar{F}_{jk}=\langle F_j\rangle_{\sigma_k}$ and we used the nonnegativity of the variance.
Now the problem is convex and SDP methods apply, but the matrix $G$ is large, although the constraints are sparse.

Our goal is to retrieve Sum-Of-Squares certificate, i.e. express
\begin{equation}
2\Lambda I-G=V+V^{T_x}
\end{equation}
where $V$ is positive semidefinite (PSD). This identity gives a nonnegative expectation on real product states; it does not assert that $V+V^{T_x}$ is PSD on arbitrary entangled states.


Let us define the entangling operation (controlled $F$)
\begin{equation}
CF=\bigoplus_i F_i\equiv \mathrm{diag}(\{F_i\})=\sum_i I_i F_i.
\end{equation}
This is a diagonal block matrix (direct sum) with blocks of the size of the qubit space, and each block
is controlled by the particular Pauli string index projection $I_i=|i\rangle\langle i|$. The operation is unitary, symmetric, and self-reverse $CF^2=I$ as $F_i^2=I$.

We postulate $V=CF(\Gamma I)CF$, i.e.
$V$ is a block matrix with qubit blocks
\begin{equation}
V_{ij}=\Gamma_{ij}F_iF_j
\end{equation}
where $\Gamma$ is a real symmetric matrix only in the string space.
By symmetry of Pauli strings
\begin{equation}
V^{T_x}_{ij}=\Gamma_{ij}F_jF_i
\end{equation}

We have  $\Gamma_{ii}=\Lambda-1$. For $i\neq j$, $\Gamma_{ij}=-1$ if $F_iF_j=F_jF_i$ (commuting).
Otherwise, for the anticommuting pairs, it is a free value which cancels with $V^{T_x}$.
The problem reduced  to find PSD $\Gamma$ with fixed commuting entries at $-1$ and minimal diagonal entries.
It turns out that determining $\Gamma$ to find the minimal $\Lambda$ is relatively easy in many cases, and still programmable (by SDP)
in more complicated cases. The commuting background can be separated to $\Gamma=\Delta-J$ with $J_{ij}=1$ always, and then $\Delta_{ii}=\Lambda$,
$\Delta_{ij}=0$ when $i$, $j$ are different and commute, and otherwise it is free (to reach PSD $\Gamma$).
 Optimizing $\Lambda $ over all free entries gives 
exactly Lovasz  function $\vartheta(\mathcal G)$ \cite{hast,degois2023,xu2024beta}.
In general $\beta(\mathcal G)\leq\vartheta(\mathcal G)$ but the equality holds in many cases, see Appendix~\ref{app:theta-duality}
for the details. That appendix derives the $\Gamma$ dual from the standard trace-normalized primal,
proves equality by strict feasibility, and connects it to the augmented
moment formulation and the duals of higher-moment relaxations.
Thus $\Gamma$ is an explicit $\vartheta$ certificate, not a different graph
parameter. Extra sparsity restrictions can weaken the certificate,
whereas averaging over graph symmetries preserves the optimum.
The hierarchy $\alpha\leq\beta\leq\vartheta\leq\overline\chi_f
\leq\overline\chi$ distinguishes exact quantum bounds from independence and fractional
and integer clique grouping; see Appendix~\ref{app:theta-duality}.

Firstly, in the case of a completely anticommuting set of $F_j$, only diagonal terms of $\Gamma$ are fixed, so $\Gamma=0$ fulfills the PSD condition with $\Lambda =1$. Having two anticommuting cliques $C_1$ and $C_2$ of sizes $c_1$ and $c_2$ we can set $\Gamma_{ij}=-1$ if $i$ and $j$ belong to the different cliques 
(even if they accidentally anticommute), and $+1$
if they belong to the same clique. Now one off-diagonal block reads $-\sqrt{c_1c_2}|\bar{c}_1\rangle\langle \bar{c}_2|$ where
\begin{equation}
|\bar{c}_k\rangle=\frac{1}{\sqrt{c_k}}\sum_{i\in C_k}|i\rangle
\end{equation}
while the diagonal ones read $c_k|\bar{c}_k\rangle\langle \bar{c}_k|$ for $k=1,2$.
The matrix
\begin{equation}
\begin{pmatrix}
c_1&-\sqrt{c_1c_2}\\
-\sqrt{c_1c_2}&c_2
\end{pmatrix}
\end{equation}
is semidefinite as $c_{1,2}>0$ and $\det=0$ (one zero eigenvalue).
The generalization to $n$ cliques with the limit $\Lambda=n$ is straightforward,
\begin{equation}
\Gamma_{ij}=\left\{\begin{array}{l}
n-1\mbox{ if }C(i)=C(j),\\
-1\mbox{ otherwise }
\end{array}\right.
\end{equation}
with $C(i)$ being the clique number of $i$.
Again $\Gamma$ has the structure
\begin{equation}
\sum_i (n-1)c_i|\bar{c}_i\rangle\langle\bar{c}_i|-\sum_{i\neq j}\sqrt{c_ic_j}|\bar{c}_i\rangle\langle\bar{c}_j|
\end{equation}
However, for each vector
$
|x\rangle=\sum_i x_i|\bar{c}_i\rangle
$
we have
\begin{equation}
\langle x|\Gamma|x\rangle=\sum_i (n-1)c_ix_i^2-\sum_{i\neq j}\sqrt{c_ic_j}x_ix_j
 \end{equation}
 Defining $y_i=\sqrt{c_i}x_i$ we get
 \begin{equation}
 \langle x|\Gamma|x\rangle=\sum_i ny_i^2-\left(\sum_i y_i\right)^2\geq 0
\end{equation}
by Cauchy--Schwarz, so $\Gamma$ is PSD. This proves $\vartheta\leq\overline\chi$ for an integer clique partition. The stronger standard relation $\vartheta\leq\overline\chi_f\leq\overline\chi$, with fractional clique-cover number $\overline\chi_f$, follows from the weighted clique-cover formulation \cite{knuth1994}.

 \begin{figure}
\includegraphics[scale=0.7]{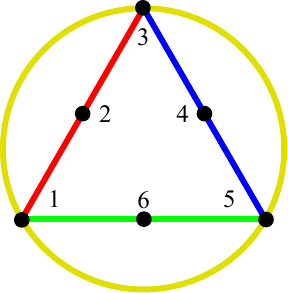}
\caption{Monogamy graph with 4 Mermin correlations on 6 qubits. It includes the $C_3$ graph when removing the yellow circle.}
\label{monc3}
\end{figure}

In many cases $\Gamma$ is not unique. In the cases solved previously
 by anticommuting cliques \cite{kurz,mironowicz,tran}, we can use  different, e.g. symmetric, but still valid $\Gamma$ choices.
For instance,
the ring-triangle, Fig. \ref{monc3}
\begin{equation}
(123)+(345)+(561)+(135)\leq 4\label{cc3}
\end{equation}
Here $\Delta_{ij}=4$ for the identical active sets in $i$ and $j$
and 1 difference on $123$, $345$, $561$ on $2$, $4$, $6$, respectively, or 3 
differences on $135$,
$\Delta_{ij}=2$ for differences $1$  when active sets
of $i$ and $j$ differ, and $\Delta_{ij}=0$ otherwise.

Another example is 
\begin{align}
&(1234)+(1345)+(1452)+(1523)+\nonumber\\
&(1678)+(1789)+(1896)+(1967)\leq 8\label{biset}
\end{align}
which consists of two pools $2345$ and $6789$ linked by the hub $1$.
Now
 $\Delta_{ij}=8$ for the identical
active sets and 3 differences but  identical on the qubit $1$ (i.e. the differences are on the remaining qubits)
$\Delta_{ij}=4$ for $i$ and $j$ in the different active sets but within the same pool (either $12345$ or $16789$)
identical on qubit $1$ ($1$  difference),
$\Delta_{ij}=2$ for $i$ and $j$ from different pools ($1$ difference), and $0$ otherwise.
In Appendix \ref{appmm} we present also the standard case of binary tree.

To demonstrate the power of the new method, we will show monogamy when no longer an anticommuting clique cover exists.
The first example is Mercedes configuration
\begin{equation}
(123)+(14)+(24)+(34)\leq 4\label{merc}
\end{equation}
presented in Fig. \ref{monmerc}.
Classical equal saturation can be traded for outer Mermin maximal violation.
The graph has the largest anticommuting clique of size $4$ and covering requires at least 6 cliques.
We can however construct $\Delta$ as follows:
$\Delta_{ij}=-2$ when $i$ and $j$ both belong to $123$ with 3 differences (4 cases);
$\Delta_{ij}=1$ when $i$ and $j$ belong to the different 2-edges (e.g $14$ and $24$) with 1 difference on 4 (24 cases)
$\Delta_{ij}=+2$ otherwise with 1 difference on common qubit.

For a comparison
\begin{equation}
(12)+(23)+(31)+(124)\leq 4\label{merc2}
\end{equation}
can be solved by clique cover because of slightly different topology.

 \begin{figure}
\includegraphics[scale=0.7]{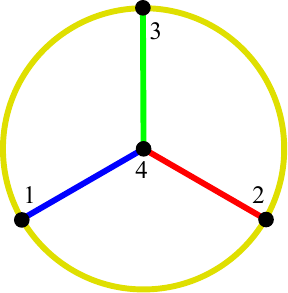}
\caption{Mercedes monogamy graph with Bell and Mermin correlations on 4 qubits.}
\label{monmerc}
\end{figure}

Another example is
\begin{equation}
(125)+(235)+(345)+(415) \leq 4\label{moo4}
\end{equation}
presented in Fig. \ref{mof4}. The graph consists of the central hub 5 always active
and four edges of the square $1234$. One cannot violate simultaneously all 4 Mermin inequalities here.
This support configuration is the three-party ladder case \cite{mironowicz}, with hub $A=5$, $B^{(1,2)}=1,3$, and
$C^{(1,2)}=2,4$. The matrix below gives an alternative exact
squared-correlation certificate.
The largest anticommuting clique is here of size 6, see the example in Table \ref{clic1}.
The valid values of $\Delta$ are as follows:
$\Delta_{ij}=-2$  for $i$ and $j$ from the same active set differing on the hub 5 (16 cases);
$\Delta_{ij}=+2$ for $i$ and $j$ from the same active set with 3 differences or 1 difference
on the square edge (48 cases);
$\Delta_{ij}=+2$ for $i$ and $j$ from different active sets and 1 difference (192 cases).

\begin{figure}
\includegraphics[scale=0.7]{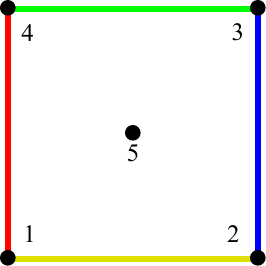}
\caption{Monogamy on $\mathcal M_{ABC}$ with the edges as excluded pairs.}
\label{mof4}
\end{figure}

\begin{table}
\begin{tabular}{cccccc}
\toprule
1&2&3&4&5\\
\midrule
$X$&$X$&$I$&$I$&$X$\\
$Z$&$X$&$I$&$I$&$X$\\
$I$&$Z$&$X$&$I$&$X$\\
$I$&$I$&$X$&$X$&$Z$\\
$I$&$X$&$Z$&$I$&$Z$\\
$I$&$I$&$X$&$Z$&$Z$\\
\bottomrule
\end{tabular}
\caption{An example of a clique of size 6 of the Pauli strings in (\ref{moo4})}
\label{clic1}
\end{table}

Another example is
\begin{align}
&(2345)+(1345)+(1245)+(1235)+\nonumber\\
&(2346)+(1346)+(1246)+(1236)\leq 8\label{mot8}
\end{align}
presented in Fig. \ref{mon8}. The graph consists of poles at $5$ and $6$,  and leaves at  the equator $1234$. It makes it impossible to violate simultaneously all 8 
MABK inequalities in the 6-qubit space.
By permutation symmetry one cannot violate all 15 level-4 MABK inequalities on the 6-qubit grid \cite{munne}.
The largest anticommuting clique is here of size 9, see the example in Table \ref{clic}.
The valid values of $\Delta$ are as follows:
$\Delta_{ij}=6$  for $i$ and $j$ from the same active set, 
identical on the pole and 3 differences on the leaves (64 cases);
$\Delta_{ij}=2$ for $i$ and $j$ from the same active set, 
different solely on the pole (64 cases);
$\Delta_{ij}=2$ for $i$ and $j$ from the same active set, 
different on the pole and two leaves (192 cases)
$\Delta_{ij}=2$ for $i$ and $j$ with different active sets
and 3 differences (512 cases);
$\Delta_{ij}=2$ for $i$ and $j$ with different active sets and 1 difference
(2-overlap $3\cdot 2^9$ cases, 3-overlap $3\cdot 2^9$ cases);
$\Delta_{ij}=-2$ for $i$ and $j$ from the same active set, 
with a single difference on the leaf (192 cases).

\begin{table}
\begin{tabular}{cccccc}
\toprule
1&2&3&4&5&6\\
\midrule
$I$&$X$&$X$&$X$&$X$&$I$\\
$I$&$X$&$X$&$X$&$Z$&$I$\\
$I$&$Z$&$X$&$X$&$I$&$X$\\
$Z$&$I$&$Z$&$X$&$I$&$X$\\
$X$&$I$&$X$&$Z$&$I$&$X$\\
$Z$&$Z$&$I$&$X$&$I$&$Z$\\
$X$&$X$&$I$&$Z$&$I$&$Z$\\
$X$&$Z$&$X$&$I$&$I$&$Z$\\
$Z$&$X$&$Z$&$I$&$I$&$Z$\\
\bottomrule
\end{tabular}
\caption{An example of a clique of size 9 of the Pauli strings in (\ref{mot8})}
\label{clic}
\end{table}

\begin{figure}
\includegraphics[scale=0.7]{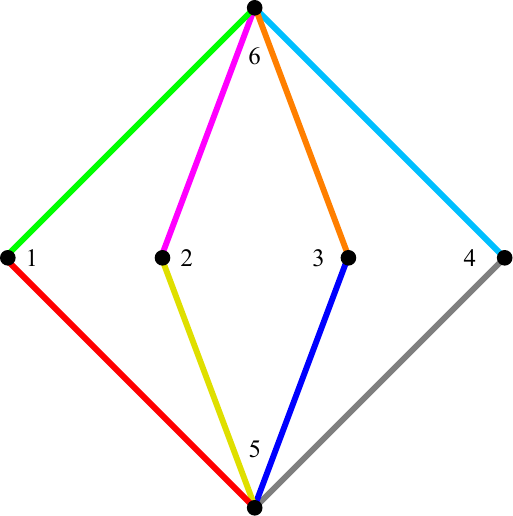}
\caption{Monogamy graph for $\mathcal M_{ABCD}$ correlations for 6 qubits. Here the edges represent pairs not included in the subset.}
\label{mon8}
\end{figure}

We can find not only monogamy relations but also trade-off.

We will show that for $n\geq 4$
\begin{equation}
\sum_j (12\cdots \hat{j}\cdots n)\leq 2^{n-2}\label{trade}
\end{equation}
where $\hat{j}$ means omission of the party. For instance
\begin{align}
&(123)+(234)+(341)+(412)\leq 4,\nonumber\\
&(1234)+(2345)+(3451)+(4512)+(5123)\leq 8
\end{align}
The first inequality is the known Mermin monogamy \cite{kurz}, while the second is a new trade-off relation.
One can check that the maximal anticommuting clique in the latter case has size 8.
In general the bound (\ref{trade}) is saturated by a single subset or jointly by a superposition of $|D^1_n\rangle$ and $|D^n_n\rangle$ Dicke states (\ref{dicc})\cite{pnas}.

It turns out that $\Gamma$ can be defined as follows.
If $i$ and $j$ have $I$ ($\bullet$) at different positions, we assign $\Gamma_{ij}=+1$ for anticommuting pairs.
Otherwise, for the same positions, $\Gamma_{ij}=\Gamma_{h(i,j)}$ is a function of Hamming distance $h(i,j)$.
When $h$ is even, the commutation fixes the $-1$ value so we can use only the odd weights.
For $n=4$, $\Gamma_1=-1$, $\Gamma_3=3$. For $n=5$, $\Gamma_1=-2$, $\Gamma_3=2$.
For $n=6$ the result is no longer unique, there exist many choices of $\Gamma$, $\Gamma_3=-2-\Gamma_1$, $\Gamma_5=20+5\Gamma_1$ and $\Gamma_1\in [-5,-1]$.
One uniform choice for every $n\geq4$ is $\Gamma_h=2h-(n-1)$ for odd $h$,
with $\Gamma_0=2^{n-2}-1$ and $\Gamma_h=-1$ for positive even $h$.
A short Walsh--Hadamard proof is given at the start of Appendix~\ref{appkk};
the subsequent Krawtchouk formulation gives the more general freedom.

Another trade-off relation
\begin{align}
&
(02345)+(03451)+(04512)+(05123)+(01234)
+\nonumber\\
&(0\overline{2345})+(0\overline{3451})+(0\overline{4512})+(0\overline{5123})
+(0\overline{1234})\leq 16,
\end{align}
belongs to the same family as (\ref{biset}) with two pools $12345$ and $\overline{12345}$ linked by the hub $0$. Here
$\Delta_{ij}=6$ for the same active sets $i$ and $j$ and $3$ differences, identical on qubit $0$,
$\Delta_{ij}=4$ for different active sets $i$, $j$, 
but within the same pool (either $012345$ or $0\overline{12345}$) identical on qubit $0$, with $1$ or $3$ differences on common qubits,
$\Delta_{ij}=2$ for the same active sets $i$ and $j$ and $1$ difference, identical on qubit $0$,
$\Delta_{ij}=2$ for the active sets $i$ and $j$ from different pools, differing on qubit $0$.
Otherwise $\Delta_{ij}=0$.

\begin{figure}
\includegraphics[scale=0.7]{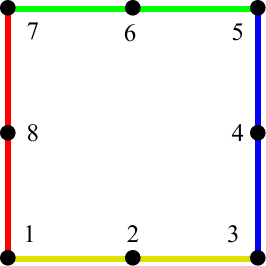}
\includegraphics[scale=0.7]{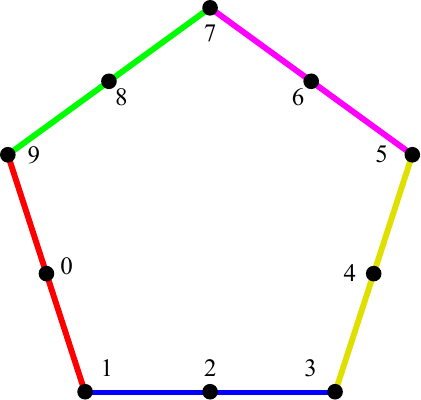}
\caption{Trade-off $C_4$ and $C_5$ graphs with Mermin correlations}
\label{mon45}
\end{figure}

The harder cases belong to the $C_n$ family \cite{tran},
\begin{equation}
C_n=\sum_{i=1}^n(2i-1\: 2i\:2i+1)\label{cnn}
\end{equation}
identifying $i+2n\equiv i$,
 e.g. $C_3$,
\begin{equation}
(123)+(345)+(561)\leq 4
\end{equation}
is an incomplete graph (\ref{cc3}) in Fig. \ref{monc3}.
For the explicit computational-basis states in this subsection, use the locally equivalent $X,Y$ measurement plane instead of the $X,Z$ plane in the definition of $(ijk)$; the norm bounds are unchanged by the corresponding local basis rotation.
All Mermin inequalities are equally violated with $\langle XYY+YXY+YYX-XXX\rangle=4/\sqrt{3}$
saturating the bound for the state
\begin{equation}
|\psi\rangle=\frac{1}{\sqrt{2}}|000000\rangle-\frac{1}{\sqrt{6}}\sum_{k=0}^2 T_3^k|001110\rangle
\end{equation}
with shift by 2, $T_n|c_1\cdots c_{2n}\rangle=|c_{2n-1}c_{2n}c_1\cdots c_{2n-2}\rangle$

For the two overlapping triples $(123)$ and $(345)$, the bound $4$ is also saturated by combined Bell states
\begin{equation}
2^{-3/2}(|00\rangle+|11\rangle)(|0\rangle+|1\rangle)(|00\rangle+|11\rangle)\label{bb2}
\end{equation}
with $(123)=(345)=2$, or
\begin{equation}
a|00000\rangle+b|11011\rangle+c|00111\rangle+d|11100\rangle\label{gg2}
\end{equation}
with $|a|^2+|b|^2=|c|^2+|d|^2=1/2$.
Then $(123)=8(|ad|^2+|bc|^2)$ and $(345)=8(|ac|^2+|bd|^2)$

The graph $C_4$, see Fig. \ref{mon45}, is bounded by
\begin{equation}
(123)+(345)+(567)+(781)\leq 8
\end{equation}
still violates equally Mermin inequality with $\langle XYY+YXY+YYX-XXX\rangle=\sqrt{6}$
for the state
\begin{align}
&|\psi\rangle = \frac{1}{\sqrt{3}}|00000000\rangle  + \frac{1}{2\sqrt{3}} \sum_{k=0}^{1} T_4^k |10111011\rangle\nonumber\\
& - \frac{1}{2\sqrt{2}} \sum_{k=0}^{3} T_4^k |00001110\rangle
\end{align}
but this no longer saturates the bound.

To saturate, we have to take qubits $123$, $567$ ($4$ and $8$ are arbitrary)
and the product
$|GHZ\rangle^2$
for
\begin{equation}
\sqrt{2}|GHZ\rangle=|000\rangle+|111\rangle
\end{equation}
Then
$\langle XXX-XYY-YXY-YYX\rangle=4$ for both groups $123$ and $567$ while
$0$ for $345$ and $781$. All even cases are saturated by analogous products of $GHZ$ states.

However, the case $C_5$, see Fig. \ref{mon45}, gives
\begin{align}
&(123)+(345)+(567)+(789)+(901)\nonumber\\
&\leq 4\sqrt{5}\simeq 8.94427191\leq 10,
\end{align}
confirmed exactly in $Q(\sqrt{5})$ field, which is below the clique limit $10$ \cite{tran}. On the other hand the see-saw lower bound reaches only $8$
(achievable by preceding even case). This leaves a gap between the feasible value and the first-level bound. The algebraic proof below closes this gap for every odd cycle; a higher-order SDP is not needed for that purpose.

Nevertheless, it is only a trade-off relation as Mermin inequality is still jointly violated $\langle XYY+YXY+YYX-XXX\rangle=2.4$ for
\begin{align}
&|\psi\rangle=\frac{\sqrt{5}}{3\sqrt{2}}|0000000000\rangle+\frac{2}{3\sqrt{10}}\sum_{k=0}^4 T_5^k
|1110111000\rangle\nonumber\\
&-\frac{1}{\sqrt{10}}\sum_{k=0}^4T^k_5|1110000000\rangle
\end{align}

 \subsection{Odd cycles }
 
 The nontrivial case when $\beta<\vartheta$ are graphs $C_{2n+1}$ \cite{tran}, with  $4\sqrt{5}$ for $C_5$ as mentioned above.
 The following algebraic argument proves
 \begin{equation}
 C_{2n+1}=(123)+(345)+\cdots (4n+1\:  0\:1)\leq 4n\label{c2n1}
 \end{equation}
 for every $n\geq1$.

We use the odd-cycle relation
\cite{xu2024beta}.
Let $A,B,C$ be Hermitian involutions satisfying
\begin{equation}
 [A,C]=0,\qquad \{A,B\}=\{B,C\}=0.
\end{equation}
Then every state $\rho$ satisfies
\begin{equation}
 \langle A\rangle^2+\langle B\rangle^2
 +\langle C\rangle^2\leq 1+\langle AC\rangle^2.
\end{equation}
Indeed, $D=AC$ is a Hermitian involution commuting with $A,B,C$.
For $P_\pm=(I\pm D)/2$, put
 \begin{align}
 p_\pm&=\operatorname{Tr}(\rho P_\pm),\nonumber\\
 u_\pm&=\operatorname{Tr}(\rho A P_\pm),\qquad
 v_\pm=\operatorname{Tr}(\rho B P_\pm).
 \end{align}
The uncertainty relation for the two anticommuting involutions $A,B$
within either sector gives
$u_\pm^2+v_\pm^2\leq p_\pm^2$; this also holds when $p_\pm=0$.
Since $C=AD$ and $p_++p_-=1$, we obtain
 \begin{align}
 &\langle A\rangle^2+\langle B\rangle^2
       +\langle C\rangle^2\nonumber\\
 &=2u_+^2+2u_-^2+(v_++v_-)^2\nonumber\\
 &\leq2(u_+^2+v_+^2+u_-^2+v_-^2)\nonumber\\
 &\leq2(p_+^2+p_-^2)
   =1+\langle D\rangle^2.
 \end{align}
The argument above uses
only the stated algebraic relations, not that particular representation.

Let $m=2n+1\geq3$, and let $B_1,\ldots,B_m$ be Hermitian
involutions. Suppose adjacent operators anticommute, including $B_m,B_1$,
and all other distinct pairs commute. Then, for every state $\rho$,
\begin{equation}
 \sum_{j=1}^m\langle B_j\rangle^2\leq n.
\end{equation}
For $m=3$ this is the usual uncertainty relation (anticommuting clique) for three pairwise
anticommuting involutions. Suppose $m\geq5$ and define $D=B_1B_3$.
Applying the preceding local inequality to $B_1,B_2,B_3$ gives
\begin{equation}
 \sum_{j=1}^m\langle B_j\rangle^2
 \leq1+\langle D\rangle^2
          +\sum_{j=4}^m\langle B_j\rangle^2.
\end{equation}
The operators $D,B_4,\ldots,B_m$ form a cycle of length $m-2$:
$D$ anticommutes with $B_4$ through its factor $B_3$, and with $B_m$
through its factor $B_1$, while it commutes with every intermediate
$B_5,\ldots,B_{m-1}$. All relations among the retained $B_j$ are unchanged.
Induction therefore bounds the last two terms by $n-1$, proving the lemma.
In particular, the $m=5$ case follows from one contraction to a triangle;
no separate pentagon calculation is needed. This proof also applies to
linear combinations of Pauli strings whenever the displayed involution
and commutation relations hold.

Write $P_0=X$, $P_1=Z$. For each $(u,v)\in\{0,1\}^2$, define
\begin{equation}
 \ell_j(u,v)=
 \begin{cases}
 u,&j<m\text{ odd},\\
 v,&j<m\text{ even},\\
 u\oplus v,&j=m,
 \end{cases}
 \qquad \ell_{m+1}=\ell_1.
\end{equation}
In the family indexed by $(u,v)$ take the two operators
\begin{equation}
 F^{uv}_{j,b}=
 P_{\ell_j}^{(2j-1)}P_b^{(2j)}P_{1-\ell_{j+1}}^{(2j+1)},
 \qquad b=0,1,
\end{equation}
with physical indices understood modulo $2m$.
For each $j$, the two binary linear forms $\ell_j,\ell_{j+1}$
are distinct nonzero forms on $\mathbb F_2^2$. Consequently,
as $(u,v)$ varies, $(\ell_j,1-\ell_{j+1})$ takes all four values.
The four families therefore partition all $8m$ Pauli strings in $C_m$.

Within a fixed family, $F^{uv}_{j,0}$ and $F^{uv}_{j,1}$ anticommute.
Every member of pair $j$ anticommutes with every member of pair $j+1$:
on their common qubit the labels are $1-\ell_{j+1}$ and
$\ell_{j+1}$. Nonadjacent pairs have disjoint supports and commute.

For a fixed state put
\begin{equation}
 a_j=\langle F^{uv}_{j,0}\rangle,\quad
 b_j=\langle F^{uv}_{j,1}\rangle,\quad
 r_j=\sqrt{a_j^2+b_j^2}.
\end{equation}
For $r_j>0$ define
$B_j=(a_jF^{uv}_{j,0}+b_jF^{uv}_{j,1})/r_j$;
if $r_j=0$, take $B_j=F^{uv}_{j,0}$.
These operators satisfy the odd-cycle lemma and
$\langle B_j\rangle=r_j$. Thus each family contributes at most $n$:
\begin{equation}
 \sum_{j=1}^m\sum_{b=0}^1
       \langle F^{uv}_{j,b}\rangle^2\leq n.
\end{equation}
Summing the four inequalities proves (\ref{c2n1}).

The bound is sharp.
Take independent GHZ-equivalent states on the $n$ disjoint triples
with indices $j=1,3,\ldots,2n-1$, and a product state on the remaining
qubits. In the $X,Z$ plane one explicit choice on each selected triple is
\begin{equation}
 |g\rangle=\frac{|000\rangle-|011\rangle-|101\rangle-|110\rangle}{2}.
\end{equation}
It has $\langle ZZZ\rangle=1$ and
$\langle ZXX\rangle=\langle XZX\rangle=\langle XXZ\rangle=-1$.
Each selected triple contributes $4$; every other triple contains a
single qubit from an independent selected state and contributes zero,
since its one-qubit marginal is maximally mixed. Hence $C_{2n+1}=4n$.

\section{Trade-offs and monogamies with 3 settings}
 
 \subsection{Pairwise Bell inequality for 3 parties}
 With only 3 parties $ABC$ we have the monogamy relation when taking the scenario with a party, say $A$ sharing its observables in $\mathcal S_{AB}$
 and $\mathcal S_{AC}$. We can weaken this assumption introducing more observables at $A$. With 4 observables $A$ and simply use two with $B$ and two with $C$
 on separate entangled states, maximizing the violation. Only the case of $3$ observables is not that obvious.
 Suppose $A$ uses $A_{1,2}$ in $\mathcal S_{AB}$ and $A_{1,3}$ in $\mathcal S_{AC}$, i.e.
 \begin{align}
&\mathcal S_{AB}=A_1B_1+A_2B_1+A_1B_3-A_2B_3,\nonumber\\
&\mathcal S_{CA}=C_1A_1+C_2A_1+C_1A_3-C_2A_3
\end{align}
It turns out that
\begin{equation}
\bar{\mathcal S}_{AB}+\bar{\mathcal S}_{CA}\leq 2+2\sqrt{2}\label{ss2}
\end{equation}
and the maximum is attained if one inequality is maximized, e.g. $\bar{\mathcal S}_{AB}=2\sqrt{2}$ 
by standard setting and the state, while $C_1=A_3=-C_2=1$ are classical values. The proof is in Appendix \ref{appb3}.

\begin{figure}
\includegraphics[scale=.5]{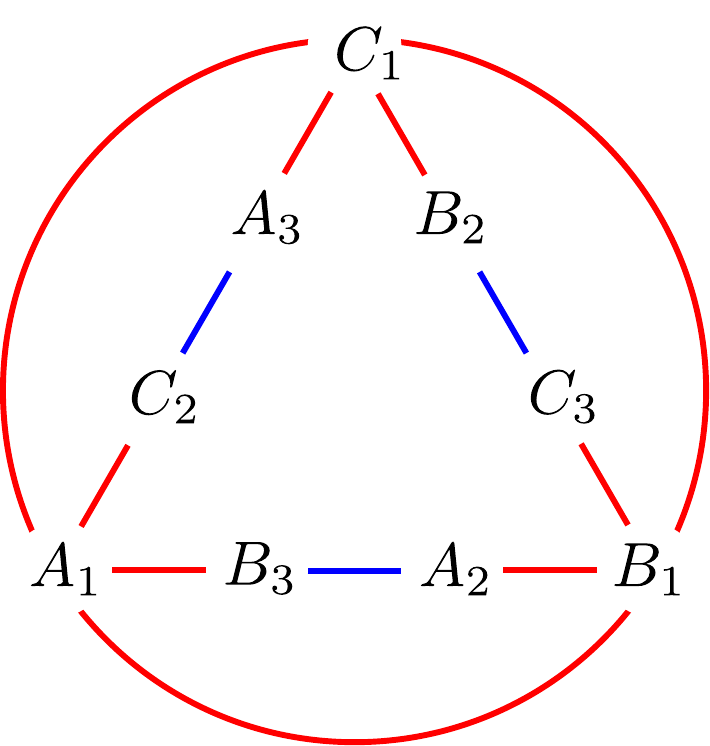}
\caption{The correlations in the triple Bell inequalities with 3 setting for each party.
The red links denote correlations with $+$ sign, while the blue links denote correlations with $-$ sign.}
\end{figure}

Interestingly the scenario can be extended to all 3 inequalities

\begin{align}
&\mathcal S_C=\mathcal S_{AB}=A_1B_1+A_2B_1+A_1B_3-A_2B_3,\nonumber\\
&\mathcal S_A=\mathcal S_{BC}=B_1C_1+B_2C_1+B_1C_3-B_2C_3,\nonumber\\
&\mathcal S_B=\mathcal S_{CA}=C_1A_1+C_2A_1+C_1A_3-C_2A_3
\end{align}
Classically we have each $\mathcal S=2$ for $A_i=B_i=C_i=1$ for $i=1,2,3$
or $\mathcal S=-2$ for $A_1=B_1=C_1=-1$, $A_i=B_i=C_i=1$ for $i=2,3$.

In quantum case, we have
\begin{equation}
-4-2\sqrt{2}\leq \bar{\mathcal S}_A+\bar{\mathcal S}_B+\bar{\mathcal S}_C\leq 4+2\sqrt{2}\label{ss3}
\end{equation}
The upper bound can be reached in the $2\times 2\times 1$ space with
the state $\sqrt{2}|\psi\rangle=|00\rangle+|11\rangle$ and  $A_1=Z$, $A_2=X$, $B_1=(Z+X)/\sqrt{2}$, $B_3=(Z-X)/\sqrt{2}$.
Also $B_2=A_3=\pm I$,  $C_1=\pm 1$, $C_2=C_3=\mp 1$. Then $\mathcal S_{AB}=2\sqrt{2}$, $\mathcal S_{BC}=2$, $\mathcal S_{CA}=2$.
The setup is invariant under shift of parties $A\to B\to C \to A$ or reflection $A\leftrightarrow B$ with $2\leftrightarrow 3$.

We can symmetrize the violation in $5=2+2+1$ dimensions for each party, numbered $01234$, with
 the state
\begin{equation}
\sqrt{6}|\Psi\rangle=|024\rangle+|134\rangle+|402\rangle+|413\rangle+|240\rangle+|341\rangle
\end{equation}
with the explicit observables
\begin{align}
&A_1=B_1=C_1=\begin{pmatrix}
1&0&0&0&0\\
0&-1&0&0&0\\
0&0&s&s&0\\
0&0&s&-s&0\\
0&0&0&0&1\end{pmatrix},\nonumber\\
&A_2=B_2=C_2=\begin{pmatrix}
0&1&0&0&0\\
1&0&0&0&0\\
0&0&1&0&0\\
0&0&0&1&0\\
0&0&0&0&-1\end{pmatrix},\nonumber\\
&
A_3=B_3=C_3=\begin{pmatrix}
1&0&0&0&0\\
0&1&0&0&0\\
0&0&s&-s&0\\
0&0&-s&-s&0\\
0&0&0&0&-1\end{pmatrix};
\end{align}
for $s=1/\sqrt{2}$,
and the violation is $2(\sqrt{2}+2)/3\simeq 2.27614237492$ or in total 

\begin{equation}
\bar{\mathcal S}_A+\bar{\mathcal S}_B+\bar{\mathcal S}_C=
2(\sqrt{2}+2)\simeq 6.82842712475
\end{equation}

The displayed pure state $|\Psi\rangle$ has Schmidt rank five across each
bipartition and is genuinely tripartite entangled. Because the observables are
block diagonal in the flag sectors, removing the coherences between the three
sectors gives the same measured correlations and a biseparable mixture. Thus
these correlations do not certify genuine tripartite entanglement.

This scenario can be implemented on IBM Quantum.
In our circuit it involves 4 qubits per party, e.g. $a=(a_0a_1a_2a_3)$ for $A$ (12 in total)

One starts by creating an entangled 3-qubit state $W$ in the space $|a_0b_0c_0\rangle$
\begin{equation}
|W\rangle=(|100\rangle+|010\rangle+|001\rangle)/\sqrt{3}
\end{equation}
using a $CZ/RZZ$ gates, see Appendix \ref{appdi}. This state serves only as a random generator of 3 options, its coherence is irrelevant.
Each of the 3 qubits becomes the root of one of the 3 maximal solutions, with maximal violation of one of $\mathcal S$
while attaining classical maximum for the other two.
Firstly, we fan it out the state $W$ to
$(|110000\rangle+|001100\rangle+|000011\rangle)/\sqrt{3}$
in the space $|a_0c_1b_0a_1c_0b_1\rangle$
by a $CX$ gates with $a_0$ as the control and $c_1$ as the target, repeated in the cycle $a\to b\to c\to a$.

Separately, we prepare Bell entangled states on pairs $a_2c_3$, $b_2a_3$, $c_2b_3$.
The settings (angles) at qubits 2 and 3 are chosen according to the standard CHSH configuration.
Since $A_1$, $B_1$, $C_1$ are shared, for the choice $A_1$ we set the identical 
measurement of $a_2$ and $a_3$, while $A_2$ means the complementary setting at $a_3$ and $A_3$ means the complementary setting at $a_2$, similarly $B$ and $C$. The outcome conversion from 4 bits to a single bit is as follows for the party $A$,
\begin{enumerate}
\item if $a_0=1$ then return the value at $a_2$ except the setting $2$ returning $0$ ($+1$ in Bell assignment),
\item otherwise if $a_1=1$ then return the value at $a_3$ except the setting $3$ returning $0$ ($+1$)
\item otherwise return $0$ ($+1$) for setting $1$ and $1$ ($-1$) for settings 2,3
\end{enumerate}
Note that ideally $a_0a_1=00,01,10$ but the errors allow the unwanted case $11$, which is appended to the $10$ case in the above protocol.

\begin{figure*}
\includegraphics[scale=.5]{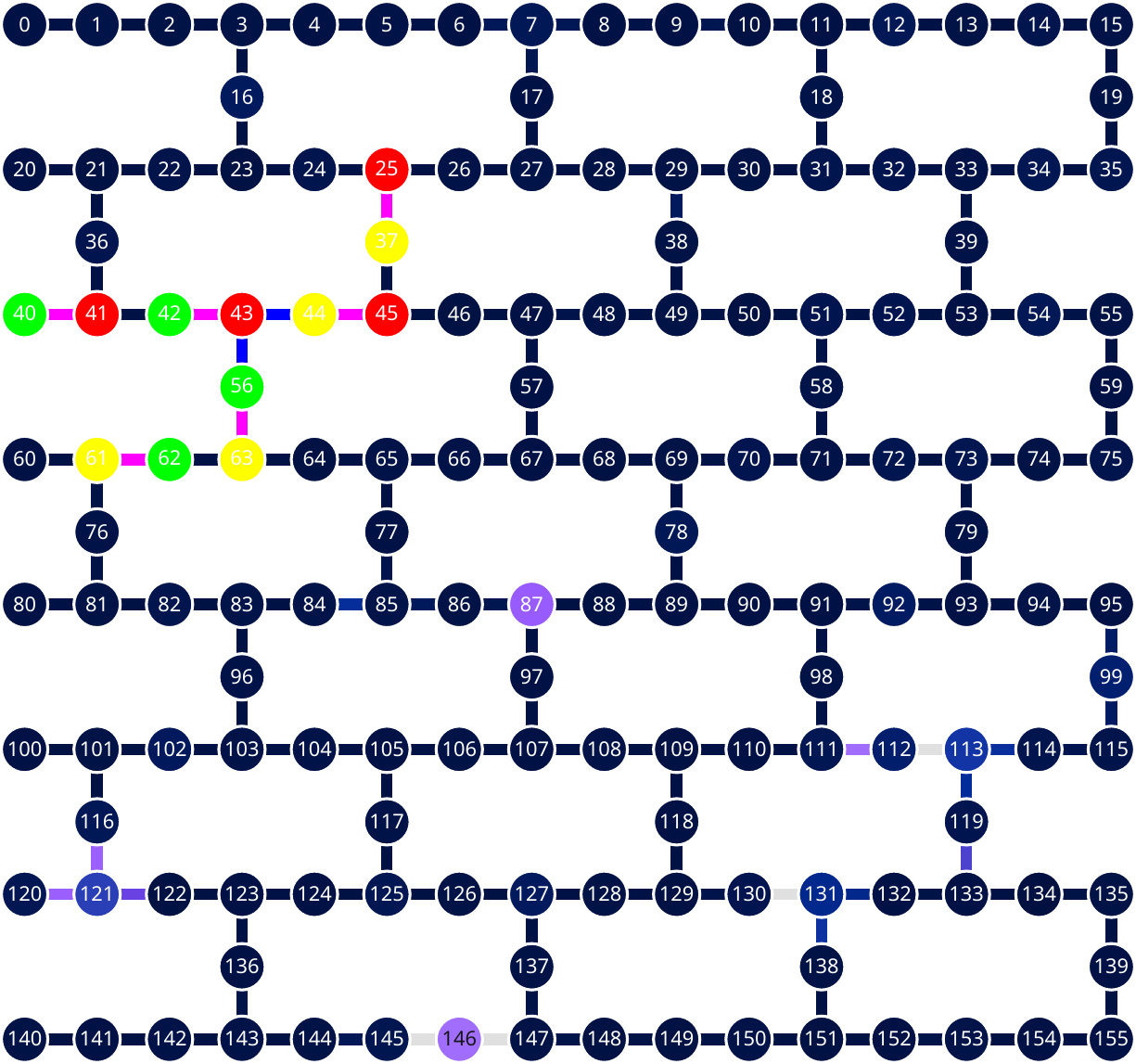}
\caption{The placement of the qubits of the respective parties $A$ (red), $B$ (yellow), $C$ (green)
on the \emph{ibm\_kingston} grid.
The central triple $0$ qubits is connected by blue links, and qubits $1$ are connected to them by violet links.
The remaining qubits $2$, $3$ are always separately entangled by violet links.}\label{ibms3}
\end{figure*}

 We have run 20 jobs, 6 repetitions, 10000 shots on \emph{ibm\_kingston}.
encoding $a_0b_0c_0c_1a_1b_1a_2b_2c_2c_3a_3b_3$ as $43,44,56,42,45,63,41,37,62,40,25,61$, see Fig. \ref{ibms3}
and Table \ref{exs}.

\begin{table}
\begin{tabular}{cccc}
\toprule
&$A$&$B$&$C$\\
\midrule
$\bar{\mathcal S}$&2.02189 & 2.03136 &2.04169\\
$\Delta{\mathcal S}$&0.00088 & 0.00088& 0.00088\\
\bottomrule
\end{tabular}
\caption{Experimental values of $\bar{S}$ with a shared observable and shot-noise standard errors for parties $ABC$ with the specified party excluded} \label{exs}
\end{table}

The total bound is reached in the three cases of maximal positive violation, $4+2\sqrt{2}$, $|A\rangle$: $A$ separate $BC$ entangled
$|B\rangle$: $B$ separate $CA$ entangled, $|C\rangle$: $C$ separate, $AB$ entangled.
For instance in case $C$: acting directly on the state maximizing state $|C\rangle$ the operators can be replaced
\begin{align}
&\sqrt{2}A_1=B_1+B_3,\:\sqrt{2}A_2=B_1-B_3,\nonumber\\
&\sqrt{2}B_1=A_1+A_2,\:\sqrt{2}B_3=A_1-A_2\nonumber\\
&B_2=A_3=C_1=-C_2=-C_3
\end{align}
For the representative saturating strategy (or on its relevant support), the first identities together with $A_i^2=B_i^2=C_i^2=1$ and commutativity $A_iB_j=B_jA_i$, $B_iC_j=C_jB_i$, $C_iA_j=A_jC_i$ imply
 $B_3B_1=-B_1B_3$ and $A_2A_1=-A_1A_2$. Moreover
$A_1A_3=A_3A_1$, $A_2A_3=A_3A_2$, $B_2B_1=B_1B_2$, $B_2B_3=B_3B_2$.
We can express all operators in terms of $A_1\equiv A$, $B_1\equiv B$, $C_1\equiv C$

The maximal negative violation is not a naive sign reversal, as all correlation are of second order and contain different observables.
Nevertheless, the three cases of maximal negative violation, $-4-2\sqrt{2}$, $|A\rangle$ are again: $A$ separate $BC$ entangled
$|B\rangle$: $B$ separate $CA$ entangled, $|C\rangle$: $C$ separate, $AB$ entangled.

For instance in case $C$: acting directly on the  minimizing state $|C\rangle$ the operators can be replaced
\begin{align}
&-\sqrt{2}A_1=B_1+B_3,\:-\sqrt{2}A_2=B_1-B_3,\nonumber\\
&-\sqrt{2}B_1=A_1+A_2,\:-\sqrt{2}B_3=A_1-A_2,\nonumber\\
&B_2=A_3=-C_1=C_2=C_3
\end{align}
It is realized taking the maximal state, reversing $C_i\to -C_i$ and either $A_{1,2}$ or $B_{1,3}$ (not both).
We proceed analogously as in the maximizing case, carefully watching sign flips.
The symbolic proof of both the exact bounds on (\ref{ss3}) using NPA hierarchy is presented in Appendix \ref{appb3}

\subsection{\texorpdfstring{$\mathcal I_{3322}$ polygamy}{I3322 polygamy}}

Another relevant two-party inequality with 3 setting is \cite{froiss,collins,sliwa}
\begin{align}
&\mathcal I\equiv \mathcal I_{3322}=A_1+A_2+B_1+B_2+(A_1+A_2)(B_1+B_2)\nonumber\\
&+A_3(B_1-B_2)+B_3(A_1-A_2)\geq -4
\end{align}
The inequality is violated with $\bar{\mathcal I}=-5$ for the Bell state $\sqrt{2}|\psi\rangle=|00\rangle+|11\rangle$ 
with $2A_{1,2}=\sqrt{3}X\pm Z$, $2B_{1,2}=-\sqrt{3}X\mp Z$, $B_3=-Z$, $A_3=Z$ but a value of about $-5.0035$ is approached by increasing-dimensional constructions; the unrestricted optimum and the need for infinite dimension were conjectured on numerical evidence \cite{i3322max}.
Unlike CHSH, it can be also jointly violated for parties $AB$ and $AC$ \cite{collins}. The violation is very weak leading to the question if it can really
discriminate two-party and three-party entanglement.

Our numerical search finds its strongest joint violation with three-dimensional $A$ and qubits $B$ and $C$.
The candidate optimal state has the following expansion in the $|abc\rangle$ basis
\begin{equation}
|\psi\rangle=
\psi_0|0+\rangle+\psi_1|0-\rangle+\psi_2|1-\rangle+\psi_3|2\bar{0}\rangle
\end{equation}
where $\sqrt{2}|\pm\rangle=|00\rangle\pm|11\rangle$, $\sqrt{2}|\bar{0}\rangle=|01\rangle+|10\rangle$
in the $BC$ space.
The observables read
$A_1=2P_+-I$,
$A_2=2P_--I$,
where $P_\pm=|a_\pm\rangle\langle a_\pm|$ is a projection on
$
|a_\pm\rangle=x|0\rangle+y|1\rangle\pm z|2\rangle
$
and
\begin{equation}
A_3=\begin{pmatrix}
0 &0 &1\\
0 &1&0\\
1&0&0\end{pmatrix}
\end{equation}
and
$B_1=C_1=Z\cos\phi+X\sin\phi$, $B_2=C_2=Z\cos\phi-X\sin\phi$, $B_3=C_3=-X$.
We get 
\begin{equation}
\bar{I}_{AB}+\bar{I}_{AC}=\mu-4\simeq -8.2003459282962
\end{equation}
where $\mu$ is the real root of degree-7 equation
\begin{align}
&49152 + 24576x - 512x^2 - 3584x^3\nonumber\\
& + 220x^4 + 114x^5 - 18x^6 + x^7=0.
\end{align}
The state is then an eigenstate with coefficients
\begin{equation}
\begin{pmatrix}
\psi_0\\
\psi_1\\
\psi_2\\
\psi_3\end{pmatrix}
\simeq
\begin{pmatrix}
0.6902281889\\
0.2189839461\\ 
-0.5080080578\\
-0.4664320870\end{pmatrix}
\end{equation}
and angle  $\phi\simeq 0.5318627701$ and direction
\begin{equation}
\begin{pmatrix}
x\\
y\\
z
\end{pmatrix}
 \simeq \begin{pmatrix}
0.5233652544\\
0.7650506662\\
0.3752149898
\end{pmatrix}
\end{equation}

Numerical NPA bounds agree with this construction to the reported precision; this agreement is not an exact certificate of global optimality.
To find it we use SDP with the basis $I$,$D_i$, $D_iD_j$, for $D=A,B,C$, $i,j=1,2,3$, $i\neq j$ and their products.
The total basis of $10^3=1000$ elements can be split into 8 groups, due to $Z_2^3$ symmetries: exchange of parties $B\leftrightarrow C$,
relabeling $1\leftrightarrow 2$ for $B,C$ with $A_3\to -A_3$, relabeling $1\leftrightarrow 2$ for $A$ and $B_3,C_3\to -B_3,-C_3$.
The symmetries fit the basis $I$,$D_+$,$D_-$,$D_3$,$D_3D_+$,$D_3D_-$,$D_+D_3$,$D_-D_3$, $D_1D_2+D_2D_1$, $D_1D_2-D_2D_1$, with $D_\pm= D_1\pm D_2$.
It reduces the SDP task to express $\mathcal I_{AB}+\mathcal I_{AC}-\lambda$ by the sum of quadratic forms generated by 8 semidefinite matrices, similarly as in the case
of $\mathcal S_{A}+\mathcal S_B+\mathcal S_C$ and maximize $\lambda$.
The reported bound agrees within the numerical error with the maximal state construction.

If all parties are qubits, the violation is smaller $\mu_2\simeq -8.12886621095481$.
In fact $\mu_2$ is the root of the equation $m(x)\equiv\sum_i m_ix^i=0$, see Table \ref{i33} in Appendix \ref{appis}

Then $A_3=B_3=X$ and $A_1=Z\cos\alpha+X\sin\alpha$, $A_2=Z\cos\alpha-X\sin\alpha$,
and $B_1=Z\cos\beta+X\sin\beta$, $B_2=Z\cos\beta-X\sin\beta$,
with
\begin{equation}
\alpha\simeq -2.566007180965164,\;\beta\simeq -2.863135308591454
\end{equation}

The numerical lower bound considered for biseparable states is higher but still $<-8$.
Suppose now that one party is classical so the state is biseparable.
Across $A|BC$, both pairs are local, so $\bar{I}_{AB}+\bar{I}_{AC}$. The other two bipartitions are equivalent under $B\leftrightarrow C$. For $AB|C$, linearity permits minimization over the eight deterministic assignments
$ (C_1,C_2,C_3)\in\{-1,+1\}^{3}$,
followed by optimization over $AB$. Finally, linearity extends the bound to mixtures across bipartitions.
We can find a state 
\begin{equation}
|\psi\rangle=\sum_{k}\lambda_k|kk 0\rangle
\end{equation}
in the space $|abc\rangle$ (i.e. $C$ is separated) and operators $A_i$, $B_i$, $C_i$
such that
\begin{align}
&\langle 2(A_1+A_2)+(A_1+A_2+1)(B_1+B_2+C_1+C_2)+\nonumber\\
&A_3(B_1-B_2+C_1-C_2)+(B_3+C_3)(A_1-A_2)\rangle< -8
\end{align}
The numerical NPA lower bound with the same basis is $\simeq -8.1294823381768871$ 


For $A$,$B$ two dimensional and $C_1=C_2=-1$, $C_3=1$ we have the minimum 
$\simeq  -8.0243205$
with $A_{1,2}=Z\cos\alpha\pm X\sin\alpha$,  
$B_{1,2}=Z\cos\beta\pm X\sin\beta$,
 $A_3=Z\cos\alpha'+ X\sin\alpha'$,  
$B_3=Z\cos\beta'+ X\sin\beta'$,
for
\begin{equation}
(\alpha,\alpha',\beta,\beta')\simeq
(2.636106,  1.6845197,   0.180805,   2.271655)
\end{equation}

The best value found numerically can be reached with 3-dimensional $A$ and $B$, of the form
\begin{align}
&A_1=\begin{pmatrix}
x&y&0\\
y&-x&0\\
0&0&-1\end{pmatrix},\;
A_2=\begin{pmatrix}
z&w&0\\
w&-z&0\\
0&0&-1\end{pmatrix},\nonumber\\
&A_3=\begin{pmatrix}
1&0&0\\
0&0&1\\
0&1&0\end{pmatrix},\;
B_1=\begin{pmatrix}
-c&0&s\\
0&-1&0\\
s&0&c\end{pmatrix},\nonumber\\
&B_2=\begin{pmatrix}
-c&0&-s\\
0&-1&0\\
-s&0&c\end{pmatrix},\;
B_3=\begin{pmatrix}
1&0&0\\
0&f&g\\
0&g&-f\end{pmatrix},\;
\end{align}
with $x^2+y^2=z^2+w^2=c^2+s^2=f^2+g^2=1$.
In particular, numerical maximization gives
\begin{equation}
\begin{pmatrix}
x\\
y\\
z\\
w\\
c\\
s\\
f\\
g
\end{pmatrix}
\simeq
\begin{pmatrix}
 0.96162288\\
-0.2743746\\
0.58851723\\
 0.80848467\\
 0.77108603\\
 0.63673097\\
0.34478694\\
 0.93868097
 \end{pmatrix}
\end{equation}
The minimizing state in the $AB$ space $|ab\rangle$, $a,b=0,1,2$, 
\begin{align}
&|\psi\rangle=\psi_0|01\rangle+\psi_1|11\rangle+\psi_2|12\rangle+\psi_3|20\rangle.\nonumber\\
&\begin{pmatrix}
\psi_0\\
\psi_1\\
\psi_2\\
\psi_3\end{pmatrix}
\simeq
\begin{pmatrix}
0.615936119\\
0.247412381\\
0.313235046\\
-0.679185994
\end{pmatrix}
\end{align}

Among the strategies studied, the largest joint violation uses a three-dimensional system $A$. The violations are relatively small, see Fig.~\ref{ibar}; the qubit values found do not exceed the biseparable values found with three-dimensional systems $A$ and $B$.

\begin{figure}
\includegraphics[scale=1.2]{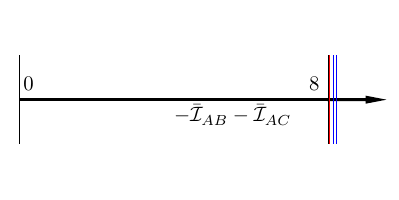}
\caption{Scale of the negated joint value $-(\bar{\mathcal I}_{AB}+\bar{\mathcal I}_{AC})$.
The classical bound is $8$. For qubits $A$, $B$ we can reach only red bounds (smaller biseparable and larger with qubit $C$).
For 3-dimensional $A$ and $B$ separated from $C$ we get minimally farther than qubit bound (left blue practically over the right red), 
and the best value found, $\sim 8.2$ (right blue), is attained using 3-dimensional $A$ and qubits $BC$.}
\label{ibar}
\end{figure}

A natural question is if the $\mathcal I_{3322}$ inequality can be violated jointly for all three pairs $AB$, $BC$, $CA$.
We shall allow an arbitrary sign combination. Namely we define
\begin{equation}
\mathcal I^{xyz}(A,B)=\mathcal I(A,\{xB_1,yB_2,zB_3\}),
\end{equation}
i.e. $B_j$ can be sign-flipped. It does not change the classical bound $\geq -4$ since $|B_j|\leq 1$ condition does not depend on the sign.
In the cyclic sum $\mathcal I(A,B)+\mathcal I(B,C)+\mathcal I(C,A)$ we consider the common sign patterns specified below. The result here concerns this family of relabellings, not every independent relabelling of the three Bell expressions.

It turns out that none of these combinations, maintaining cyclic symmetry, violates simultaneously the classical bound, i.e.
\begin{equation}
\mathcal I^{xyz}(A,B)+\mathcal I^{xyz}(B,C)+\mathcal I^{xyz}(C,A)+12\geq 0
\end{equation}
for each choice of signs $x,y,z=\pm$.
The proof involves SDP protocol to reduce the left-hand side to the sum of squares,
see Appendix \ref{appis}.

Surprisingly, there is a combination of $\mathcal I$ with larger violation
\begin{align}
&\mathcal I^{a}=\mathcal I(A,B)+\mathcal I(-A,C)=\nonumber\\
&B_1+B_2+C_1+C_2+(A_1+A_2)(B_1+B_2-C_1-C_2)+\nonumber\\
&A_3(B_1-B_2-C_1+C_2)+(A_1-A_2)(B_3-C_3).\label{issa}
\end{align}
The actual relabellings $1\leftrightarrow 2$ and sign flips of $B$ and $C$ are irrelevant.

The minimum reads $-6-2\sqrt{2}\simeq -8.828<-8$.
The minimum is attained, for example, by a biseparable combination of classical $B$,  $B_1=B_2=-1$, $B_3=A_2=1$ and the Bell state $AC$, $(|01\rangle-|10\rangle)/\sqrt{2}$, $C_3=-1$ or $C_3=(Z+X)/\sqrt{2}$ with
\begin{equation}
C_1=X,\: C_2=Z,\: A_1=-(Z+X)/\sqrt{2}, A_3=(Z-X)/\sqrt{2},
\end{equation}
We remind that $\mathcal I(A,B)$ inequality reduces to Bell-CHSH choosing e.g. $A_2=-1$, $B_3=-1$ \cite{collins}.
Therefore this combination cannot detect full tripartite entanglement.
The proof of the bound is in Appendix \ref{appis}. 

These two representative minimizing strategies have the properties
$C_+=\sqrt{2}A_1$, $C_-=\sqrt{2}A_3$ and $B_+=-2$, $B_-=0$,
when acting on the minimal state, with $D_\pm=D_1\pm D_2$, and $C_3=-1$ or $C_3=A_1$ on the state.

Making the above substitutions we get directly $-6-2\sqrt{2}$.
The other solutions are generated by the symmetries, see Appendix \ref{appis}.

Finally, we can consider a hybrid inequality
\begin{equation}
\mathcal I(A,B)+\mathcal S(\{A_1,A_3\},C)\label{ich}
\end{equation}
where $
\mathcal S=A_1C_++A_3C_-$
The joint minimum is $-4-2\sqrt{2}$ obtained in 5 biseparable cases, one with entangled $AC$ and 4 with entangled $AB$, see Appendix \ref{appis} (also the proof).

\section{Conclusions}

We have presented several cases of polygamous behavior of quantum nonlocality with two-outcome observables,
with either two or three settings. The 3-party local  inequality  turns out to have a sufficiently large gap between classical and quantum bound
to be violated simultaneously in all subsets of the 4-party set on IBM Quantum. It demonstrates both technological progress in hardware
development and the important role of polygamous inequalities to verify large-scale entanglement.
We have analysed 4-party MABK inequality violated simultaneously in each 4-subset of the 5 parties and also confirmed the simultaneous violation on IBM Quantum. The violation of classical bounds
is possible not only by the fully entangled state but also on randomly separable ones. We have developed explicit SDP and partial-transpose certificate constructions for monogamy and trade-off relations, building on the graph-theoretic bounds cited above. The first-level bound need not be tight for larger graphs.

The picture is more complicated with 3 settings, where already CHSH can be jointly pairwise violated in the 3-party setup, with the optimum already attainable by mixtures of biseparable strategies. Nevertheless, the simultaneous violation is confirmed again on IBM Quantum.

The inherently 3-settings $\mathcal I_{3322}$ inequality can be violated in either a symmetric or an antisymmetric way.
The symmetric violation is weak and can certify genuine tripartite entanglement only when the applicable biseparable bound and dimension assumptions are established. The antisymmetric violation is larger, but its optimum is already attainable by biseparable strategies. The three cyclic copies of $\mathcal I_{3322}$ with the common sign patterns considered here cannot all be violated simultaneously, although joint violation of two copies is possible. The hybrid $\mathcal I_{3322}$--CHSH optimum is likewise attainable by a biseparable strategy.

Our findings point out the convergence between progress in quantum computing and nonlocal inequalities.
Future research can test more complicated inequalities on quantum computing, as a diagnostic tool
but also confirmation of nonclassical predictions. Further work on SDP and trade-off constraints can help
to filter out best candidates for polygamous behavior.

\emph{Acknowledgements.}
		The results have been created using IBM Quantum. The views expressed are those of the
	authors and do not reflect the official policy or position of the IBM Quantum or IonQ team.
	TR gratefully acknowledges the funding support by program "{}Excellence
	initiative -- research university"{} for the AGH University in Krakow as well as the
	ARTIQ project: UMO-2021/01/2/ST6/00004 and ARTIQ/0004/2021. We also acknowledge the use of ChatGPT (OpenAI, https://chat.openai.com) for assistance in developing some proofs and partially checking the work. The authors remain responsible for all claims and proofs. We thank Wies{\l}aw Laskowski and Jaros{\l}aw Korbicz for fruitful discussions.

\appendix

\section{Pauli algebra and IBM Quantum gates}

\label{appdi}

Throughout the text we work with standard Pauli algebra
\begin{align}
&X=\begin{pmatrix}
0&1\\
1&0
\end{pmatrix},\;
&Y=\begin{pmatrix}
0&-i\\
i&0
\end{pmatrix},\nonumber\\
&Z=\begin{pmatrix}
1&0\\
0&-1
\end{pmatrix},\;
&I=\begin{pmatrix}
1&0\\
0&1
\end{pmatrix}
\end{align}
in the basis $|0\rangle$, $|1\rangle$,
or equivalently
\begin{align}
&X=|0\rangle\langle 1|+|1\rangle\langle 0|,\nonumber\\
&Y=i|1\rangle\langle 0|-i|0\rangle\langle 1|,\nonumber\\
&Z=|0\rangle\langle 0|-|1\rangle\langle 1|,\nonumber\\
&I=|0\rangle\langle 0|+|1\rangle\langle 1|,
\end{align}
Then $Z|0\rangle=|0\rangle$, $Z|1\rangle=-|1\rangle$.
The useful properties are
\begin{align}
&X^2=Y^2=Z^2=I,\nonumber\\
&XY=iZ=-YX,\nonumber\\
&ZX=iY=-XZ,\nonumber\\
&YZ=iX=-ZY
\end{align}

The basic IBM Quantum gates \cite{ibm,qis} are unitary rotations defined by Pauli matrices. We shall denote $V_\theta=\exp(-iV\theta/2)=I\cos(\theta/2)-iV\sin(\theta/2)$
when $V^2=I$.
Then the basic gates are $X_\theta\equiv RX(\theta)$ and $Z_\theta\equiv RZ(\theta)$.
All gates are defined put to the global phase factor $e^{i\phi}$ canceling by unitarity.
The initial state is always $|0\rangle$ and we use shorthand notation $V_\pm\equiv V_{\pm\pi/2}$ or $V_\pi=V$ and
$V_{\theta+2\pi}\equiv V_{\theta}$. An important gate is also Hadamard
\begin{equation}
H=\frac{X+Z}{\sqrt{2}}=\frac{1}{\sqrt{2}}\begin{pmatrix}
1&1\\
1&-1\end{pmatrix}
\end{equation}
which is also self-inverse, $H^2=I$.
The measurement is by default in the $Z$ basis i.e. $I_0=|0\rangle\langle 0|$ and $I_1=|1\rangle\langle 1|$ are the projectors corresponding to the two outcomes
$0$ and $1$, which can be mapped to $+1$ and $-1$ for Bell correlators.

The basic two-qubit gates are $CZ$ (controlled-$Z$) defined in the $|ab\rangle$ basis
as $CZ|ab\rangle=(-1)^{ab}|ab\rangle$ or directly 
\begin{equation}
\begin{pmatrix}
1&0&0&0\\
0&1&0&0\\
0&0&1&0\\
0&0&0&-1\end{pmatrix}
\end{equation}
or $RZZ(\theta)$, see below.

Now we explain the generation of Dicke states by available gates.
The key mechanism is the conditional rotation, keeping two qubit states $|00\rangle$ and $|11\rangle$ unaffected while
applying rotation in the basis $|01\rangle$, $|10\rangle$,
\begin{equation}
\begin{pmatrix}
\cos\theta& -i\sin\theta\\
-i\sin\theta & \cos\theta
\end{pmatrix}
\end{equation}
by applying $R=\exp(-i\theta(XX+YY)/2)$ operation. The operation is effectively realized by IBM Quantum $RZZ(\theta)$ gates equivalent
to $(ZZ)_\theta=\exp(-iZZ\theta/2)$. For the rotation we apply this gate twice
as $R=(XX)_\theta(YY)_\theta$, while $(XX)_\theta=(Y_+Y_+)(ZZ)_\theta(Y_-Y_-)$ and
$(YY)_\theta=(X_-X_-)(ZZ)_\theta(X_+X_+)$ since $X=Y_+ZY_-$. It can be further shortened
by observing that $X_+Y_+=X_+Z_+X_+Z_-=Z_+X_+Z_+Z_-=Z_+X_+$ since $H=Z_+X_+Z_+=X_+Z_+X_+$
 \begin{figure}
\begin{tikzpicture}[scale=1]
		\begin{yquant}
			qubit {} q[2];
			box {$Y_-$}  q[0];
			box {$Y_-$}   q[1];
			[name=ptA] zz (q[0, 1]);
			box {$X_+$}  q[0];
			[name=ptB] box {$X_+$}   q[1];
			box {$Z_+$}  q[0];
			box {$Z_+$}  q[1];
			[name=rtA] zz {$\theta$} (q[0],q[1]);
			box {$X_-$}  q[0];
			[name=rtB] box {$X_-$}   q[1];
			
	\end{yquant}
	\path (ptA) -- (ptA |- ptB) node[midway, fill=white, inner sep=1.5pt] {$\theta$};
	\path (rtA) -- (rtA |- rtB) node[midway, fill=white, inner sep=1.5pt] {$\theta$};
\end{tikzpicture}

\caption{The $V_\theta$ gate expressed by two $RZZ(\theta)$ gates, denoted by vertical links with $\theta$. }
\label{cnotr}
\end{figure}
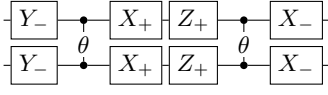

\section{Maximum of sum of squares of averages, Jordan's lemma and convexity}
\label{appjor}

Let $F_i$ be a real linear combination of tensor products of Hermitian operators $A^{(1)}_{i_1}A^{(2)}_{i_2}\cdots A^{(n)}_{i_n}$
with $i_k=0,1,2$ and $A^{(k)}_0=I^{(k)}$ when a party is omitted, and $A^{(i)}_{1,2}=A^{(i)}_{1,2}(+)-A^{(i)}_{1,2}(-)$
with $A^{(i)}_{1,2}(\pm)\succeq 0$ and $A^{(i)}_{1,2}(+)+A^{(i)}_{1,2}(-)=I^{(i)}$. This is the general form
of two-setting two-outcome Bell-type correlators.
Then it suffices to check the inequality
\begin{equation}
\sum_j \langle F_j\rangle^2\leq \Lambda \label{qqd}
\end{equation}
in the case of qubits, real operators, and real pure states.

The inequality is equivalent to
\begin{equation}
\sum_j x_j\langle F_j\rangle\leq \sqrt{\Lambda}\label{lqd}
\end{equation}
for all real vectors $x$, normalized $\sum_j x_j^2=1$.
By Cauchy-Schwarz inequality it follows from (\ref{qqd}).
Conversely, we take $x_j=\langle F_j\rangle/\sqrt{S}$ with $S=\sum_j\langle F_j\rangle^2>0$; the case $S=0$ is immediate.
On the other hand (\ref{lqd}) is linear with respect to operators $A^{(k)}$ so its maximum at fixed $x$ must be achievable at the endpoints.
For each fixed setting separately, its two complementary effects commute and have a joint diagonal representation. Keeping all other measurements fixed, the diagonal terms can be pushed to the boundaries, either $0$ or $1$,
becoming projections. With projections, we have $A^{(k)2}_{1,2}=I$.

We can use Jordan's lemma \cite{mironowicz,jorle}.
It says that operators $A_1^2=A_2^2=I$ can be reduced to the $2\times 2$ block representation.
Let us denote $A_+=A_1+A_2$ and $A_-=A_1-A_2$.
As $A_i^2=1$ we have 
\begin{align}
&A_+A_-=(A_1+A_2)(A_1-A_2)=\nonumber\\
&A_1^2+A_2A_1-A_1A_2-A_2^2=[A_2,A_1]=-A_-A_+
\end{align}
for the commutator $[F,G]=FG-GF$. Both $A_\pm$ are Hermitian. Moreover $A_+^2A^2_-=A^2_-A^2_+$ so we can jointly diagonalize $A_\pm^2$.
Note also that $A_+^2+A_-^2=4$. We can arrange the eigenspaces so that $A_+^2=4\cos^2\phi$, $A_-^2=4\sin^2\phi$
for different values of $\sin^2\phi$.
If $\sin^2\phi=0$, then $A_+=\pm 2$ and can be diagonalized, and similarly $A_-$ if $\cos^2\phi=0$. If both $\cos^2\phi$ and $\sin^2\phi$
are nonzero we diagonalize $A_+$ with eigenstates $v_\pm$ of eigenvalues $\pm 2\cos\phi$.
Then
\begin{equation}
A_+A_-|v_\pm\rangle=-A_-A_+|v_\pm\rangle=\mp2A_-|v_\pm\rangle \cos\phi
\end{equation}
It means that $A_-$ takes states $|v_\pm\rangle$ into the space of $|v_\mp\rangle$ states. By singular value decomposition
we can diagonalize the off-diagonal block of $A_-$ by separate basis changes in the eigenspaces $A_+=\pm2\cos\phi$.
We end up with $2\times 2$ blocks of $A_+=2Z\cos\phi$ and $A_-=2X\sin\phi$. 

Every quantum state $\rho$ can be written in the tensor basis of these blocks but the cross-block terms do not contribute
to $\langle F_j\rangle$. Therefore $\langle F_j\rangle=\langle F_j\rangle'$ where $\rho'$ is obtained from $\rho$
by discarding cross-blocks by projecting onto the blocks. In the diagonal basis $\rho'$ is a convex combination of projections,
so we can restrict to pure states  $|\psi\rangle$ within single blocks.
Finally, if the state $|\psi\rangle=|\psi_r\rangle+i|\psi_i\rangle$
with real $|\psi_{r,i}\rangle$, we can
write
\begin{equation}
\langle F_j\rangle=\langle \psi_r|F_j|\psi_r\rangle+\langle \psi_i|F_j|\psi_i\rangle
\end{equation}
because $F_j$ is real in our basis. So we can consider a mixture of the normalized real and imaginary parts, weighted by their squared norms, and the problem reduces completely to the real case.

\section{Monogamy relations from anticommuting sets}
\label{appmon}

We can write a general CHSH form
\begin{equation}
\mathcal S=\epsilon_{11}A_\pm B_1+\epsilon_{22} A_\mp B_2
\end{equation}
or
\begin{equation}
|2\bar{\mathcal S}|=|\langle A_\pm (B_++B_-)\rangle+\epsilon'\langle A_\mp(B_+-B_-)\rangle|
\end{equation}
with the arbitrary sign $\epsilon'=\pm 1$.
Then
\begin{align}
&|\bar{\mathcal S}|/2\leq \nonumber\\
&|c_Ac_B \langle ZZ\rangle|+|c_As_B\langle ZX\rangle|+|s_Ac_B\langle XZ\rangle|+|s_As_B\langle XX\rangle|
\end{align}

with $c_{A,B}=\cos\phi_{A,B}$, $s_{A,B}=\sin\phi_{A,B}$
By Cauchy-Schwarz inequality $(V\cdot W)^2\leq V^2W^2$, applied to the vectors
\begin{align}
&V=(c_Ac_B,c_As_B,s_Ac_B,s_As_B),\nonumber\\
&W=(\langle ZZ\rangle,\langle ZX\rangle,\langle XZ\rangle,\langle XX\rangle)
\end{align}
we have \cite{zuk}
\begin{equation}
\bar{\mathcal S}^2/4\leq \langle ZZ\rangle^2+\langle ZX\rangle^2+\langle XZ\rangle^2+\langle XX\rangle^2
\end{equation}
since $V^2=1$.
For a set of anticommuting dichotomic observables $F_j$ such that $F_j^2=1$ and $F_iF_j=-F_jF_i$ for $i\neq j$
we have \cite{kurz}
\begin{align}
&0\leq \left\langle\left(\sum_j\bar{F}_j F_j-\Bigl(\sum_j\bar{F}_j^2\Bigr)I\right)^2\right\rangle\nonumber\\
&=\sum_j\bar{F}_j^2-\left(\sum_j\bar{F}_j^2\right)^2\nonumber\\
&=\sum_j\bar{F}_j^2\left(1-\sum_j\bar{F}_j^2\right)
\end{align}
which leads to
\begin{equation}
\sum_j\bar{F}_j^2\leq 1
\end{equation}
Applying it to $F_1=ZZ$, $F_2=ZX$,
\begin{equation}
\langle ZZ\rangle^2+ \langle ZX\rangle^2\leq 1
\end{equation}
and similarly $\langle XZ\rangle^2+ \langle XX\rangle^2\leq 1$ for $F_1=XZ$, $F_2=XX$,
giving finally  $\bar{\mathcal S}^2\leq 8$ (Tsirelson bound \cite{tsirel}).
We have proven it for a particular choice of $\phi_A$, $\phi_B$ but 
\begin{equation}
\bar{\mathcal S}=\sum_{AB}p_{AB}\bar{\mathcal S}_{AB}
\end{equation}
for $AB$ denoting combinations of eigenspaces of $A_\pm$ and $B_\pm$ and 
\begin{equation}
\rho\to \sum_{AB}p_{AB}\rho_{AB}
\end{equation}
with probabilities $p_{AB}=\mathrm{Tr}_{AB}\rho$ that of occurring a particular combination $AB$ and scaling to $\mathrm{Tr}\rho_{AB}=1$
We can ignore all off-diagonal terms of $\rho$, i.e. connecting different eigenspaces, because they do not contribute to the correlation $\mathcal S$.
Then from simple classical inequality for a variance of $\mathcal S_{AB}$, we have
\begin{equation}
\bar{\mathcal S}^2=\left(\sum_{AB}p_{AB}\bar{\mathcal S}_{AB}\right)^2\leq \sum_{AB}p_{AB}\bar{\mathcal S}_{AB}^2\leq 8\sum_{AB}p_{AB}=8
\end{equation}
as $\sum_{AB}p_{AB}=1$.

In the 3-party $ABC$ scenario there is a trade-off of the violations, say
\begin{equation}
\mathcal S_{AB}=\sum_{ij}\epsilon_{ij}\langle A_iB_j\rangle,\;
\mathcal S_{AC}=\sum_{ij}\epsilon'_{ij}\langle A_iC_j\rangle,
\end{equation}
The choice of signs $\epsilon$ and $\epsilon'$ can be different.
We can use similar arguments as for single Bell inequality to prove
\begin{align}
&(\bar{\mathcal S}_{AB}^2+\bar{\mathcal S}_{AC}^2)/4\leq \nonumber\\
&\langle ZZI\rangle^2+\langle ZXI\rangle^2+\langle XZI\rangle^2+\langle XXI\rangle^2\nonumber\\
&+\langle ZIZ\rangle^2+\langle ZIX\rangle^2+\langle XIZ\rangle^2+\langle XIX\rangle^2
\end{align}
where $I$ stands for the identity for the omitted party. Taking $F_1=ZZI$, $F_2=ZXI$, $F_3=XIZ$, $F_4=XIX$
and $F_1=XZI$, $F_2=XXI$, $F_3=ZIZ$, $F_4=ZIX$ (two anticommuting sets) we get
\begin{equation}
\bar{\mathcal S}_{AB}^2+\bar{\mathcal S}_{AC}^2\leq 8
\end{equation}
valid
for arbitrary dichotomic observables and states by Jordan's lemma.

For 3-party setup we have Mermin inequality \cite{mermin}
\begin{equation}
\mathcal M_{ABC}=A_1B_1C_2+A_1B_2C_1+A_2B_1C_1- A_2B_2C_2
\end{equation}
with the same assumptions on observables as in Bell case.
In classical scenario $\mathcal M\leq 2$ for $A_i,B_j,C_k=\pm 1$.
In the quantum case, we can proceed analogously, translating to $A_\pm$, $B_\pm$, $C_\pm$.
It gives
\begin{align}
&4\mathcal M_{ABC}=\nonumber\\
& A_+B_+C_++ A_-B_+C_++ A_+B_-C_++ A_+B_+C_-\nonumber\\
&-A_+B_-C_--A_-B_+C_-- A_-B_-C_+-A_-B_-C_-
\end{align}
Restricting to eigenspaces of $A_\pm$, $B_\pm$, $C_\pm$, we get
\begin{align}
&\bar{\mathcal M}/2=c_Ac_Bc_C\langle ZZZ\rangle+s_Ac_Bc_C\langle XZZ\rangle\nonumber\\
&+c_As_Bc_C\langle ZXZ\rangle+c_Ac_Bs_C\langle ZZX\rangle\nonumber\\
&-c_As_Bs_C\langle ZXX\rangle-s_Ac_Bs_C\langle XZX\rangle\nonumber\\
&-s_As_Bc_C\langle XXZ\rangle-s_As_Bs_C\langle XXX\rangle
\end{align}
As in the Bell inequality, by Cauchy-Schwarz inequality applied to
\begin{equation}
V=\begin{pmatrix}
c_Ac_Bc_C\\
s_Ac_Bc_C\\
c_As_Bc_C\\
c_Ac_Bs_C\\
-c_As_Bs_C\\
-s_Ac_Bs_C\\
-s_As_Bc_C\\
-s_As_Bs_C
\end{pmatrix},\;
W=\begin{pmatrix}
\langle ZZZ\rangle\\
\langle XZZ\rangle\\
\langle ZXZ\rangle\\
\langle ZZX\rangle\\
\langle ZXX\rangle\\
\langle XZX\rangle\\
\langle XXZ\rangle\\
\langle XXX\rangle
\end{pmatrix}
\end{equation}
we get
\begin{align}
&\bar{\mathcal M}^2/4\leq \langle ZZZ\rangle^2+\langle XZZ\rangle^2\nonumber\\
&+\langle ZXZ\rangle^2+\langle ZZX\rangle^2+\langle ZXX\rangle^2\nonumber\\
&+\langle XZX\rangle^2+\langle XXZ\rangle^2+\langle XXX\rangle^2
\end{align}
Note also that the inequality remains valid under exchange of roles $A_1\leftrightarrow A_2$ which changes $A_-\to -A_-$, and changing signs
$A_2\to -A_2$ equivalent to $A_+\leftrightarrow A_-$ as the bound takes all 8 combinations of $A_\pm$, $B_\pm$, $C_\pm$.
We can immediately find the quantum bound by taking anticommuting sets $F=(ZZZ,XXX)$, $F=(XZZ,ZXX)$, $F=(ZXZ,XZX)$, $F=(ZZX,XXZ)$
\begin{equation}
\bar{\mathcal M}^2\leq 16
\end{equation}
which agrees with the $|\bar{\mathcal M}|\leq 4$ result. Again the inequality extends to arbitrary states by Jordan's lemma
 We assume that all correlations involve the same sets of 2 observables, but may differ by sign changes or relabeling.
Then we can use the individual bound to get
\begin{align}
&(\bar{\mathcal M}_{ABC}^2+\bar{\mathcal M}_{BCD}^2+\bar{\mathcal M}_{CDA}^2+\bar{\mathcal M}_{DAB}^2)/4\leq\nonumber\\
&\sum_{ABC=ZX}\langle ABCI\rangle^2+\sum_{BCD=ZX}\langle IBCD\rangle^2\nonumber\\
&+\sum_{CDA=ZX}\langle AICD\rangle^2+\sum_{DAB=ZX}\langle ABID\rangle^2
\end{align}
The right-hand side consists of $32$ terms with $ABCD$ substituted by $X,Z$ and the missing party replaced by identity $I$.
Start with the following set of eight pairwise anticommuting strings
\begin{equation}
F=\begin{pmatrix}
XXXI\\
ZZZI\\
ZXIX\\
XZIZ\\
XIZX\\ ZIXZ\\
IZXX\\ IXZZ
\end{pmatrix}
\end{equation}
Apply all $2^4=16$ independent local exchanges $X\leftrightarrow Z$
to this set, one choice at each party. Each transformed set remains
pairwise anticommuting. Among the 128 occurrences, each of the 32
three-body strings occurs exactly four times.
Therefore, distributing the terms equally we get \cite{kurz}
\begin{equation}
\bar{\mathcal M}_{ABC}^2+\bar{\mathcal M}_{BCD}^2+\bar{\mathcal M}_{CDA}^2+\bar{\mathcal M}_{DAB}^2\leq 16
\end{equation}
The bound extends to arbitrary states by Jordan's lemma.

  \begin{table}
 \begin{tabular}{cc}
 \toprule
 $i$&$f_i$\\
 \midrule 
14& 16\\
13& 1512\\
12& -18927\\
11& + 789321\\
10& -  11382144\\
9&  - 8866028\\
8&  + 1092699070\\
7& + 149884014\\
6&  - 28458491172\\
5&  + 21947086716\\
4& + 220984337433\\
3& - 534906776367\\
2& + 1166578934700\\
1&  - 1375056878080\\
0&  -1202033209344\\
\bottomrule
\end{tabular}
\caption{Table of coefficients of the polynomial $f(x)$}\label{fs5}
\end{table}

 \begin{table}
 \begin{tabular}{cc}
 \toprule
 $i$&$h_i$\\
 \midrule 
14& 24\\
13&  + 176\\
12&  - 1618\\
11&  + 6402\\
10& + 34481\\
9& +  73327\\
8&  + 159636\\
7& + 320904\\
6& + 394183\\
5&  + 253595\\
4& + 66156\\
3& - 10468\\
2& - 17950\\
1&  - 10080\\
0& -1568\\
 \bottomrule
\end{tabular}
\caption{Table of coefficients of the polynomial $h(x)$}\label{us5}
\end{table}

\section{\texorpdfstring{$\mathcal W^5$ polygamy}{W5 polygamy}}
\label{appw5}
\begin{widetext}

For real product states across $AB|CD$, the following expectation-value identity holds. The difference of the unaveraged polynomials is antisymmetric under partial transpose on $AB$, whereas the product density matrix is invariant under that transpose. Thus its expectation vanishes.

\begin{align}
 &16(12-\mathcal W^5_{ABC}- \mathcal W^5_{BCD}- \mathcal W^5_{CDA}- \mathcal W^5_{DAB})
 =\nonumber\\
&\bigl\langle(A_1-D_1+A_2C_1-B_2C_1+A_1C_2+B_1C_2-A_2D_1-B_2D_1+A_1D_2-B_1D_2+A_2B_2-A_2C_2+B_2D_2-C_2D_2\nonumber\\
&-A_1B_1C_1-A_2B_1C_1-A_1B_1C_2+A_2B_1D_1+B_1C_1D_1+B_2C_1D_1-A_2B_2D_2-A_1C_1D_2+B_1C_1D_2+A_2C_2D_2)^2+\nonumber\\
&
(B_1-C_1-A_2C_1-B_2C_1-A_1C_2+B_1C_2-A_2D_1+B_2D_1+A_1D_2+B_1D_2+A_2B_2+A_2C_2-B_2D_2-C_2D_2\nonumber\\
&
-B_1C_2D_1+A_1B_2C_1-A_2B_2C_2-A_1B_1D_1-A_1B_2D_1+A_1C_1D_1+A_2C_1D_1
+A_1C_2D_1-A_1B_1D_2+B_2C_2D_2)^2+\nonumber\\
&(1+A_2)^2(1-B_1)^2(1-C_1)^2(1+D_2)^2+(1-A_1)^2(1+B_2)^2(1+C_2)^2(1-D_1)^2+(A_1+A_2+B_1+C_1+D_1+D_2+\nonumber\\
&A_2B_1-A_2B_2+B_1C_1+B_2C_1+A_1C_2-A_2C_2+B_1C_2+A_2D_1+B_2D_1+A_1D_2-A_2D_2-B_2D_2+C_1D_2-C_2D_2\nonumber\\
&-A_1B_1C_1-A_1B_1C_2-A_2B_1D_1-B_1C_1D_1-B_2C_1D_1-A_2B_1D_2+A_2B_2D_2\nonumber\\
&-A_1C_1D_2-A_2C_1D_2+A_2C_2D_2+A_2B_1C_1D_2-5)^2+\nonumber\\
&
(A_1+B_1+B_2+C_1+C_2+D_1+A_1B_2-A_2B_2+A_2C_1+B_2C_1-A_2C_2+B_1C_2-B_2C_2+A_1D_1+A_2D_1\nonumber\\
&+C_2D_1+A_1D_2+B_1D_2-B_2D_2-C_2D_2-A_1B_2C_1-A_1B_2C_2+A_2B_2C_2-A_1B_1D_1-A_1C_1D_1\nonumber\\
&-A_2C_1D_1-B_1C_2D_1-B_2C_2D_1-A_1B_1D_2+B_2C_2D_2+A_1B_2C_2D_1-5)^2\bigr\rangle\label{sq5}
\end{align}
\end{widetext}

  \begin{table}
 \begin{tabular}{cc}
 \toprule
 $i$&$\bar{f}_i$\\
 \midrule 
11& 5\\
10&+ 4084\\
9&  - 1944208\\
8&  + 65288128\\
7& - 129175040\\
6&  -27052701696\\
5& + 334853234688\\
4& - 404496613376\\
3& - 714946838528\\
2& - 70762727473152\\
1& + 397032928837632\\
0& -451298709209088\\
\bottomrule
\end{tabular}
\caption{Table of coefficients of the polynomial $\bar{f}(x)$}\label{bs5a}
\end{table}

\section{MABK polygamy with 5 parties}
\label{appm5}

 \begin{table}
 \begin{tabular}{cc}
 \toprule
 $i$&$w_i$\\
 \midrule
 0&
-1348993610803265103608453857408879888957440000\\
1&
 -1862027439368874196315264434261591732309196800\\
2& - 544284001638028890174731129868245406993874944\\
3& + 186275487542605903899950936112300629503770624\\
4& + 77959192080112434517604337992982574393196544\\
5& - 7868343165721543276454266995480688022519808\\
6& - 3989371250800960166613091601631171991044096\\
7& + 199529672562320553165519922230478334066688\\
8& + 98889264682092492046683976496751197552640\\
9& - 3678275235052666587013740732120112824320\\
10& - 1260009487394625599271144573929300951040\\
11& + 48536148686365447944470643365774360576\\
12& + 8881547702420601937554739451698085888\\
13& - 414688510707196350419008010900733952\\
14& - 36118216572329937625551112801091584\\
15& + 2145829744705972550957979154776064\\
16& + 86484244238501463508879964897280\\
17& - 6565142087832070154994336137216\\
18& - 139625279073419759521651228672\\
19& + 11025099934518830973740646400\\
20& + 214752412409271585309982720\\
21& - 6252894111076475000061952\\
22& - 193016288946419973226496\\
23& - 3666655092253282795520\\
24& - 15694520666142539776\\
25& + 1897325745878597632\\
26& + 7270519365566464\\
27& + 136612730171392\\
28& + 17412201573376\\
29& - 105543443968\\
30& - 4897300608\\
31& + 18494480\\
32& + 431540\\
33& + 976\\
34& + 1\\
\bottomrule
\end{tabular}
\caption{Table of coefficients of the polynomial $w(x)$}\label{m5a}
\end{table}

 \begin{table}
 \begin{tabular}{cc}
 \toprule
 $i$&$v_i$\\
 \midrule 
0& -17496\\
1& - 23328\\
2& + 2562192\\
3& + 8963460\\
4& - 87124221\\
5& - 1085663142\\
6& - 3087408519\\
7& + 49132536291\\
8& + 257131065951\\
9& - 710011046268\\
10& - 5241750399930\\
11& - 3142358882574\\
12& + 30178598302403\\
13& + 83015723329318\\
14& + 98985283300107\\
15& + 99236167396045\\
16& + 197062587648471\\
17& + 343678745669288\\
18& + 258504880553540\\
19& - 143178551046512\\
20& - 486985987297423\\
21& - 408513678199802\\
22& - 53840529633217\\
23& + 177521434339437\\
24& + 157068862444629\\
25& + 46016697356388\\
26& - 13512271335690\\
27& - 16867623183462\\
28& - 6488272278711\\
29& - 1192944925638\\
30& - 64464708195\\
31& + 10272994371\\
32& + 955436877\\
33& - 60781104\\
34& + 944784\\
\bottomrule
\end{tabular}
\caption{Table of coefficients of the polynomial $v(x)$}\label{mu5a}
\end{table} 

 \begin{table}
 \begin{tabular}{cc}
 \toprule
 $i$&$\bar{w}_i$\\
 \midrule
0&-1348993610803265103608453857408879888957440000\\
1& - 1862027439368874196315264434261591732309196800\\
2&  - 544284001638028890174731129868245406993874944\\
3& + 186275487542605903899950936112300629503770624\\
4& + 77959192080112434517604337992982574393196544\\
5& - 7868343165721543276454266995480688022519808\\
6& - 3989371250800960166613091601631171991044096\\
7& + 199529672562320553165519922230478334066688\\
8& + 98889264682092492046683976496751197552640\\
9& - 3678275235052666587013740732120112824320\\
10& - 1260009487394625599271144573929300951040\\
11& + 48536148686365447944470643365774360576\\
12& + 8881547702420601937554739451698085888\\
13& - 414688510707196350419008010900733952\\
14& - 36118216572329937625551112801091584\\
15& + 2145829744705972550957979154776064\\
16& + 86484244238501463508879964897280\\
17& - 6565142087832070154994336137216\\
18& - 139625279073419759521651228672\\
19& + 11025099934518830973740646400\\
20& + 214752412409271585309982720\\
21& - 6252894111076475000061952\\
22& - 193016288946419973226496\\
23& - 3666655092253282795520\\
24& - 15694520666142539776\\
25& + 1897325745878597632\\
26& + 7270519365566464\\
27& + 136612730171392\\
28& + 17412201573376\\
29& - 105543443968\\
30& - 4897300608\\
31& + 18494480\\
32& + 431540\\
33& + 976\\
34& + 1\\
\bottomrule
\end{tabular}
\caption{Table of coefficients of the polynomial $\bar{w}(x)$}\label{m5b}
\end{table} 

\begin{table}
 \begin{tabular}{cc}
 \toprule
 $i$&$\bar{v}_i$\\
 \midrule 
0& -17496\\
1& - 23328\\
2& + 2562192\\
3& + 8963460\\
4& - 87124221\\
5& - 1085663142\\
6& - 3087408519\\
7& + 49132536291\\
8& + 257131065951\\
9& - 710011046268\\
10& - 5241750399930\\
11& - 3142358882574\\
12& + 30178598302403\\
13& + 83015723329318\\
14& + 98985283300107\\
15& + 99236167396045\\
16& + 197062587648471\\
17& + 343678745669288\\
18& + 258504880553540\\
19& - 143178551046512\\
20& - 486985987297423\\
21& - 408513678199802\\
22& - 53840529633217\\
23& + 177521434339437\\
24& + 157068862444629\\
25& + 46016697356388\\
26& - 13512271335690\\
27& - 16867623183462\\
28& - 6488272278711\\
29& - 1192944925638\\
30& - 64464708195\\
31& + 10272994371\\
32& + 955436877\\
33& - 60781104\\
34& + 944784\\
 \bottomrule
\end{tabular}
\caption{Table of coefficients of the polynomial $\bar{v}(x)$}\label{mu5b}
\end{table} 
 
 The maximum for the biseparable $ABC|DE$ is
\begin{equation}
8+4\sqrt{10}\simeq 20.649110640673517328
\end{equation}
for $4\sqrt{5+\sqrt{10}}|\psi_{ABC}\rangle$ of the form
\begin{align}
&(|100\rangle+|010\rangle+|001\rangle+(1+\sqrt{10})(|011\rangle+|101\rangle+|110\rangle)\nonumber\\
&+3|111\rangle-(5+\sqrt{10})|000\rangle
\end{align}
and $A_1=B_1=C_1=Z$, $A_2=B_2=C_2=X$, $D_j=E_j=1$. It is therefore also the upper triseparable $ABC|D|E$ bound.
However, one can get larger violation by the lower bound
$-21.0787936983$
being the smallest real root of
\begin{equation}
-71424 - 23808x - 1504x^2 + 80x^3 + 5x^4=0
\end{equation}
and with $E_1=E_2=D_1=D_2=-1$, 
and 
\begin{align}
&A_1=((1-u)Z+2\sqrt{u}X)/(u+1)\nonumber\\
&A_2=((1-u)Z-2\sqrt{u}X)/(u+1)\nonumber\\
&B_1=((1 - u)Z + 2\sqrt{u}X)/(u + 1),\nonumber\\
&B_2=((1 - u)Z - 2\sqrt{u}X)/(u + 1),\nonumber\\
&C_1=((u - 1)Z + 2\sqrt{u}X)/(u + 1),\nonumber\\
&C_2=((u - 1)Z - 2\sqrt{u}X)/(u + 1)
\end{align}
with $u\simeq 0.107142854395$,
being the second smallest root of
\begin{align}
&1 - 56x + 588x^2 - 1608x^3 + 1894x^4\nonumber\\
& - 1608x^5 + 588x^6 - 56x^7 + x^8=0
\end{align}
   
The triseparable $AB|CD|E$ and the fourseparable case gives the upper bound only classical $20$
but the lower bound
is $-20.2624403701$ which
is the smallest root of
\begin{equation}
-1000704 + 238592x - 8544x^2 - 576x^3 + 27x^4=0
\end{equation}
for the fourseparable state with only $AB$ entangled and $C_i=D_i=E_i=1$ for $i=1,2$
with
\begin{align}
&A_1=((u - 1)Z - 2\sqrt{u}X)/(u + 1),\nonumber\\
&A_2= ((u - 1)Z + 2\sqrt{u}X)/(u + 1),\nonumber\\
&B_1=((u - 1)Z + 2\sqrt{u}X)/(u + 1),\nonumber\\
&B_2=((u - 1)Z - 2\sqrt{u}X)/(u + 1)
\end{align}
with $u\simeq 0.082007887975$
being the smallest root of
\begin{align}
&5 - 112x + 821x^2 - 2756x^3 + 4276x^4\nonumber\\
& - 2756x^5 +  821x^6 - 112x^7 + 5x^8=0
\end{align}

In the complete classical case $A_i=B_i=C_i=D_i=E_i=1$ we get $-20$ ($-4$ per party).
The positive result $+20$ is obtained e.g. if $D_2=E_2=-1$.

\section{\texorpdfstring{$AB|CD$ biseparable bound of $\mathcal M_{ABCD}$}{AB|CD biseparable bound of M(ABCD)}}
\label{appms}

We can write $\mathcal M_{ABCD}$ in the form
\begin{align}
&(A_1B_1-A_2B_2)(C_1D_1+C_1D_2+C_2D_1-C_2D_2)\nonumber\\
&+(A_1B_2+A_2B_1)(C_1D_1-C_1D_2-C_2D_1-C_2D_2)
\end{align}
For the biseparable case we can use Cauchy-Schwarz inequality

\begin{align}
&\bar{\mathcal M}^2_{ABCD}\leq (\langle A_1B_1-A_2B_2\rangle^2+ \langle A_1B_2+A_2B_1\rangle^2)\nonumber\\
&\times(\langle C_1D_1+C_1D_2+C_2D_1-C_2D_2\rangle^2+\nonumber\\
&\langle C_1D_1-C_1D_2-C_2D_1-C_2D_2\rangle^2)
\end{align}
However
\begin{align}
&\langle C_1D_1+C_1D_2+C_2D_1-C_2D_2\rangle^2+\nonumber\\
&\langle C_1D_1-C_1D_2-C_2D_1-C_2D_2\rangle^2\nonumber\\
&=2(\langle C_1D_1-C_2D_2\rangle^2+\langle C_1D_2+C_2D_1\rangle^2)
\end{align}
from paralelloid identity.
Applying now Uffink inequality \cite{uff} (see Appendix \ref{appuff}) to both expressions we arrive at $\bar{\mathcal M}^2_{ABCD}\leq 32$.
Note that the maximum needs a tradeoff of entanglement. 
Uffink inequality is saturated at product and maximally entangled states while the observables are orthogonal, e.g.
$A_1=Z$, $A_2=X$, $B_1=Z\cos\phi+X\sin\phi$, $B_2=Z\sin\phi-X\cos\phi$
in the state $(|00\rangle+|11\rangle)/\sqrt{2}$ giving
\begin{equation}
\langle A_1B_1-A_2B_2\rangle=2\cos\phi,\;
\langle A_1B_2+A_2B_1\rangle=2\sin\phi
\end{equation} 
 The nontrivial saturation is when
both $AB$ and $CD$ are entangled. Taking identical entangled states and replacing $AB\to CD$ with $\phi\to\phi'$
we get
\begin{equation}
\bar{\mathcal M}_{ABCD}=4((\cos\phi'+\sin\phi')\cos\phi+(\cos\phi'-\sin\phi')\sin\phi)
\end{equation}
which is maximized  at $4\sqrt{2}$ whenever $\phi'=\pi/4-\phi$.

\section{Uffink inequality}
\label{appuff}
The  Uffink inequality \cite{uff} reads
\begin{equation}
\langle A_1B_1-A_2B_2\rangle^2+ \langle A_1B_2+A_2B_1\rangle^2\leq 4
\end{equation}
or more generally
\begin{equation}
|\langle A^{(1)}A^{(2)}A^{(3)}\cdots A^{(n)}\rangle|\leq 2^{n-1}
\end{equation}
for $n\geq2$ and $A^j=A^j_1+iA^j_2$ for each party $j$.
To prove it, again due to Jordan's lemma we can focus on $A^{(j)}_1=Z^{(j)}\cos\phi_j+X^{(j)}\sin\phi_j$,
$A^{(j)}_2=Z^{(j)}\cos\phi_j-X^{(j)}\sin\phi_j$
The inequality then reads
\begin{equation}
|\langle \prod_j ((1+i)Z^{(j)}\cos\phi_j+(1-i)X^{(j)}\sin\phi_j)\rangle|\leq 2^{n-1}
\end{equation}
Factoring out $1+i$, we simplify it to
\begin{equation}
|\langle \prod_j (Z^{(j)}\cos\phi_j-iX^{(j)}\sin\phi_j)\rangle|\leq 2^{n/2-1}
\end{equation}
now let us use a different representation $Z\to X$, $X\to Y$. Then we get $X^j\cos\phi_j-iY^j\sin\phi_j$
which has the block structure
\begin{equation}
\begin{pmatrix}
0&\cos\phi_j-\sin\phi_j\\
\cos\phi_j+\sin\phi_j&0
\end{pmatrix}=\sqrt{2}\begin{pmatrix}
0&\cos\tilde{\phi}_j\\
\sin\tilde{\phi}_j&0
\end{pmatrix}
\end{equation}
for $\tilde{\phi}_j=\phi_j+\pi/4$. We get essentially the whole space split into block of two states with flipped
$|0\rangle\leftrightarrow |1\rangle$.
Let us denote each state by a bitstring $K=k_1k_2\dots k_n$, $k_j=0,1$.
Each such block has the contribution
\begin{equation}
\psi_K\phi^\ast_{\bar{K}}\prod_j \alpha_j(k_j)+\psi_K^\ast\phi_{\bar{K}}\prod_j \alpha_j(\bar{k}_j)
\end{equation}
where $\alpha_j(0)=\cos\tilde{\phi}_j$, $\alpha_j(1)=\sin\tilde{\phi}_j$ and $\bar{K}$ means complete spin flipping 
($0\leftrightarrow 1$)
Each individual term is maximized when the product $\psi_K\phi^\ast_{\bar{K}}$ is real/imaginary and maximal or minimal
with respect to the constraint $|\psi_K|^2+|\phi_{\bar{K}}|^2=p_K$.
Therefore $\psi_K=\pm \phi_{\bar{K}}$  or  $\psi_K=\pm i\phi_{\bar{K}}$ and all the terms have the form
\begin{equation}
p_K\left(\prod_j |\alpha_j(k_j)|+\prod_j |\alpha_j(\bar{k}_j)|\right)/2
\end{equation}
By convexity the maximum is attained at a boundary point, namely the block with the largest value of
\begin{equation}
\prod_j |\alpha_j(k_j)|+\prod_j |\alpha_j(\bar{k}_j)|.
\end{equation}
By Cauchy-Schwarz inequality applied to any of the terms e.g. $j=1$
\begin{align}
&\left(\prod_j |\alpha_j(k_j)|+\prod_j |\alpha_j(\bar{k}_j)|\right)^2\leq\nonumber\\
&(\cos^2\tilde{\phi}_1+\sin^2\tilde{\phi}_1) \left(\prod_{j>1} \alpha^2_j(k_j)+\prod_{j>1}\alpha^2_j(\bar{k}_j)\right).
\end{align}
The last expression has all $\alpha_j^2\leq 1$ for $j>2$ so
\begin{equation}
\prod_{j>1} \alpha^2_j(k_j)+\prod_{j>1}\alpha^2_j(\bar{k}_j)\leq \alpha^2_2(k_2)+\alpha^2_2(\bar{k}_2)=1
\end{equation}
We end up with the bound $1/2$ for $p_K=1$.

\section{\texorpdfstring{Binary tree solved by PSD $\Gamma$}{Binary tree solved by PSD Gamma}}
\label{appmm}

Here we present the bounds on sums of squares solved by our method based on partial transpose.
The Bell monogamy $BC-AC$ \cite{toner}
\begin{equation}
\sum_j\langle F_j\rangle^2\leq 2,
\end{equation}
requires 8 Pauli strings $IXX$, $IXZ$, $IZX$, $IZZ$, $XIX$, $XIZ$, $ZIX$, $ZIZ$, in $ABC$ representation,
The PSD $\Gamma$ is realized taking 
\begin{equation}
\Gamma=\begin{pmatrix}
1&-1&1&-1&-1&1&-1&1\\
-1&1&-1&1&1&-1&1&-1\\
1&-1&1&-1&-1&1&-1&1\\
-1&1&-1&1&1&-1&1&-1\\
-1&1&-1&1&1&-1&1&-1\\
1&-1&1&-1&-1&1&-1&1\\
-1&1&-1&1&1&-1&1&-1\\
1&-1&1&-1&-1&1&-1&1
\end{pmatrix}
\end{equation}
in this basis.
It agrees with the general rule that $\Gamma_{ij}=-1$ for commuting pairs.
and $\Gamma_{ij}=+1$ for anticommuting pairs with $I$ at different indices.
The matrix is PSD and has rank one: $\Gamma=vv^T$ with $v=(1,-1,1,-1,-1,1,-1,1)^T$. Its nonzero eigenvalue is $8$, so it is eight times a rank-one orthogonal projector.


The monogamy is generalized  to binary tree geometry \cite{kurz} with parties indexed by tree prefixes i.e.
$A$ (root), $A_0$ and $A_1$ (sons), $A_{00}$, $A_{01}$, $A_{10}$, $A_{11}$ (grandsons) etc.
We restrict to the straight ancestry line, i.e the Pauli string over the growing suffix, e.g.
\begin{equation}
AA_0A_{01}A_{011}A_{0110}A_{01101}\cdots
\end{equation}
To formalize it, each string is uniquely identified by its full index $x=x_1x_2\cdots x_N$ for $N+1$ depth.
We define the prefix $x^{(k)}\equiv x_1x_2\cdots x_k$ including empty one $x^{(0)}$ and the full one $x^{(N)}=x$
then the string consists of
\begin{equation}
A_{x^{(0)}}A_{x^{(1)}}\cdots A_{x^{(N)}}
\end{equation}
with each $A_x$ replaced by $X$ or $Z$ in the appropriate space. The full space consists of $2^{N+1}-1$ qubits.
As the strings are essentially enumerated by $x$ and $2$ possibilities at each node of the path,
 the space of $\Gamma$ is of dimension $2^{2N+1}$. The index carries two objects, $p=p_1p_2\cdots p_N$ - the path of $2^N$ combinations
 and $v=v_0v_1v_2\cdots v_N$ - the value (choices) of the Pauli operators  $X,Z$ replaced by $0,1$ along the path.
 
 We shall prove that
\begin{equation}
\sum_{p,v}\langle F_{pv}\rangle^2\leq 2^N
 \end{equation}
This time our construction of $\Gamma$ is as follows, $\Gamma=\Delta-J$ where
$J$ is the matrix of $1$ i.e. $J_{pv,rw}=1$ always and by default it covers all commuting pairs. 
The contribution $-J$ has rank one and one negative eigenvalue.
On the other hand we set $\Delta$ to be a sparse matrix with
\begin{equation}
\Delta_{pv,rw}=
\left\{\begin{array}{l}
2^N\mbox{ for }pv=rw\\
2^N\mbox{ for }p\oplus r=0^N,\; v\oplus w=0^{N}1\\
2^{k+1}\mbox{ for } 0\leq k\leq N-1\mbox{ and }\\
 (p\oplus r)^{(k+1)}=0^k1,\; (v\oplus w)^{(k)}=0^k1\\
0\mbox{ otherwise }
\end{array}
\right.
\end{equation}
with $x\oplus y$  for  $x,y=0,1$ being the binary (XOR) addition, i.e. 
\begin{equation}
x\oplus y=\left\{\begin{array}{l}
0\mbox{ if }x=y,\\
1\mbox{ if  }x\neq y
\end{array}\right..
\end{equation}
Here $v^{(k)}=v_0\cdots v_k$. The notation should be understood as follows.
Diagonal elements of $\Delta$ keep the value $2^N$.
Elements on the same paths but differing exactly on the last bit of the value are assigned $2^N$, too.
Now whenever the paths split at $k$, i.e. the first differing path bit is at position $k+1$
we assign $2^{k+1}$ only if the values $v$ and $w$ differ at position $k$ and while identical before $k$. Note that the rule affects only anticommuting 
strings, as the suffixes beyond the splitting point are irrelevant. In this way, for $k=N-1$ we get again $2^N$
which is correct. On the other end, at $k=0$ we get $2$ for the paths differing right from the root
and values differing on the root itself.

The goal of the proof is to show that $\Gamma$ is PSD.
We apply Walsh-Hadamard (WH) transform \cite{whtt,wh2} to $\Gamma$ in \emph{both} the value and path space.
The WH transform in $n$ qubit space is essentially application of Hadamard gate in each dimension.
Formally it takes a vector $\psi$ to $\psi'$,
\begin{equation}
\psi'_y=2^{-n/2}\sum_{x}(-1)^{x\cdot y}\psi_x
\end{equation}
with 
\begin{equation}
x\cdot y=\sum_k x_ky_k
\end{equation}
for the bitstrings $x=x_1x_2\cdots x_n$, $y=y_1y_2\cdots y_n$.
For matrices
\begin{equation}
M'_{st}=2^{-n}\sum_{xy}(-1)^{x\cdot s+y\cdot t}M_{xy}\label{whm}
\end{equation}

We can use a simple lemma. 
\begin{lemma}
\label{le1}
Let the matrix entries $C_{xy}$ with $x$, $y$ being the binary strings of length $L$, depend solely
on $x\oplus y$, i.e. the relative bitstring, $C_{x,y}\equiv C_{x\oplus y}$. Then the WH transform is diagonal.
\end{lemma}
\emph{Proof}.
Substituting $y=x\oplus z$ shift,
\begin{equation}
2^LC'_{st}=\sum_{xy}(-1)^{s\cdot x+ t\cdot y}C_{xy}=\sum_{xz}(-1)^{(s\oplus t)\cdot x+t\cdot z}C_z
\end{equation}
but 
\begin{equation}
\sum_x(-1)^{x\cdot s}=\left\{\begin{array}{l}
2^L\mbox{ for } s=0^L\\
0\mbox{ otherwise }
\end{array}
\right.
\end{equation}
because every bit $1$ of $s$ triggers $+1$ and $-1$ cancellation.
As a result $C'_{st}=0$ for $s\neq t$ and additionally
\begin{equation}
C'_{ss}=\sum_z(-1)^{s\cdot z}C_z.
\end{equation}
$\square$

In our case $L=2N+1$.
Since $\Delta_{pv,rw}$  and $J_{pv,rw}$ depend only on $pv \oplus rw$, Lemma 1 applies to them both
so $\Gamma' = \Delta' - J'$ is diagonal and PSD iff all diagonal entries are non-negative.
First, let us see what happens to $J$, whose elements do not depend on the relative string at all.
Whenever $s\neq 0^{2N+1}$, some its bit is $1$, say $s_k=1$. Then the sum will contain $z_k=0$ and $z_k=1$ and
both terms will cancel. The only nonvanishing term is for $s=0^{2N+1}$, i.e.
\begin{equation}
J'_{ss}=2^{2N+1}\delta_{s,0^{2N+1}}
\end{equation}
Now the diagonal part of $\Delta$ will remain unchanged, i.e. $2^{N}I^{2N+1}$.
Let us calculate the contribution of the other terms.
The first term is $\Delta^{(N)}_{pv}=2^N$ for $p=0^N$, $v=0^N1$,
which is transformed to $s=rw$ (respective bitstrings in the transformed space)
\begin{equation}
\Delta^{(N)\prime}_{rw}=2^N(-1)^{r\cdot 0^N+w\cdot 0^N1}=2^N(-1)^{w_N1}
\end{equation}
which is $+2^N$ for $w_N=0$ and $-2^N$ for $w_N=1$ and the other bits are arbitrary.
The second term is
$\Delta^{(N-1)}_{pv}=2^N$ for $p^{(N)}=0^{N-1}1$ , $v^{(N-1)}=0^{N-1}1$, which gives
\begin{equation}
\Delta^{(N-1)\prime}_{rw}=2^N\sum_{v_N}(-1)^{r\cdot 0^{N-1}1+w\cdot 0^{N-1}1v_N}.
\end{equation}
This sum cancels if $w_N=1$ because of the alternating bit $v_N$,
and the only nonzero case is $w_N=0$.
Therefore it becomes $+2^{N+1}$ if $r_N=w_{N-1}$ and $-2^{N+1}$ otherwise.
Now  $\Delta^{(k)}_{pv}=
2^{k+1}$ for $p^{(k+1)}=0^k1$ , $v^{(k)}=0^k1$, which gives
\begin{equation}
\Delta^{(k)\prime}_{rw}=2^{k+1}\sum_{\substack{p_{k+2}\cdots p_N\\
v_{k+1}\cdots v_N}}
(-1)^{r\cdot 0^k1p_{k+2}\cdots p_N+w\cdot 0^k1v_{k+1}\cdots v_N}
\end{equation}
Now the alternating free suffix bits $p_{k+2}\cdots p_N$, $v_{k+1}\cdots v_N$
will leave the only nonzero terms with $r_{k+2},\cdots r_N=0$ and $w_{k+1},\cdots, w_N=0$.
In this case, we get terms
$
+2^{2N-k}
$ when $r_{k+1}=w_k$ and 
$-2^{2N-k}$
otherwise, with the prefix bits arbitrary.

Let us check the overall values for $s=rw$, focusing only on cases with potentially negative contributions.
At $s=0^{2N+1}$, the total contribution reads
\begin{equation}
\Delta'_{ss}=2^N+2^N+2^{N+1}+\dots 2^{2N-k}+\dots + 2^{2N}=2^{2N+1}
\end{equation}
Therefore it cancels precisely with $J'$.
If $w_N=1$, we get
\begin{equation}
\Delta'_{ss}=2^N-2^N=0
\end{equation}
from the diagonal and $\Delta^{(N)\prime}$.
Note that other $\Delta^{(k)\prime}$ cannot reach $w_N=1$.
If $w_N=0$ and $r_N\neq w_{N-1}$ then at least one of them must be $1$.
We have then
\begin{equation}
\Delta'_{ss}=2^N+2^N-2^{N+1}=0
\end{equation}
from the diagonal, $\Delta^{(N)\prime}$, and $\Delta^{(N-1)\prime}$.
Now, $r_N=w_{N-1}=w_N=0$ but $r_{N-1}\neq w_{N-2}$, we get
\begin{equation}
\Delta'_{ss}=2^N+2^N+2^{N+1}-2^{N+2}=0
\end{equation}
from the diagonal, $\Delta^{(N)\prime}$, $\Delta^{(N-1)\prime}$,  $\Delta^{(N-2)\prime}$.
Iteratively, for $r_N=\cdots =r_{k+2}=w_N=\cdots =w_{k+1}=0$ and $r_{k+1}\neq w_k$ we get
\begin{equation}
\Delta'_{ss}=2^N+2^N+2^{N+1}+\dots+2^{2N-k-1} -2^{2N-k}=0
\end{equation}
from the diagonal, $\Delta^{(N)\prime}$, $\Delta^{(N-1)\prime}$, $\cdots$,   $\Delta^{(k)\prime}$.
The final term is at $k=0$, 
for $r_N=\cdots =r_2=w_N=\cdots =w_1=0$ and $r_1\neq w_0$ we get
\begin{equation}
\Delta'_{ss}=2^N+2^N+2^{N+1}+\dots+2^{2N-1} -2^{2N}=0
\end{equation}
from the diagonal, $\Delta^{(N)\prime}$, $\Delta^{(N-1)\prime}$, $\cdots$,   $\Delta^{(0)\prime}$.

\section{\texorpdfstring{Tradeoff between $n-1$ subsets}{Tradeoff between n-1 subsets}}
\label{appkk}

Here we prove the general trade-off relation (\ref{trade}).

\paragraph{A closed-form certificate for all $(N-1)$-party subsets.}
Put $M=N-1\geq3$ and $\Lambda=2^{M-1}$.
Index the Pauli strings by the omitted party and their $M$ binary
Pauli labels. When the omitted parties differ, set
$\Gamma_{ij}=-(-1)^{h(i,j)}$, where the Hamming distance is counted only
on the common active positions. In a diagonal block set
\[
 \Gamma_h=
 \begin{cases}
 2^{M-1}-1,&h=0,\\
 -1,&h>0\text{ even},\\
 2h-M,&h\text{ odd}.
 \end{cases}
\]
This satisfies the required commuting-entry constraints.
A normalized Walsh--Hadamard transform on the $M$ label bits makes
each diagonal block diagonal, with entries depending on the transformed
Hamming weight $w$ as follows:
\[
 \widehat\Gamma_w=
 \begin{cases}
 0,&w=0,1,M,\\
 2^M,&w=M-1,\\
 2^{M-1},&2\leq w\leq M-2.
 \end{cases}
\]
To check this without Krawtchouk matrices, write $\chi(s)=(-1)^{|s|}$
and use
\[
 \Gamma(s)=2^{M-1}\delta_{s,0}
 -\frac{1+\chi(s)}2
 -\frac{1-\chi(s)}2\sum_{a=1}^M(-1)^{s_a}.
\]
For different omitted parties $a,b$, the transformed off-diagonal block
has only one nonzero entry, equal to $-2^M$: it couples the coordinate
whose only zero active bit is $b$ in sector $a$ to its counterpart in
sector $b$. These coordinates form disjoint two-dimensional blocks
\[
 2^M\begin{pmatrix}1&-1\\-1&1\end{pmatrix}\succeq0.
\]
Every remaining coordinate is nonnegative and uncoupled. Therefore
$\Gamma\succeq0$, proving the bound $2^{N-2}$ in Eq.~\eqref{trade}.

The longer transform formulation follows.

To shorten notation we will concatenate symbols
\begin{equation}
\underbrace{00\cdots 0}_x\equiv 0^x,\;\underbrace{11\cdots 1}_x\equiv 1^x,
\end{equation}
Now the question of commutation is related to Hamming distance between the indices $i$ and $j$.
The Hamming distance is the number of differing bits (originally $X$ and $Z$) on positions that neither $i$ nor $j$
has $\bullet$ (originally $I$). 
For bitstrings $x=x_1x_2\cdots x_N$ and $y=y_1y_2\cdots y_N$ the Hamming distance is defined
\begin{equation}
h(x,y)=\sum_m x_m\oplus y_m.
\end{equation}
The strings commute at even Hamming distance  and anticommute at odd distance.
Whenever $i$ and $j$ have $\bullet$ at different positions, $\Gamma_{ij}=-(-1)^{h(i,j})$ excluding bits 
with $\bullet$ at $i$ or $j$.

We shall apply WH transform (\ref{whm}) to the space of $i$ excluding $\bullet$ ($I$).
The transform can be regarded as the tensor product of Hadamard transformations
\begin{equation}
H=\frac{Z+X}{\sqrt{2}}=\frac{1}{\sqrt{2}}\begin{pmatrix}
1&1\\
1&-1
\end{pmatrix}
\end{equation}
or $H_{ij}=(-1)^{ij}/\sqrt{2}$ in the basis $i,j=0,1$. 
Note that this time $X$ and $Z$ are not the original string Pauli matrices
but the operators in the $2^MN$ space, with $M=N-1$. Nevertheless, the standard properties of Hadamard matrix are useful,
$H^2=I$, $HXH=Z$, $HZH=X$, $HIH=I$. Let $H^{(k)}$ be the Hadamard matrix acting in the subspace of index $k$.
The full transform we need, reads
\begin{equation}
W=\bigoplus_k H^{(1)}\cdots H^{(k-1)}H^{(k+1)}\cdots H^{(N)}
\end{equation}
(again omitting tensor sign $\otimes$),
meaning that we split the $N\cdot 2^M$ into the direct sum $N$ sectors according to the position of $\bullet$ and apply Hadamard transformation
in each subspace. The total transformation is unitary and self-inverse.
Let us check how it affects $\gamma$ block $2^M\times 2^M$
 between sectors with $\bullet$ at positions $a$ and $b$. Without loss of generality, let $a=1$, $b=2$.
According to our assignment
$\Gamma_{ab}=-(I+X)(I-X)\cdots (I-X)$
because the sign of $\gamma$ does not depend on the bits at $a$ and $b$ but otherwise is it alternating.
Our transform takes it to
$
\Gamma_{ab}=-(I+Z)(I-Z)\cdots(I-Z)
$
which collapses to a single element $-2^M|\bullet 01^{M-1} \rangle\langle 0\bullet 1^{M-1}|$, and $0$ otherwise. More generally, the row (first index) bit at $b$ is $0$ and the column (second index)
bit at $a$ is $0$ and the common bits are $1$.

Now let us see what happens on the diagonal blocks, i.e. $a=b$. In this case we will shorten
$
\Gamma_{ij}\equiv \Gamma_k
$
where $k$ is the Hamming distance between $i$ and $j$ (now $\bullet$ is at the same position so we can simply ignore it).
we have $\Gamma_k=-1$ for the even $k$ with added $2^{M-1}$ at $k=0$.
For a moment let us set the remaining $\gamma_k=\mp 1$.
The block, say at $a=1$ reads
\begin{equation}
\Gamma_{aa}=2^{M-1}-(I\pm X)\cdots (I\pm X)
\end{equation}
so it will be transformed to
\begin{equation}
\Gamma'_{aa}=2^{M-1}-(I\pm Z)\cdots (I\pm Z)
\end{equation}
The striking feature is that both block collapses to the diagonal one,
$(I+ Z)\cdots (I+ Z)$ becomes $2^{M}$ times projection on the index $0^{M}$
(the smallest index)
and $(I- Z)\cdots (I- Z)$ becomes $2^{M}$ times projection on the (highest) index $1^{M}$.
Averaging them both we get
\begin{equation}
\Gamma'_{aa}=2^{M-1}(I-I_{0^{M}}-I_{1^{M}})
\end{equation}
Here we use $I_x=|x\rangle\langle x|$ to denote the projection onto the bitstring $x$.
Our goal is to have $\Gamma'$ PSD. The sole $aa$ block is already PSD but we have the off-diagonal coupling of
the entries $\bullet 01^{M-1}$ with the weight $2^{M}$ to counter. Therefore we need help from odd Hamming distances $k$.

Fortunately, for a matrix with entries depending solely on the Hamming distances, WH transform reduces to Kravchuk transform.
Let us take the matrix 
$
\Gamma_{xy}=\Gamma_{h(x,y)}
$.
From Lemma \ref{le1}, the WH transform takes it to
\begin{equation}
\Gamma'_{yz}=\sum_{sv}(-1)^{y\cdot s+z\cdot v}\Gamma_{h(s,v)}/2^{M}
\end{equation}
summing over all $2^{M}$ indices $x$ and $v$, and denoting the bitwise scalar (dot) product
$
x\cdot y=\sum_m x_my_m
$.
For $y\neq z$ we have $\Gamma'_{yz}=0$ substituting $v=s\oplus x$, i.e.
bitwise $XOR$  (it leaves the summation range invariant). Since the exponent is always modulo $2$, and due to distributivity
$(x\oplus y)z=xz\oplus yz$, we have
\begin{equation}
\Gamma'_{yz}=\sum_{sx}(-1)^{(y\oplus z)\cdot  s+z\cdot x}\Gamma_{h(x)}/2^{M}
\end{equation}
with $h(x)\equiv h(0^M,x)$ being the number of bits $1$ in $x$ (Hamming weight or bit sum).
The inner summation
\begin{equation}
\sum_s(-1)^{y\cdot s}=\left\{\begin{array}{l}
2^M\mbox{ for }y=0^M\\
0\mbox{ otherwise}
\end{array}\right.
\end{equation}
because whenever any bit of $y$ is equal to $1$, the sum contains the equal amounts of $+1$ and $-1$.
The only nondestructive case is for all bits equal to $0$. Therefore $y\neq z$ gives $\Gamma'_{yz}=0$.
The nonzero terms are
\begin{equation}
\Gamma'_{yy}=\sum_{x}(-1)^{y\cdot x}\Gamma_{h(x)}.
\end{equation}
The permutation symmetry gives $\Gamma'_{yy}$ also depending on the 
number of nonzero bits of  $s$ i.e. $w=h(s)$.
We can identify $\Gamma'_{yy}\equiv \Gamma'_{w}$ and get
\begin{equation}
\Gamma'_w=\sum_{k=0}^M\sum_j(-1)^j\binom{w}{j}\binom{M-w}{k-j}\Gamma_k
\end{equation}
with the binomial
\begin{equation}
\binom{w}{j}=\frac{ w!}{j!(w-j)!}=\binom{w}{w-j}
\end{equation}
being the number of $j$-element subsets of a $w$-element set.
Indeed, summing over $x$ we have to choose $j$ position of bits of $x$ with value $1$ on the positions where $y$ has bits $1$
and choose $k-j$ positions of the remaining bits of $x$ in the remaining positions of $y$ (where it has bits $0$).
The coefficient is known as Kravchuk polynomial or matrix \cite{krav1,krav2,krav3}, i.e.
\begin{equation}
\Gamma'_w=\sum_k \Gamma_k K^M_{kw}
\end{equation}
with 
\begin{align}
&K^M_{kw}=K_k(w,M)=K^M_k(w)=\sum_j(-1)^j\binom{w}{j}\binom{M-w}{k-j}\nonumber\\
&=
K^M_{kw}=\sum_{h(x)=k}(-1)^{x\cdot y}
\end{align}
at $y$ such that $h(y)=w$ summing over bitstrings $x$ with Hamming weight $k$.
being the element of $(M+1)\times (M+1)$ Kravchuk matrix indexed $w,k=0,1,\dots, M$.
Traditionally, we shall omit $M$ whenever unambiguous.
The summation is over integer $j$  whenever binomials are nonzero, i.e.  $0\leq j\leq w$, $0\leq k-j\leq M-w$.

\begin{table*}
$$
\begin{pmatrix}
1
\end{pmatrix},\begin{pmatrix}
1&1\\
1&-1
\end{pmatrix},
\begin{pmatrix}
1&1&1\\
2&0&-2\\
1&-1&1
\end{pmatrix},
\begin{pmatrix}
1&1&1&1\\
3&1&-1&-3\\
3&-1&-1&3\\
1&-1&1&-1
\end{pmatrix},
\begin{pmatrix}
1&1&1&1&1\\
4&2&0&-2&-4\\
6&0&-2&0&6\\
4&-2&0&2&-4\\
1&-1&1&-1&1
\end{pmatrix}
$$
\caption{The first Kravchuk Matrices for $M=0,1,2,3,4$}
\end{table*}

Let us recall properties of Kravchuk polynomials \cite{krav1} with proofs.
The generating function reads
\begin{equation}
\sum_k K_k(w)t^k=(1-t)^w(1+t)^{M-w}
\end{equation}
Weighted symmetry
\begin{equation}
\binom{M}{w}K_k(w)=\binom{M}{k}K_w(k)
\end{equation}
since, expanding the defining sum, each term reads, omitting $(-1)^j$,
\begin{equation}
\frac{M!}{j!(w-j)!(k-j)!(M-w-k+j)!}
\end{equation}
which is symmetric under $w\leftrightarrow k$.
In particular, by the properties of binomials
$K_0(w)=1$,  $K_{M-k}(w)=(-1)^w K_k(w)$ since
\begin{align}
&K_{M-k}(w)=\sum_j(-1)^j\binom{w}{j}\binom{M-w}{M-k-j}=\nonumber\\
&\sum_j(-1)^{w-j}\binom{w}{j}\binom{M-w}{M-k-w+j}
=\nonumber\\
&(-1)^w\sum_j(-1)^{j}\binom{w}{j}\binom{M-w}{k-j}
\end{align}
where we have replaced $j$ by $w-j$ and used the symmetry of binomials.
This also implies $K_{M-k}(w)=(-1)^{k+w}K_k(M-w)$ by
\begin{align}
&\binom{M}{w}(-1)^w K_{M-k}(w)=\binom{M}{w} K_k(w)=\nonumber\\
&\binom{M}{k}K_w(k)=(-1)^k\binom{M}{k}K_{M-w}(k)
=\nonumber\\
&(-1)^k\binom{M}{M-w}K_k(M-w)
\end{align}
or $K_k(M-w)=(-1)^kK_k(w)$.
The most important property is orthogonality
\begin{equation}
\sum_w \binom{M}{w}K_k(w)K_m(w)=2^M\binom{M}{k}\delta_{mk}
\end{equation}
To prove it we use generating functions
\begin{widetext}
\begin{align}
&\sum_{wkm} \binom{M}{w}K_k(w)t^kK_m(w)u^m=\sum_w\binom{M}{w}(1-t)^w(1+t)^{M-w}(1-u)^w(1+u)^{M-w}\nonumber\\
&=\sum_w\binom{M}{w}[(1-t)(1-u)]^w[(1+t)(1+u)]^{M-w}=[(1-t)(1-u)+(1+t)(1+u)]^M\nonumber\\
&=2^M(1+tu)^M=2^M\sum_k\binom{M}{k}t^k u^k
\end{align}
\end{widetext}
Alternatively we can use directly the WH representation of Kravchuk polynomials.
Since the number of bitstrings $y$ of Hamming weight $w$ is exactly $\binom{M}{w}$ we evaluate
\begin{equation}
\sum_{yxz}(-1)^{x\cdot y+z\cdot y}=\sum_{yxz}(-1)^{y\cdot(x\oplus z)}
\end{equation}
over bitstrings $y$, $x$, $z$ with $h(x)=k$, $h(z)=m$.
Again, the sum is $0$ whenever $x\neq z$ which leads to vanishing cases $k\neq m$.
Still when $k=m$ we need $x=z$ so the actual number of choices is binomial $M$ over $k$
and the dummy sum over $y$ gives $2^M$.

Combining with the previous property, we have
\begin{align}
&\sum_w K_k(w)K_w(m)=\sum_w \binom{M}{w}K_k(w)K_m(w){\binom{M}{m}}^{-1}\nonumber\\
&=2^M\delta_{mk}
\end{align}
i.e. $K^2=2^M I$, so the normalized matrix $2^{-M/2}K$ is an involution.
Therefore
\begin{equation}
\Gamma_k=\sum_w \Gamma'_w K_w(k)/2^M.
\end{equation}

Having this powerful tool we can determine the coefficients $\Gamma_k$ for odd $k$ keeping the even cases at $-1$.
Notice that the matrix $\Gamma$ with zeros at even $k$ will return $\Gamma'_w$ antisymmetric under $w\to M-w$,
by the above properties of Kravchuk polynomials. Therefore we can correct $\Gamma'_w$ only by an antisymmetric diagonal.
On the other hand antisymmetric $\Gamma'$ will give $0$ contribution to $\Gamma_k$, at even $k$.
Indeed, the antisymmetry implies
\begin{align}
&2^{M+1}\Gamma_k=\sum_w (\Gamma'_w-\Gamma'_{M-w}) K_w(k)=\nonumber\\
&\sum_w\Gamma'_w(K_w(k)-K_{M-w}(k))\nonumber\\
&=\sum_w\Gamma'_w(K_w(k)-(-1)^kK_{w}(k))=\nonumber\\
&\sum_w\Gamma'_wK_w(k)(1-(-1)^k)
\end{align}
which is zero for even $k$.
Therefore we can make the matrix $\Gamma$ PSD adding correction
\begin{equation}
\Gamma''_k=\left\{\begin{array}{l}
0\mbox{ for } k=0,\\
-2^{M-1}\mbox{ for }k=1,\\
\mbox{ between } -2^{M-1}\mbox{ and }2^{M-1}\mbox{ otherwise}
\end{array}\right.
\end{equation}
In this  way the diagonal elements at $k=0,M$ will become $0$, $\Gamma'_1\to 0$, $\Gamma'_{M-1}\to 2^M$.
The latter allows to match the offdiagonal entries with their $-2^{M}$ weight. The rest of terms are arbitrary
(can be set to 0), but remember to keep antisymmetry.
The construction is unique for $N=4$ (all $\Gamma''_k$ specified) and $N=5$ (by antisymmetry $\Gamma''_2=0$).
For larger values we get more freedom.

\section{Lov\'asz theta, its dual certificate, and higher moments}
\label{app:theta-duality}

Throughout this appendix, $\mathcal G=(V,E)$ is the anticommutation graph
of $L\geq1$ observables, not the hypergraph of physical parties.
An edge $\{i,j\}$ means $F_iF_j=-F_jF_i$ (a nonedge means $F_iF_j=F_jF_i$).
We write $J=\boldsymbol{1}\boldsymbol{1}^{T}$, where $\boldsymbol 1$ is a column vector with all entries equal to $1$, and
$\langle A,B\rangle_{\mathrm{tr}}=\operatorname{Tr}(A^{T}B)$ for real
symmetric matrices. Matrix positivity always means positive
semidefiniteness, not entrywise nonnegativity.
The equivalences below are standard; see Lov\'asz~\cite{lovasz1979} and
Knuth~\cite{knuth1994}. We include the derivation to fix the graph-complement and sign conventions.

\subsection{Trace-normalized primal and the \texorpdfstring{$\Gamma$}{Gamma} dual}

A standard primal definition is
\begin{align}
 \vartheta(\mathcal G)=\max_{X=X^T}\quad&\langle J,X\rangle_{\mathrm{tr}}
 \label{eq:theta-primal}\\
 \text{subject to}\quad&X\succeq0,\quad\operatorname{Tr}X=1,\nonumber\\
 &X_{ij}=0\quad\text{for }\{i,j\}\in E.\nonumber
\end{align}
The constraints involve edges of $\mathcal G$, not edges of its
complement. For example, this convention gives
$\vartheta(K_L)=1$ and $\vartheta(\overline K_L)=L$.

For $i<j$, set
$B^{ij}=(e_ie_j^T+e_je_i^T)/2$, so that
$\langle B^{ij},X\rangle_{\mathrm{tr}}=X_{ij}$.
Introduce a real multiplier $\Lambda$ for the trace constraint and free
real multipliers $y_{ij}$ for the edge equations. Set
$Y=\sum_{\{i,j\}\in E}y_{ij}B^{ij}$.
Let us consider
\begin{align}
 \mathcal L(X,\Lambda,Y)
 &=\langle J,X\rangle_{\mathrm{tr}}
   +\Lambda(1-\operatorname{Tr}X)-\langle Y,X\rangle_{\mathrm{tr}}
   \nonumber\\
 &=\Lambda-\langle\Lambda I+Y-J,X\rangle_{\mathrm{tr}}.
 \label{eq:theta-lagrangian}
\end{align}
Its supremum over $X\succeq0$ is $\Lambda$ precisely when
$\Lambda I+Y-J\succeq0$. Otherwise, scaling the projection onto a negative-eigenvalue eigenvector makes the supremum unbounded.
Thus the dual is
\begin{align}
 \min_{\Lambda,Y=Y^T}\quad&\Lambda\label{eq:theta-dual-Y}\\
 \text{subject to}\quad&\Lambda I+Y-J\succeq0,\nonumber\\
 &Y_{ii}=0,\quad
 Y_{ij}=0\quad(i\ne j,\ \{i,j\}\notin E).\nonumber
\end{align}
The primal is strictly feasible at $X=I/L$, and the dual is strictly feasible for $Y=0$ and $\Lambda>L$. Thus strong duality applies. Writing $\Gamma=\Lambda I+Y-J$ gives exactly
\begin{align}
 \vartheta(\mathcal G)=\min_{\Lambda,\Gamma=\Gamma^T}\quad&\Lambda
 \label{eq:theta-dual-Gamma}\\
 \text{subject to}\quad&\Gamma\succeq0,\quad
 \Gamma_{ii}=\Lambda-1,\nonumber\\
 &\Gamma_{ij}=-1\quad(i\ne j,\ \{i,j\}\notin E).\nonumber
\end{align}
Conversely, any matrix with this pattern gives a feasible
$Y=\Gamma+J-\Lambda I$ in~\eqref{eq:theta-dual-Y}.
In particular, the edge entries of $\Gamma$ are free; there is no
additional factor of two in its fixed nonedge entries.

Let $\beta(\mathcal G)$ denote the supremum of this sum over states and
Hermitian-involution representations with both the specified edge
anticommutations and nonedge commutations, as in \cite{xu2024beta}. Then
\begin{equation}
 \alpha(\mathcal G)\leq\beta(\mathcal G)\leq\vartheta(\mathcal G)
 \leq\overline\chi_f(\mathcal G)\leq\overline\chi(\mathcal G).
 \label{eq:graph-bound-chain}
\end{equation}
Here $\alpha$ is the independence number, while $\overline\chi_f$ and
$\overline\chi$ are the fractional and integer anticommuting clique-cover numbers (chromatic numbers of the  complementary graph).
Independence number is the maximal size of a commuting clique.
Thus an explicit $\Gamma$ certificate can improve on all fractional
combinations of the original-family clique inequalities. It remains
useful even when the general theta theorem is available: theorem
does not itself evaluate $\vartheta$ for each network graph.
The equality $\beta=\alpha$ is not assumed in general.

\subsection{Equivalence with the augmented moment formulation}

The first-level formulation used in \cite{bermejo2024},
\begin{align}
 \max_{z,K=K^T}\quad&\sum_i z_i\label{eq:theta-body}\\
 \text{subject to}\quad&
 M=\begin{pmatrix}1&z^T\\z&K\end{pmatrix}\succeq0,\nonumber\\
 &K_{ii}=z_i,\quad K_{ij}=0\quad(\{i,j\}\in E).\nonumber
\end{align}
We give explicit comparisons with~\eqref{eq:theta-primal}.
For a feasible $(z,K)$ with $t=\sum_i z_i>0$, the matrix $X=K/t$
is feasible in~\eqref{eq:theta-primal}. Since $K-zz^T\succeq0$,
\begin{equation}
 \langle J,X\rangle_{\mathrm{tr}}
 =\frac{\boldsymbol{1}^TK\boldsymbol{1}}{t}\geq t.
 \label{eq:theta-body-to-trace}
\end{equation}
The case $t=0$ is harmless.

Conversely, let $X_{ij}=v_i^Tv_j$ be a Gram decomposition of a feasible
$X$ in~\eqref{eq:theta-primal}, and set
$s=\boldsymbol{1}^TX\boldsymbol{1}$.
For $s>0$, define
\begin{equation}
 h=\frac{\sum_i v_i}{\sqrt{s}},\qquad
 w_i=\frac{h^Tv_i}{\|v_i\|^2}v_i
 \quad(v_i\ne0),
 \label{eq:theta-trace-to-body}
\end{equation}
and $w_i=0$ otherwise. The Gram matrix of $h,w_1,\ldots,w_L$ is feasible
in~\eqref{eq:theta-body}, because $\|h\|=1$,
$h^Tw_i=\|w_i\|^2$, and edge orthogonality is preserved.
Its objective obeys Cauchy--Schwarz:
\begin{equation}
 \sum_i\|w_i\|^2
 =\sum_{v_i\ne0}\frac{(h^Tv_i)^2}{\|v_i\|^2}
 \geq\frac{(\sum_i h^Tv_i)^2}{\sum_i\|v_i\|^2}=s.
 \label{eq:theta-trace-to-body-value}
\end{equation}
Here $\sum_i\|v_i\|^2=\operatorname{Tr}X=1$.
Again $s=0$ cannot obstruct equality of the maxima. Together these
constructions prove that~\eqref{eq:theta-body} has value
$\vartheta(\mathcal G)$.

There is also a particularly simple conversion of dual certificates.
For a feasible $\Gamma$ with $\Lambda>0$, define
\begin{equation}
 Z_1=\begin{pmatrix}
 \Lambda&-\boldsymbol{1}^T\\
 -\boldsymbol{1}&(\Gamma+J)/\Lambda
 \end{pmatrix}\succeq0.
 \label{eq:theta-lifted-dual}
\end{equation}
Its Schur complement is $\Gamma/\Lambda$.
For every feasible $M$ in~\eqref{eq:theta-body},
\begin{equation}
 \operatorname{Tr}(Z_1M)=\Lambda-\sum_i z_i\geq0.
 \label{eq:theta-lifted-gap}
\end{equation}
The bottom diagonal contributes $\sum_i z_i$, the first row and column
contribute $-2\sum_i z_i$, and the remaining terms vanish on edges or
have zero coefficients on nonedges. Thus the $L\times L$ matrix
$\Gamma$ and the $(L+1)\times(L+1)$ moment-dual slack are different
matrices certifying the same bound.

\section{Bell trade-off with 3 parties and 3 settings}
\label{appb3}

Here we present the proof of inequality (\ref{ss2})
For Hermitian operators $Y\succeq X$ means that $Y-X=Z\succeq 0$ and $Z\succeq 0$ means semidefinite operator with nonnegative eigenvalues
or equivalently $\langle\psi|Z|\psi\rangle\geq 0$ for every state $|\psi\rangle$.
We can write
\begin{equation}
\mathcal S_{AB}+\mathcal S_{CA}=A_1(B_++C_+)+A_2B_-+A_3C_-
\end{equation}
where $B_\pm =B_1\pm B_3$, $C_\pm=C_1\pm C_2$
We have 
\begin{align}
&A_1(B_++C_+)\leq |A_1(B_++C_+)|\nonumber\\
&=|A_1||B_++C_+|=|B_++C_+|
\end{align}
where $|F|=\sqrt{F^2}$ for Hermitian $F$ means replacing eigenvalues by their modulus in the diagonal basis.
Therefore
\begin{equation}
\mathcal S_{AB}+\mathcal S_{CA}\leq  |B_++C_+|+A_2B_-+A_3C_-
\end{equation}
We can divide all the space into eigenspaces of
\begin{equation}
\begin{aligned}
B_+^2&=4\cos^2\beta,& B_-^2&=4\sin^2\beta,\\
C_+^2&=4\cos^2\gamma,& C_-^2&=4\sin^2\gamma.
\end{aligned}
\end{equation}
with $\beta,\gamma\in[0,\pi/2]$
and restrict to a single pair $\pm \beta$, $\pm\gamma$ by linearity, just like in Jordan's lemma Appendix \ref{appjor}.
Then again $B_+=2Z\cos\beta$, $B_-=2X\sin\beta$, $C_+=2Z\cos\gamma$, $C_-=2X\sin\gamma$.

Without loss of generality we assume $\cos\beta\geq \cos\gamma\geq 0$.
If $\beta=\pi/2$ then also $\gamma=\pi/2$. Then $B_1=-B_3=\pm 1$ and $C_1=-C_2=\pm 1$ and the case reduces to
\begin{equation}
2A_3C_1+2A_2B_1\leq 4
\end{equation}
If $\beta=0$ then we get
\begin{align}
&2\pm C_++A_3C_-=2+C_1(A_3\pm 1)-C_2(A_3\mp 1)=\nonumber\\
&2\pm 2C_2P(A_3=\mp 1)\pm 2C_1P(A_3=\pm 1)\nonumber\\
&\leq 2+2P(A_3=\mp 1)+2P(A_3=\pm 1)=2+2=4
\end{align}
with $P(A_3=\pm 1)$ meaning the projection onto the eigenspace of $A_3=\pm 1$.

For $0<\beta<\pi/2$,
\begin{equation}
|B_++C_+|=2II\cos\beta +2ZZ\cos\gamma .
\end{equation}
Now
\begin{align}
&|B_++C_+|+A_2B_-+A_3C_-=\nonumber\\
&2II\cos\beta+2ZZ\cos\gamma+2A_2XI\sin\beta+2A_3IX\sin\gamma
\end{align}
We use
\begin{equation}
II\cos\beta+A_2XI\sin\beta\leq \cos\beta+\sin\beta\leq \sqrt{2}
\end{equation}
with the maximum for $\beta=\pi/4$.
On the other hand taking
\begin{equation}
W=ZZ\cos\gamma+A_3IX\sin\gamma
\end{equation}
we have 
\begin{equation}
\langle W\rangle^2\leq \langle W^2\rangle=\langle \cos^2\gamma+\sin^2\gamma\rangle=1
\end{equation}
because matrices $ZZ$ and $IX$ anticommute and square to 1, like in monogamy argumentation.
$\square$.

\begin{table}
\begin{tabular}{|c|c|c|c|}
\toprule
&$A$&$B$&$C$\\
\midrule
$A_2$&$-A$&$B$&$\sqrt{2}B-A$\\
$A_3$&$-A$&$\sqrt{2}C-A$&$C$\\
$B_2$&$\sqrt{2}C-B$&$-B$&$C$\\
$B_3$&$A$&$-B$&$\sqrt{2}A-B$\\
$C_2$&$A$&$\sqrt{2}A-C$&$-C$\\
$C_3$&$\sqrt{2}B-C$&$B$&$-C$\\
\bottomrule
\end{tabular}
\caption{Reduction of operators acting on the {\bf maximizing} state. The first column contains all operators, the next 3 column correspond
to the maximal states with party $A$, $B$, $C$ separated.}\label{rmax}
\end{table}

\begin{table}
\begin{tabular}{|c|c|c|c|}
\toprule
&$A$&$B$&$C$\\
\midrule
$A_2$&$-A$&$-B$&$-\sqrt{2}B-A$\\
$A_3$&$-A$&$-\sqrt{2}C-A$&$-C$\\
$B_2$&$-\sqrt{2}C-B$&$-B$&$-C$\\
$B_3$&$-A$&$-B$&$-\sqrt{2}A-B$\\
$C_2$&$-A$&$-\sqrt{2}A-C$&$-C$\\
$C_3$&$-\sqrt{2}B-C$&$-B$&$-C$\\
\bottomrule
\end{tabular}
\caption{Reduction of operators acting on the {\bf minimizing} state. The first column contains all operators, the next 3 column correspond
to the maximal states with party $A$, $B$, $C$ separated.}\label{rmin}
\end{table}

To find the exact lower and upper bound on (\ref{ss3}) using NPA, we construct two sets of operator polynomials (Gr{\"o}bner bases).
The first set consists of $II=1$, $A_iB_j$, $B_iC_j$, $C_iA_j$ (28 elements in total), and we make a quadratic form in this basis
with the coefficients $g$, $g_{ViWj}$, $g_{ViWjV'i'W'j'}$,
\begin{align}
&G=\sum_{ViWj,V'i'W'j'}g_{ViWjV'i'W'j'}W'_{j'}V'_{i'}V_iW_j\nonumber\\
&+\sum_{ViWj}2g_{ViWj}V_iW_j+g
\end{align}
with $V,W=A,B,C$.
The second set constrains triples $A_iB_jB_k$, $B_iA_jA_k$, $B_iC_jC_k$, $C_iB_jB_k$, $C_iA_jA_k$, $A_iC_jC_k$, $j\neq k$,
(108 elements in total), and we make a quadratic form in this basis
with the coefficients $h_{ViWjkV'i'W'j'k'}$
\begin{equation}
H=\sum_{ViWjk,V'i'W'j'k'}h_{ViWjkV'i'W'j'k'}W'_{k'}W'_{j'}V'_{i'}V_iW_jW_k
\end{equation}
such that $G\geq 0$, $H\geq 0$ and
\begin{equation}
G+H=2(2+\sqrt{2})\mp(\mathcal S_A+\mathcal S_B+\mathcal S_C)
\end{equation}
with $-$ for the maximum and $+$ for the minimum.

Now, using squares $A_i^2=B_i^2=C_j^2=1$ the one can find a set of coefficients $g$ and $h$ to reach the above identification
by SDP. Unfortunately, simply numerical run cannot give an exact value.
To obtain exact agreement we need a complicated workaround.
Firstly, knowing the reductions on the maximizing  and minimizing states in Tables \ref{rmax} and \ref{rmin}, we determine the kernel of $g$ and $h$ (matrices of coefficients),
as both $G$ and $H$ must vanish acting on maximizing states. In particular
\begin{widetext}
\begin{align}
&G(|A\rangle,|B\rangle,|C\rangle)=gr_g((1,BC,CA,AB)|A\rangle,(1,BC,CA,AB)|B\rangle,(1,BC,CA,AB)|C\rangle)=0
\end{align}
where $r_g$ is a  $28\times 12$ reduction matrix as given in Tables \ref{redg} and \ref{redgm} (with omitted identity).
Due to linear independence we expect $gr_g=0$ which allows to reduce the basis of $g$ to the kernel of $r_g^T$.

Similarly
\begin{align}
&H(|A\rangle,|B\rangle,|C\rangle)=hr_h((D,A,B,C)|A\rangle,(D,A,B,C)|B\rangle,(D,A,B,C)|C\rangle)=0
\end{align}
\end{widetext}
with $D=ABC$ where $r_h$ is a  $108\times 12$ reduction matrix as given in Tables \ref{redh} and \ref{redhm}.
Due to linear independence we expect $hr_h=0$ which allows to reduce the basis of $h$ to the kernel of $r_h^T$.

\begin{table}
\begin{tabular}{|c|c|c|c|}
\toprule
&$A$&$B$&$C$\\
\midrule
$A_1B_1$&$AB$&$AB$&$AB$\\
$A_1B_2$&$\sqrt{2}CA-AB$&$-AB$&$CA$\\
$A_1B_3$&$1$&$-AB$&$\sqrt{2}-AB$\\
$A_2B_1$&$-AB$&$1$&$\sqrt{2}-AB$\\
$A_2B_2$&$AB-\sqrt{2}CA$&$-1$&$\sqrt{2}BC-CA$\\
$A_2B_3$&$-1$&$-1$&$-AB$\\
$A_3B_1$&$-AB$&$\sqrt{2}BC-AB$&$BC$\\
$A_3B_2$&$AB-\sqrt{2}CA$&$AB-\sqrt{2}BC$&$1$\\
$A_3B_3$&$-1$&$AB-\sqrt{2}BC$&$\sqrt{2}CA-BC$\\
\bottomrule
\end{tabular}
\caption{Reduction of products of operators in the operator-monomial basis of $G$ for {\bf maximizing} state. The other operators are obtained by a cyclic shift 
$A\to B\to C\to A$.}
\label{redg}
\end{table}

\begin{table}
\begin{tabular}{|c|c|c|c|}
\toprule
&$A$&$B$&$C$\\
\midrule
$A_1B_1$&$AB$&$AB$&$AB$\\
$A_1B_2$&$-\sqrt{2}CA-AB$&$-AB$&$-CA$\\
$A_1B_3$&$-1$&$-AB$&$-\sqrt{2}-AB$\\
$A_2B_1$&$-AB$&$-1$&$-\sqrt{2}-AB$\\
$A_2B_2$&$AB+\sqrt{2}CA$&$1$&$\sqrt{2}BC+CA$\\
$A_2B_3$&$1$&$1$&$-AB$\\
$A_3B_1$&$-AB$&$-\sqrt{2}BC-AB$&$-BC$\\
$A_3B_2$&$AB+\sqrt{2}CA$&$AB+\sqrt{2}BC$&$1$\\
$A_3B_3$&$1$&$AB+\sqrt{2}BC$&$\sqrt{2}CA+BC$\\
\bottomrule
\end{tabular}
\caption{Reduction of products of operators in the operator-monomial basis of $G$ for {\bf minimizing} state. The other operators are obtained by a cyclic shift 
$A\to B\to C\to A$.}
\label{redgm}
\end{table}


\begin{table}
\begin{tabular}{|c|c|c|c|}
\toprule
&$A$&$B$&$C$\\
\midrule
$B_2B_1A_1$&$A-\sqrt{2}D$&$-A$&$D$\\
$B_3B_1A_1$&$B$&$-A$&$A-\sqrt{2}B$\\
$B_3B_2A_1$&$\sqrt{2}C-B$&$A$&$\sqrt{2}C-D$\\
$B_1B_2A_1$&$\sqrt{2}D-A$&$-A$&$D$\\
$B_1B_3A_1$&$B$&$-A$&$\sqrt{2}B-A$\\
$B_2B_3A_1$&$\sqrt{2}C-B$&$A$&$\sqrt{2}C-D$\\
\midrule
$B_2B_1A_2$&$\sqrt{2}D-A$&$-B$&$\sqrt{2}C-D$\\
$B_3B_1A_2$&$-B$&$-B$&$A$\\
$B_3B_2A_2$&$B-\sqrt{2}C$&$B$&$-D$\\
$B_1B_2A_2$&$A-\sqrt{2}D$&$-B$&$\sqrt{2}C-D$\\
$B_1B_3A_2$&$-B$&$-B$&$-A$\\
$B_2B_3A_2$&$B-\sqrt{2}C$&$B$&$-D$\\
\midrule
$B_2B_1A_3$&$\sqrt{2}D-A$&$A-\sqrt{2}C$&$B$\\
$B_3B_1A_3$&$-B$&$A-\sqrt{2}C$&$C-\sqrt{2}D$\\
$B_3B_2A_3$&$B-\sqrt{2}C$&$\sqrt{2}C-A$&$\sqrt{2}A-B$\\
$B_1B_2A_3$&$A-\sqrt{2}D$&$A-\sqrt{2}C$&$B$\\
$B_1B_3A_3$&$-B$&$A-\sqrt{2}C$&$\sqrt{2}D-C$\\
$B_2B_3A_3$&$B-\sqrt{2}C$&$\sqrt{2}C-A$&$\sqrt{2}A-B$\\
\bottomrule
\end{tabular}
\caption{Reduction of
operators in the operator-monomial basis of $H$ for {\bf maximizing} state. The other operators are obtained by a cyclic shift $A\to B\to C\to A$.
$D=ABC$.
}\label{redh}
\end{table}

\begin{table}
\begin{tabular}{|c|c|c|c|}
\toprule
&$A$&$B$&$C$\\
\midrule
$B_2B_1A_1$&$A+\sqrt{2}D$&$-A$&$-D$\\
$B_3B_1A_1$&$-B$&$-A$&$A+\sqrt{2}B$\\
$B_3B_2A_1$&$\sqrt{2}C+B$&$A$&$\sqrt{2}C+D$\\
$B_1B_2A_1$&$-\sqrt{2}D-A$&$-A$&$-D$\\
$B_1B_3A_1$&$-B$&$-A$&$-\sqrt{2}B-A$\\
$B_2B_3A_1$&$\sqrt{2}C+B$&$A$&$\sqrt{2}C+D$\\
\midrule
$B_2B_1A_2$&$-\sqrt{2}D-A$&$B$&$\sqrt{2}C+D$\\
$B_3B_1A_2$&$B$&$B$&$A$\\
$B_3B_2A_2$&$-B-\sqrt{2}C$&$-B$&$D$\\
$B_1B_2A_2$&$A+\sqrt{2}D$&$B$&$\sqrt{2}C+D$\\
$B_1B_3A_2$&$B$&$B$&$-A$\\
$B_2B_3A_2$&$-B-\sqrt{2}C$&$-B$&$D$\\
\midrule
$B_2B_1A_3$&$-\sqrt{2}D-A$&$A+\sqrt{2}C$&$B$\\
$B_3B_1A_3$&$B$&$A+\sqrt{2}C$&$-C-\sqrt{2}D$\\
$B_3B_2A_3$&$-B-\sqrt{2}C$&$-\sqrt{2}C-A$&$-\sqrt{2}A-B$\\
$B_1B_2A_3$&$A+\sqrt{2}D$&$A+\sqrt{2}C$&$B$\\
$B_1B_3A_3$&$B$&$A+\sqrt{2}C$&$\sqrt{2}D+C$\\
$B_2B_3A_3$&$-B-\sqrt{2}C$&$-\sqrt{2}C-A$&$-\sqrt{2}A-B$\\
\bottomrule
\end{tabular}
\caption{Reduction of
operators in the operator-monomial basis of $H$ for {\bf minimizing} state. The other operators are obtained by a cyclic shift $A\to B\to C\to A$.
$D=ABC$.
}\label{redhm}
\end{table}

In addition it turns out there the remaining freedom allows reduction involving symmetry $A\leftrightarrow B$  together with setting relabeling $2\leftrightarrow 3$
and cyclic shift $A\to B\to C\to A$.  We could also discard 24 monomials with representative $A_2B_2B_3$, $A_3B_2B_3$, $A_2B_3B_2$, $A_3B_3B_2$ up to symmetries and from  the basis.
 Note that these reductions will be nevertheless confirmed a posteriori by the successful construction of semidefinite matrices $g$ and $h$.
Another obstruction is the appearance of $\sqrt{2}$. We resolve it using quadratic field $\mathbb Q(\sqrt{2})$.
We recall that the field is generated by rational pairs $(x_i,x_s)\equiv x=x_i+x_s\sqrt{2}$
with $0=(0,0)$, $1=(1,0)$ and all linear operations, addition, subtraction, multiplication and division, are well defined.
\begin{align}
&x\pm y=(x_i,x_s)\pm(y_i,y_s)=(x_i\pm y_i,x_s\pm y_s)\nonumber\\
&x\cdot y=(x_i,x_s)\cdot(y_i,y_s)=(x_iy_i+2x_sy_s,x_iy_s+x_sy_i)\nonumber\\
&x/y=\frac{(x_i,x_s)}{(y_i,y_s)}=\left(\frac{x_iy_i-2x_sy_s}{y_i^2-2y^2_s},\frac{x_sy_i-y_sx_i}{y_i^2-2y^2_s}\right)
\end{align}
The denominator is zero only if $y_i=y_s=0$ since $x^2=2$ has no solution in rational numbers. It is also useful to define conjugate $x'=x_i-x_s\sqrt{2}=(x_i,-x_s)$
The rest of the task to find the numerical candidates for $g$ and $h$ and replace the independent irreducible  entries
by elements of $\mathbb Q(\sqrt{2})$. The positivity must be verified symbolically which requires operations on giant integers (of $\sim 2000$ digits).

We can also use matrix operations. Let $A=A_i+A_s\sqrt{2}$ where $A_{i/s}$ are rational.

Then
\begin{equation}
\begin{pmatrix}
I/\sqrt{2}&I/\sqrt{2}\\
I&-I\end{pmatrix}
\begin{pmatrix}
A&0\\
0&A'\end{pmatrix}
\begin{pmatrix} 
I/\sqrt{2}&I/2\\
I/\sqrt{2}&-I/2\end{pmatrix}=\begin{pmatrix}
A_i&A_s\\
2A_s&A_i\end{pmatrix}
\end{equation}
where $I$ is the identity matrix of the size matching the number of rows/columns of $A$.
Then for square matrices
\begin{align}
&\begin{pmatrix}
A_i&A_s\\
2A_s&A_i\end{pmatrix}^{-1}=\nonumber\\
&\begin{pmatrix}
I/\sqrt{2}&I/\sqrt{2}\\
I&-I\end{pmatrix}
\begin{pmatrix}
A^{-1}&0\\
0&A^{\prime -1}\end{pmatrix}
\begin{pmatrix} 
I/\sqrt{2}&I/2\\
I/\sqrt{2}&-I/2\end{pmatrix}
\end{align}
We have also $(A')^{-1}=(A^{-1})'$, i.e. the conjugation can be applied after/before inversion.
Now, to find the kernel of $r^T_g$, and $r^T_h$ we can complete the matrices to the nondegenerate square matrices,
invert and extract bases of the kernels. For the 28-dimensional $G$ basis, $k_g$ has rank 16; for the 108-dimensional $H$ basis, $k_h$ has rank 96 (rank 72 after removing the 24 monomials). In each case $r^T_g k_g=0$ and $r^T_h k_h=0$.

Numerically, one can reduce both numerically and symbolically to the problem to 214 free parameters, to be either determined directly from
the SDP dual solution or running SDP with them as the input variables. One must maintain sufficient accuracy (safely 9 digits) to 
prevent small deviations from spoiling the positivity. The algebraic bottleneck is the symbolic kernel of a large matrix of all constraints.
Fortunately SageMath \cite{sage} deals with it within minutes.

\section{\texorpdfstring{$\mathcal I_{3322}$ bound and triple monogamy}{I3322 bound and triple monogamy}}
\label{appis}
  \begin{table}
 \begin{tabular}{cc}
 \toprule
 $i$&$m_i$\\
0& 8257536\\
1& - 12124160\\
2& - 11075584\\
3& - 24936448\\
4& - 25357824\\
5& - 6137088\\
6& + 922336\\
7& + 284048\\
8& - 15696\\
9& - 4696\\
10& + 118\\
11& + 29\\
\bottomrule
\end{tabular}
\caption{Table of coefficients of the polynomial $m(x)$}\label{i33}
\end{table}

Define
\[
\mathcal J^{xyz}=12+\mathcal I^{xyz}(A,B)+\mathcal I^{xyz}(B,C)+\mathcal I^{xyz}(C,A).
\]
Below we express $4\mathcal J^{xyz}$ as sums of positive squares (products
of commuting squares in the factored terms).
\begin{align}
&4\mathcal J^{+++}=\nonumber\\
&(C_3-A_2)^2(C_3-B_2)^2+(C_3+A_1)^2(C_3+B_1)^2+\nonumber\\
&(A_3-B_2)^2(A_3-C_2)^2+(A_3+B_1)^2(A_3+C_1)^2+\nonumber\\
&(B_3-C_2)^2(B_3-A_2)^2+(B_3+C_1)^2(B_3+A_1)^2+\nonumber\\
&(1+A_1)^2(1+B_2)^2+(1+A_2)^2(1+B_1)^2+\nonumber\\
&(1+B_1)^2(1+C_2)^2+(1+B_2)^2(1+C_1)^2+\nonumber\\
&(1+C_1)^2(1+A_2)^2+(1+C_2)^2(1+A_1)^2,
\end{align}

\begin{align}
&4\mathcal J^{++-}=\nonumber\\
&(1+A_1)^2(1+B_1)^2+(1+B_1)^2(1+C_1)^2+\nonumber\\
&(1+C_1)^2(1+A_1)^2+(1+A_2)^2(1+B_2)^2+\nonumber\\
&(1+B_2)^2(1+C_2)^2+(1+C_2)^2(1+A_2)^2+\nonumber\\
&(A_3+B_1)^2(A_3+C_2)^2+(A_3-B_2)^2(A_3-C_1)^2+\nonumber\\
&(B_3+C_1)^2(B_3+A_2)^2+(B_3-C_2)^2(B_3-A_1)^2+\nonumber\\
&(C_3+A_1)^2(C_3+B_2)^2+(C_3-A_2)^2(C_3-B_1)^2,
\end{align}

\begin{align}
&4\mathcal J^{+-+}=\nonumber\\
&(1+A_1)^2(1-B_2)^2+(1+B_1)^2(1+A_2)^2+\nonumber\\
&(1+B_1)^2(1-C_2)^2+(1+C_1)^2(1+B_2)^2+\nonumber\\
&(1+C_1)^2(1-A_2)^2+(1+A_1)^2(1+C_2)^2+\nonumber\\
&(A_3+B_1)^2(A_3+C_1)^2+(A_3+B_2)^2(A_3-C_2)^2+\nonumber\\
&(B_3+C_1)^2(B_3+A_1)^2+(B_3+C_2)^2(B_3-A_2)^2+\nonumber\\
&(C_3+A_1)^2(C_3+B_1)^2+(C_3+A_2)^2(C_3-B_2)^2=\nonumber\\
&(1+A_2)^2(1-B_2)^2+(1+B_1)^2(1+A_1)^2+\nonumber\\
&(1+B_2)^2(1-C_2)^2+(1+C_1)^2(1+B_1)^2+\nonumber\\
&(1+C_2)^2(1-A_2)^2+(1+A_1)^2(1+C_1)^2+\nonumber\\
&(A_3+B_2)^2(A_3+C_1)^2+(A_3+B_1)^2(A_3-C_2)^2+\nonumber\\
&(B_3+C_2)^2(B_3+A_1)^2+(B_3+C_1)^2(B_3-A_2)^2+\nonumber\\
&(C_3+A_2)^2(C_3+B_1)^2+(C_3+A_1)^2(C_3-B_2)^2=\nonumber\\
&(1+A_1)^2(1+B_1)^2+(1+A_1)^2(1+C_1)^2+\nonumber\\
&(1+B_1)^2(1+C_1)^2+(C_2-A_2)^2(C_2-B_1)^2+\nonumber\\
&(B_2-A_1)^2(B_2-C_2)^2+(A_2-B_2)^2(A_2-C_1)^2+\nonumber\\
&(2+A_1B_3-A_2B_3+C_3A_1+C_3A_2)^2+\nonumber\\
&(2+C_1A_3-C_2A_3+B_3C_1+B_3C_2)^2+\nonumber\\
&(2+B_1C_3-B_2C_3+A_3B_1+A_3B_2)^2,
\end{align}

\begin{align}
&4\mathcal J^{+--}=\nonumber\\
&(1+A_1)^2(1+B_1)^2+(1+A_1)^2(1+C_1)^2+\nonumber\\
&(1+B_1)^2(1+C_1)^2+(C_2-A_2)^2(C_2-B_1)^2+\nonumber\\
&(B_2-A_1)^2(B_2-C_2)^2+(A_2-B_2)^2(A_2-C_1)^2+\nonumber\\
&(2-A_1B_3+A_2B_3+A_1C_3+A_2C_3)^2+\nonumber\\
&(2-C_1A_3+C_2A_3+C_1B_3+C_2B_3)^2+\nonumber\\
&(2-B_1C_3+B_2C_3+B_1A_3+B_2A_3)^2.
\end{align}
Cases $\mathcal J^{-+\ast}$ reduce to $\mathcal J^{+-\ast}$ by exchange $X_1\leftrightarrow X_2$ and $X_3\leftrightarrow -X_3$
for $X=A,B,C$,

\begin{align}
&4\mathcal J^{--+}=\nonumber\\
&(1+A_1)^2(1-B_2)^2+(1+A_2)^2(1-B_1)^2+\nonumber\\
&(1+B_1)^2(1-C_2)^2+(1+B_2)^2(1-C_1)^2+\nonumber\\
&(1+C_1)^2(1-A_2)^2+(1+C_2)^2(1-A_1)^2+\nonumber\\
&(A_3-B_1)^2(A_3+C_1)^2+(A_3+B_2)^2(A_3-C_2)^2+\nonumber\\
&(B_3-C_1)^2(B_3+A_1)^2+(B_3+C_2)^2(B_3-A_2)^2+\nonumber\\
&(C_3-A_1)^2(C_3+B_1)^2+(C_3+A_2)^2(C_3-B_2)^2,
\end{align}

\begin{align}
&4\mathcal J^{---}=\nonumber\\
&(1+A_1)^2(1-B_1)^2+(1+A_2)^2(1-B_2)^2+\nonumber\\
&(1+B_1)^2(1-C_1)^2+(1+B_2)^2(1-C_2)^2+\nonumber\\
&(1+C_1)^2(1-A_1)^2+(1+C_2)^2(1-A_2)^2+\nonumber\\
&(A_3-B_1)^2(A_3+C_2)^2+(A_3+B_2)^2(A_3-C_1)^2+\nonumber\\
&(B_3-C_1)^2(B_3+A_2)^2+(B_3+C_2)^2(B_3-A_1)^2+\nonumber\\
&(C_3-A_1)^2(C_3+B_2)^2+(C_3+A_2)^2(C_3-B_1)^2.
\end{align}

\begin{table*}
\begin{tabular}{cccc}
\toprule
$+++$&$++-$&$-++$&$-+-$\\ 
\midrule
$B_++C_+$&$B_-+C_-$&$B_+-C_+$&$B_--C_-$\\ 
$B_1C_1+B_2C_2$&$B_1C_1-B_2C_2$&$A_+(B_++C_+)$&$B_2C_1-B_1C_2$\\ 
$B_1C_2+B_2C_1$&$A_+(B_--C_-)$&$A_+(B_1C_2+B_2C_1)$&$A_+(B_-+C_-)$\\ 
$I$&$A_+(B_2C_1-B_1C_2)$&$A_+$&$A_-(B_-C_3-C_-B_3)$\\ 
$A_+(B_+-C_+)$&$A_-(B_-C_3+C_-B_3)$&$A_-(B_3-C_3)$&$A_3(B_+-C_+)$\\ 
$A_-(B_3+C_3)$&$A_3(B_++C_+)$&$A_-(B_+C_3-C_+B_3)$&\\ 
$A_-(B_+C_3+C_+B_3)$&$A_3(B_1C_2+B_2C_1)$&$A_3(B_--C_-)$&\\ 
$A_3(B_-+C_-)$&&$A_3(B_2C_1-B_1C_2)$&\\ 
\bottomrule
\end{tabular}
\caption{Distribution of monomials into buckets with respect to $Z_2$ symmetries}\label{ias}
\end{table*}

The general idea to prove the symbolic bound on the joint antisymmetric inequality (\ref{issa})  is similar to the bound on $\bar{\mathcal S}_A+\bar{\mathcal S}_B+\bar{\mathcal S}_C$.
By swapping $C_1\leftrightarrow C_2$, $C_3\to -C_3$, we rewrite the expression
\begin{align}
&B_1+B_2+C_1+C_2+(A_1+A_2)(B_1+B_2-C_1-C_2)\nonumber\\
&+A_3(B_1-B_2+C_1-C_2)+(A_1-A_2)(B_3+C_3).
\end{align}
which can be also written as
\begin{equation}
B_++C_++A_+(B_+-C_+)+A_3(B_-+C_-)+A_-(B_3+C_3)
\end{equation}
denoting $D_\pm= D_1\pm D_2$.
Note the symmetries 
\begin{enumerate}
\item swap $B$ with $C$ and $A_1$ with $-A_2$ ($A_+\to -A_+$), 
\item swap $B_3,C_3\to -B_3,-C_3$ and $A_1\leftrightarrow A_2$ (or $A_-\to -A_-$),
\item swap $B_1,C_1\leftrightarrow B_2,C_2$ and $A_3\to -A_3$.
\end{enumerate}
All terms $A_iB_jC_k$ can be then reduced to the basis $I,A_1,A_3,A_1A_3$, with the last one of antisymmetric nature $A_3A_1=-A_1A_3$,
using the fact that $A_i^2=B_j^2=C_k^2=1$.
For our representative solution we can take out $B_j$ and $A_2$, $C_3$ as numbers and the remaining terms involve $A_1$, $A_3$, $C_1$, $C_2$.
For instance $\sqrt{2}A_1C_1=1-A_1A_3$, and 
and $A_1A_3=-A_3A_1=[C_2,C_1]/2$.

By numerical inspection we have found that 
we can reduce the full $4^3=64$ dimensions to four buckets of different $Z_2$ signatures, see Table \ref{ias}. This is still exact,
confirmed by a posteriori symbolic matrix construction.

We  had to reduce these space to the kernel of optimal solutions, which gives 4 equations, but we identified
an extra kernel equation numerically (again confirmed a posteriori by the exact solution)
in the bucket $-++$,
\begin{equation}
(-1,\sqrt{2}+1, -1-\sqrt{2}, 0, \sqrt{2}, \sqrt{2}+1, 0, 0)
\end{equation}

The following explicit strategies saturate (\ref{ich}); the displayed relations hold when the operators act on their minimizing states.
Entangled $AC$ in Bell state $(|01\rangle-|10\rangle)/\sqrt{2}$
\begin{align}
&C_+=-\sqrt{2}A_1,\;C_-=-\sqrt{2}A_3,\nonumber\\
&B_+=-2,\;B_-=0,\; B_3=A_2=1
\end{align}
(e.g. $A_1=Z$, $A_3=X$, $C_+=\sqrt{2}Z$, $C_-=\sqrt{2}X$).
The other 4 cases are for the entangled $AB$ in Bell state $(|01\rangle-|10\rangle)/\sqrt{2}$.
Case 2
\begin{align}
&B_+=-\sqrt{2}A_2,\;B_-=-\sqrt{2}A_3,\nonumber\\
&A_1=-1,\;B_3=1,\; C_+=2,\; C_-=0
\end{align}
(e.g. $A_2=Z$, $A_3=X$, $B_+=\sqrt{2}Z$, $B_-=\sqrt{2}X$).
Case 3
\begin{align}
&B_+=-\sqrt{2}A_2,\;B_-=-\sqrt{2}A_3,\nonumber\\
&A_1=-1,\;B_3=A_2,\; C_+=2,\; C_-=0
\end{align}
(e.g. $A_2=Z$, $A_3=X$, $B_+=\sqrt{2}Z$, $B_-=\sqrt{2}X$, $B_3=-Z$).
Case 4
\begin{align}
&\sqrt{2}B_3=-A_-,\;\sqrt{2}B_1=-A_+,\nonumber\\
&A_3=B_2=-1,\;C_+=0,\;C_-=2
\end{align}
(e.g. $B_1=Z$, $B_3=X$, $A_+=\sqrt{2}Z$, $A_-=\sqrt{2}X$).
Case 5 is obtained from 4 by swapping $B_1,C_1\leftrightarrow B_2,C_2$, and flipping $A_3\to -A_3$.
In case 1 we have also $A_1A_3=-A_3A_1$, in case 2 and 3 we have $A_2A_3=-A_3A_2$, in case 4 and 5 we have $A_1A_2=-A_2A_1$.

The NPA sum of squares can be split by $Z_2$ symmetry corresponding to swapping $B_1,C_1\leftrightarrow B_2,C_2$, and flipping $A_3\to -A_3$, see 
Table \ref{fch}.
Additionally the basis can be reduced to the kernel of equations for the minimal cases (they belong to the kernel).

This reduction allows to find the matrices $G_\pm\succeq 0$.

\begin{table}
\begin{tabular}{cc}
\toprule
$+$&$-$\\ 
\midrule
$I$&$A_3$\\ 
$A_1$&$A_3B_+$\\ 
$A_2$&$A_3C_+$\\ 
$B_+$&$B_-$\\ 
$A_1B_+$&$A_1B_-$\\ 
$A_2B_+$&$A_2B_-$\\ 
$C_+$&$C_-$\\ 
$A_1C_+$&$A_1C_-$\\ 
$A_2C_+$&$A_2C_-$\\ 
$A_3B_-$&$A_3B_3$\\ 
$A_3C_-$&$A_3(B_1C_2+B_2C_1)$\\ 
$B_3$&$B_1C_2-B_2C_1$\\ 
$A_1B_3$&$A_1(B_1C_2-B_2C_1)$\\ 
$A_2B_3$&$A_2(B_1C_2-B_2C_1)$\\ 
$B_1C_2+B_2C_1$&$A_3(B_1C_1+B_2C_2)$\\ 
$A_1(B_1C_2+B_2C_1)$&$B_1C_1-B_2C_2$\\ 
$A_2(B_1C_2+B_2C_1)$&$A_1(B_1C_1-B_2C_2)$\\ 
$A_3(B_1C_2-B_2C_1)$&$A_2(B_1C_1-B_2C_2)$\\ 
$B_1C_1+B_2C_2$&$A_3B_3C_+$\\ 
$A_1(B_1C_1+B_2C_2)$&$B_3C_-$\\ 
$A_2(B_1C_1+B_2C_2)$&$A_1B_3C_-$\\ 
$A_3(B_1C_1-B_2C_2)$&$A_2B_3C_-$\\ 
$B_3C_+$&\\ 
$A_1B_3C_+$&\\ 
$A_2B_3C_+$&\\ 
$A_3B_3C_-$&\\ 
\bottomrule
\end{tabular}
\caption{Distribution of monomials in the hybrid inequality into buckets with respect to $Z_2$ symmetry}\label{fch}
\end{table}

\end{document}